\documentclass[11pt]{article}
\usepackage[margin = 1in]{geometry}

\usepackage{natbib}

\newcount\Comments  %

\newcommand{\kibitz}[2]{\ifnum\Comments=1{\color{#1}{#2}}\fi}

\newcommand{\hmadd}[1]{{\color{black}{#1}}}

\newcount\TODO
\newcommand{\kibitzTODO}[2]{\ifnum\TODO=1{\color{#1}{#2}}\fi}
\newcommand{\todo}[1]{\kibitzTODO{auburn}{[TODO: #1]}}

\newcount\Draft 
\usepackage{etex}
\usepackage{latexsym}
\usepackage{amssymb}
\usepackage{amsmath}
\usepackage{amsfonts}
\usepackage{bbm}
\usepackage{ifthen}
\usepackage{enumerate}
\usepackage{psfrag}
\usepackage{floatflt}
\usepackage{verbatim}
\usepackage{dsfont}
\usepackage{scalefnt}
\usepackage{epsfig,epstopdf}
\usepackage{calc}
\usepackage{graphicx}
\usepackage{subcaption}
\usepackage{multicol}
\usepackage{tikz}
\usetikzlibrary{shapes.multipart,arrows,automata}
\usetikzlibrary{positioning,arrows.meta,bending}
\usetikzlibrary{patterns,calc,intersections,through,backgrounds,decorations.pathreplacing}
\usepackage{xstring}

\usepackage{amsthm}		%
\usepackage{thmtools, thm-restate}
\usepackage{algorithmic}
\usepackage[linesnumbered,ruled,vlined]{algorithm2e}

\usepackage{color, xcolor}
\colorlet{darkblue}{blue!40!black}
\definecolor{auburn}{rgb}{0.43, 0.21, 0.1}
\definecolor{orange}{rgb}{1, 0.5, 0}
\definecolor{lightblue}{rgb}{0.1176, 0.5647, 1}

\usepackage[colorlinks,
	citecolor = darkblue,
	urlcolor = black,
	linkcolor = darkblue]{hyperref}

\theoremstyle{plain}

\theoremstyle{definition}
\newtheorem{definition}{Definition} 
\newtheorem{example}{Example}

\newtheorem{lemma}{Lemma}
\newtheorem{proposition}{Proposition}

\newcommand{\loc}{\calL}								%
\newcommand{\numLoc}{\ell}						%

\newcommand{\rate}{p}									%
\newcommand{\dist}{\delta}							%
\newcommand{\minDist}{\underline{\dist}}	%

\newcommand{\numRiders}{\mu}				%
\newcommand{\numDrivers}{\lambda}		%

\newcommand{\patience}{P}								%

\newcommand{\cost}{c}									%
\newcommand{\mechCost}{c_p}	%

\newcommand{\gb}{w}									%

\newcommand{\maxJ}{{i^\ast}}						%

\newcommand{\numSkip}{n}
\newcommand{\waitGap}{\tau}

\newcommand{\queue}{q}
\newcommand{\Queue}{Q}
\newcommand{\maxQL}{\bar{\Queue}}
\newcommand{\equQL}{\Queue^\ast}

\newcommand{\sigmast}{\sigma^\ast}

\newcommand{\alphast}{\alpha^\ast}
\newcommand{\betast}{\beta^\ast}
\newcommand{\gammast}{\gamma^\ast}

\newcommand{\util}{u}
\newcommand{\Util}{U}
\newcommand{\equtil}{\util^\ast}
\newcommand{\eqUtil}{\Util^\ast}

\newcommand{\pist}{\pi^\ast}
\newcommand{\pihat}{\hat{\pi}}

\newcommand{\completion}{z}

\newcommand{\history}{h}

\newcommand{\mech}{\calM}

\newcommand{\rev}{R}
\newcommand{\tp}{T}

\usepackage{accents}
\newcommand{\ubar}[1]{\underaccent{\bar}{#1}}
\newcommand{\binLB}{\ubar{b}}
\newcommand{\binUB}{\bar{b}}
\newcommand{\numBins}{m}

\newcommand{\maxInBin}{\bar{j}}
\newcommand{\minInBin}{\underline{j}}

\newcommand{\supPar}[1]{^{(#1)}}

\newcommand{\1}{^{(1)}}
\newcommand{\2}{^{(2)}}
\newcommand{\supk}{^{(k)}}
\newcommand{\supkmo}{^{(k-1)}}
\newcommand{\supkpo}{^{(k+1)}}
\newcommand{\supkprime}{^{(k')}}
\newcommand{\supK}{^{(\numBins)}}
\newcommand{\supKmo}{^{(\numBins-1)}}

\newcommand{\fb}{_{\mathrm{FB}}}
\newcommand{\strict}{_{\mathrm{strict}}}
\newcommand{\rand}{_{\mathrm{rand}}}
\newcommand{\direct}{_{\mathrm{direct}}}

\newcommand{\tpRand}{s}

\newcommand{\meangb}{\bar{\gb}}

\newcommand{\bestJ}{{j^\ast}}

\newcommand{\subok}{_{1:k}}

\newcommand{\subokmo}{_{1:k-1}}
\newcommand{\subkk}{_{k:k}}

\newcommand{\goft}{g}
\newcommand{\dd}{\mathrm{d}}

\newcommand{\qtilde}{\tilde{\queue}}

\newcommand{\mutilde}{\tilde{\mu}}
\newcommand{\etatilde}{\tilde{\eta}}

\newcommand{\calM}{\mathcal{M}}			
\newcommand{\calL}{\mathcal{L}}

\newcommand{\setZ}{\mathbb{Z}}

\newcommand{\txtif}{~\mathrm{if}~}

\newcommand{\txtst}{~\mathrm{s.t.}~}
\newcommand{\txtand}{~\mathrm{and}~}

\renewcommand{\th}{^{\mathrm{th}}}

\newcommand{\pwfun}[1]{\left\lbrace \begin{array}{ll} #1 \end{array} \right.}

\newcommand{\one}[1]{\mathds{1} \{ #1\}}
\newcommand{\E}[1]{\mathbb{E}\left[ #1 \right]}

\newcommand{\var}[1]{\mathrm{Var}\left( #1 \right)}

\newcommand{\ii}{_{i,i+1}}
\newcommand{\ot}{_{1,2}}

\title{Randomized FIFO Mechanisms%
\thanks{The authors would like to thank Nick Arnosti, Achal Bassambsoo, Omar Besbes, Yeon-Koo Che, Peter Cohen, Jose Correa, John Dickerson, Amos Fiat, Daniel Freund, Sergey Gitlin, Srikanth Jagabathula, Yash Kanoria, Cinar Kilcioglu, Thodoris Lykouris, Jake Marcinek, Eoin O'Mahony, David Parkes, Scott Rodilitz, Garrett van Ryzin, Lior Seeman, James Shummer, Nicolas Stier, Carmen Wang, Adam Wierman, Zhixi Wan, and participants at Uber Marketplace Matching Science Deep Dive, INFORMS 2020, Simons Institute %
Matching-Based Market Design Reunion, Columbia DRO Brown Bag seminar, Google Algorithms Workshop, the 6th Marketplace Innovation Workshop, Design of Online Platforms Workshop at EC'21, MSOM Service Management Sig 2021, Fall 2021 NBER Market Design Working Group Meeting, the Harvard EconCS seminar, and the Foster School of Business ISOM seminar for helpful comments and discussions. %
}
}

\author{Francisco Castro%
\thanks{UCLA Anderson School of Management. 110 Westwood Plaza, room B505, Los Angeles, CA 90095, USA. Email: francisco.castro@anderson.ucla.edu.}
\and
Hongyao Ma%
\thanks{Decision, Risk, and Operations Division, Columbia Business School. 423 Uris Hall, 3022 Broadway, New York, NY, 10027, USA. Email: hongyao.ma@columbia.edu.} 
\and Hamid Nazerzadeh%
\thanks{Uber Technologies, Inc. and USC Marshall School of Business. Bridge Memorial Hall, University of Southern California, Los Angeles, CA 90089  USA. Email: nazerzad@usc.edu. %
}
\and Chiwei Yan%
\thanks{Department of Industrial and Systems Engineering, University of Washington. 3900 E Stevens Way NE, Seattle, WA 98195, USA. Email: chiwei@uw.edu.}
}

\begin{document}

\maketitle

\begin{abstract}
We study the matching of jobs to workers in a queue, e.g. a ridesharing platform dispatching drivers to pick up riders at an airport. Under FIFO dispatching, the heterogeneity in trip earnings incentivizes drivers to cherry-pick, increasing riders' waiting time for a match and resulting in a loss of efficiency and reliability. We first present \emph{the direct FIFO mechanism}, which offers lower-earning trips to drivers further down the queue. The option to skip the rest of the line incentivizes drivers to accept all dispatches, but the mechanism would be considered \emph{unfair} since drivers closer to the head of the queue may have lower priority for trips to certain destinations. To avoid the use of unfair dispatch rules, we introduce a family of \emph{randomized FIFO mechanisms}, which send declined trips gradually down the queue in a randomized manner. %
We prove that a randomized FIFO mechanism achieves the first best throughput and the second best revenue in equilibrium. Extensive counterfactual simulations using data from the City of Chicago demonstrate substantial improvements of revenue and throughput, highlighting the effectiveness of using waiting times to align incentives and reduce the variability in driver earnings. %
\end{abstract}

\section{Introduction} \label{sec:intro}

Matching marketplaces play an instrumental role in economic exchanges and the allocation of public and private resources.  %
Over the past decade, the rise of online platforms connecting people with gig workers
has also radically changed many aspects of our daily lives. 
To improve efficiency and reduce waiting times, platforms often aim to match rider or grocery delivery trips with the closest available drivers. When requests are concentrated in space, however, matching by proximity has unintended consequences. 
As an example, %
Amazon drivers have been reportedly hanging their smartphones in trees near Amazon delivery stations and Whole Foods stores, in order to %
appear even closer and gain higher priority for job offers.\footnote{\url{https://www.bloomberg.com/news/articles/2020-09-01/amazon-drivers-are-hanging-smartphones-in-trees-to-get-more-work}, accessed 09/07/2020.}
A similar problem existed for Uber and Lyft at airports and event venues.%
\footnote{Airport trips account for 15\% of Uber's %
gross bookings. See Form S-1 of Uber's IPO filing: \url{https://www.sec.gov/Archives/edgar/data/1543151/000119312519103850/d647752ds1.htm}.}
Matching riders to the closest drivers incentivizes drivers to get as close to the terminal or venue as possible, leading to traffic congestion.%
\footnote{\url{https://www.vice.com/en/article/gvy357/the-new-system-uber-is-implementing-at-airports-has-some-drivers-worried}, accessed 02/23/2021.}%

Many ridesharing platforms now maintain \emph{virtual queues} at airports for drivers who are waiting in designated areas, and dispatch drivers from the queue in a first-in-first-out (FIFO) manner.%
\footnote{\url{https://help.lyft.com/hc/en-us/articles/115012922787-Receiving-Airport-FIFO-pickup-requests}, \url{https://www.uber.com/us/en/drive/dayton/airports/day/}, accessed 02/18/2021.}
This resolves the congestion issues and %
is also considered more fair by many since drivers who have waited the longest in the queue are now the first in line to receive trip offers.
At major U.S. airports, however, a driver at the head of the queue will receive the next trip offer in a few seconds under FIFO dispatching, if she declines an offer from the platform (see Figure~\ref{fig:ohare_how_trips_from_airport}). As we shall see, this lowered cost of cherry-picking substantially exacerbates existing problems on incentive alignment. %

\newcommand{\heatMapHeight}{2.2}
\newcommand{\subfigWidthTwo}{0.49}

\begin{figure}[t!]
\centering
\begin{subfigure}[t]{\subfigWidthTwo \textwidth}
	\centering    \includegraphics[height=\heatMapHeight in]{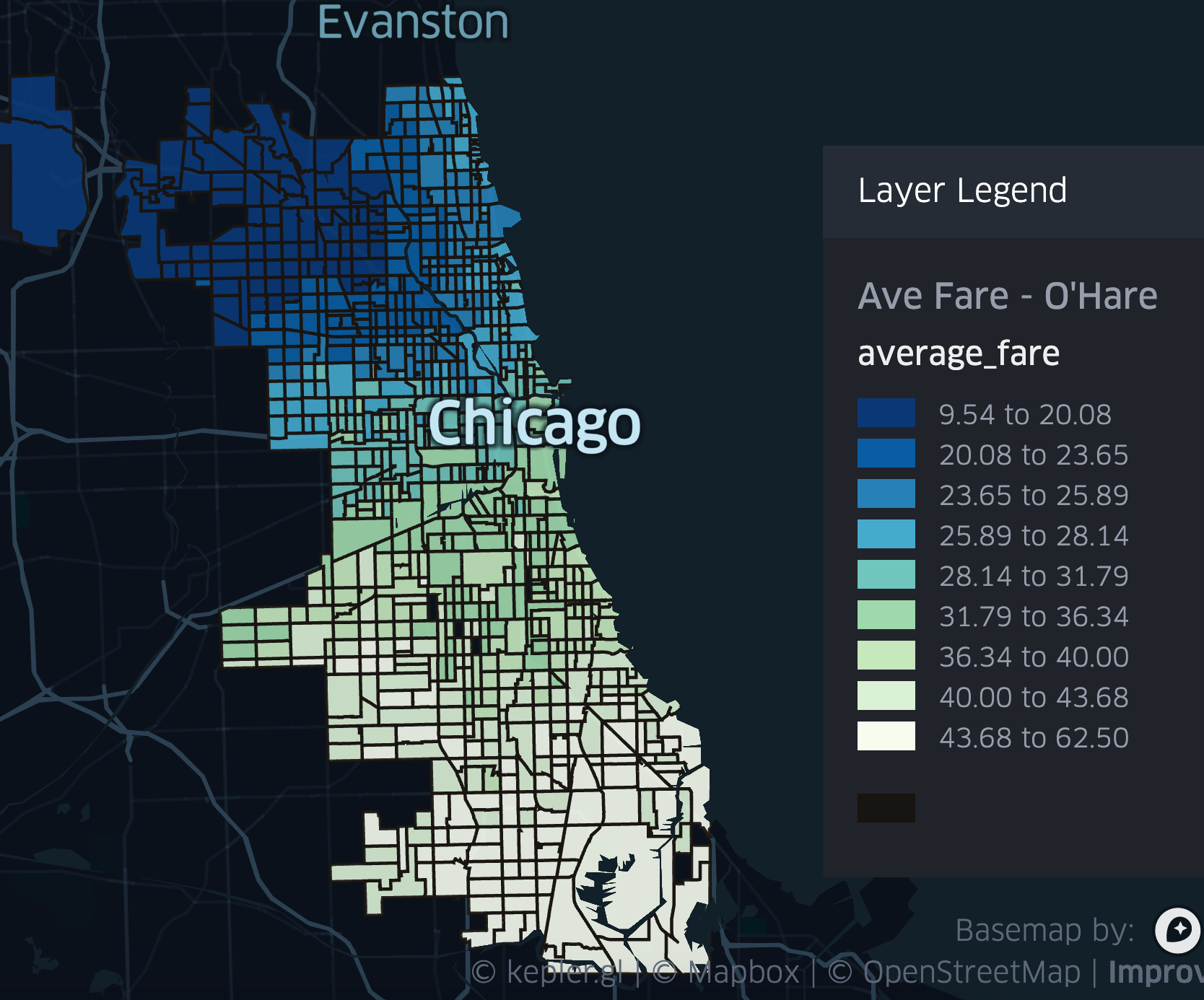}
    \caption{Trips from O'Hare. \label{fig:heatmap_ohare_average_fare}}
\end{subfigure}%
\begin{subfigure}[t]{\subfigWidthTwo \textwidth}
	\centering
    \includegraphics[height=\heatMapHeight in]{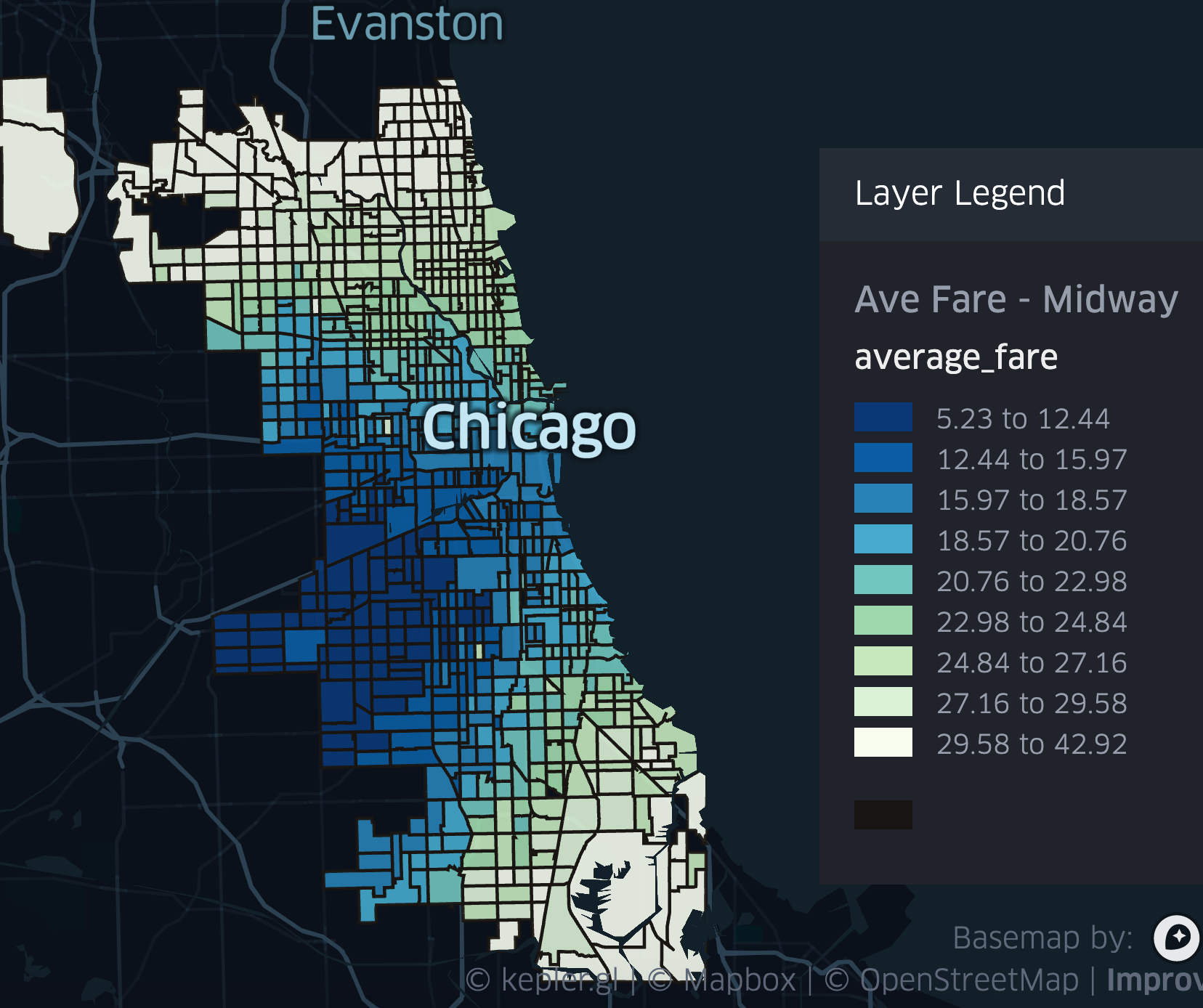}
    \caption{Trips from Midway. \label{fig:heatmap_midway_average_fare}}
\end{subfigure}%
\caption{Average trip fare by destination Census Tract in Chicago, for trips originating from the O'Hare International Airport and the Midway International Airport. %
See Section~\ref{sec:simulations} for more details. \label{fig:average_fare_from_chicago_airports}}
\end{figure}

Figure~\ref{fig:average_fare_from_chicago_airports} shows the average trip fare by destination Census Tract in Chicago, for trips originating from the O'Hare and Midway airports. A short trip from O'Hare to a nearby area pays an average of \$10-\$20, but a long trip can pay %
an average of \$60. During busy hours, instead of accepting an average trip, drivers who are close to the head of the queue are better off declining most trip offers and waiting for only the highest earning trips.
Riders, however, have finite patience, despite being willing to wait for some time for a match. When each driver decline takes an average of 10 seconds, 2 minutes had passed after a trip with low or moderate earnings (e.g. trips to downtown Chicago) was offered to and declined by the top 12 drivers in the queue.%
\footnote{\label{fn:offer_card_15_seconds}%
When offered a trip, Uber and Lyft drivers have 15 seconds to decide whether to accept.  \url{https://help.uber.com/driving-and-delivering/article/getting-a-trip-request?nodeId=e7228ac8-7c7f-4ad6-b120-086d39f2c94c}, \url{https://help.lyft.com/hc/en-/articles/115013080028-How-to-give-a-Lyft-ride}, accessed 02/24/2021.}
At this point, it is very likely that the rider cancels her trip request, not knowing when a driver will be assigned, if at all.

What we have seen is that in the presence of heterogeneous earnings and finite rider patience, trips with moderate or low earnings never reach drivers in the queue who are willing to accept them. This undercuts the platforms' mission of providing reliable transportation for riders, and leads to low revenue and trip throughput for the platform. Moreover, fulfilling only the small number of high earning trips is also a poor outcome for the drivers, since many drivers who just dropped off a rider at the airport will have to relocate back to the city with an empty car, and those who do join the queue would need to wait for a very long time for a ride. %

\medskip

Simple fixes by limiting dispatching transparency or drivers' flexibility are not desirable--- in recent years, ridesharing platforms are moving towards sharing trip destination and earnings estimation upfront, as well as providing drivers the options to accept or decline any trips without penalties.%
\footnote{\label{fn:prop22}%
The specific policies and their implementations vary across companies and geographical regions. As an example, see \url{https://www.uber.com/blog/california/keeping-you-in-the-drivers-seat-1/} (accessed 02/21/2021). This is in part due to regulatory requirements for categorizing drivers as independent contractors. For context, we cite an excerpt from California Proposition 22: ``The network company does not require the app-based driver to accept any specific rideshare service or delivery service request as a condition of maintaining access to the network company's online-enabled application or platform." %
}
Hiding information or imposing penalties are not fully effective either. For example, experienced drivers often call riders to ask about trip details when destinations are hidden %
before the pick-up~\citep{cook2018gender}. %
Forcing drivers to accept every dispatch %
improves reliability in the short run, but also imposes a lottery (with possible outcomes ranging from \$9 to over \$60) on drivers who might have waited for two hours in line. Such high variance in earnings discourages future engagement, and leads to drivers' churning from the platform in the long run. %
\todo{Insert the article on ``open access'' for all.}

Recognizing the inefficiencies under FIFO dispatching, alternative mechanisms have been studied extensively in the literature. In particular, last-in-first-out (LIFO) dispatching is shown to be optimal
in the presence of waiting costs, %
with or without heterogeneous rewards~\citep{hassin1985optimality,su2004patient}. Intuitively, participants' losing (instead of gaining) priority over time substantially reduces the incentive to ``wait for a better offer''. %
However, LIFO dispatching is perceived as %
``blatantly unfair'' by many~\citep{su2004patient,breinbjerg2016strategic}. Moreover, as discussed by \citet{hassin1985optimality} and \citet{su2004patient}, %
LIFO dispatching is easy to manipulate since participants may rejoin the queue at the end to (re)gain priority.
This renders LIFO unsuitable and ineffective for ridesharing platforms, as drivers have the option to go offline and online again to rejoin the end of the virtual queue at any time.

Ideally, platforms may properly price trips by destination %
and eliminate drivers' incentives to cherry-pick. In recent work, \citet{ma2019spatio} propose the spatio-temporal pricing mechanism, which is welfare optimal, incentive aligned, and guarantees that drivers at the same origin are indifferent towards trips to all destinations. 
This remains an idealized target for our current setting, but is hard to achieve in practice. %
Consider, again, the O'Hare example as shown in Figure~\ref{fig:heatmap_ohare_average_fare}. 
Tripling the fares of the short trips to match the earnings from the long trips is suboptimal. On the other hand, the platform is unable to decrease driver payouts below some pre-determined per-minute and per-mile rates, thus earnings from long trips cannot be effectively reduced either. %

\subsection{Our Results}

\todo{Right now this is the only place where we mentioned ``not changing prices''.}

We study the dispatch of trips to drivers who are waiting in a virtual queue, where some trips are necessarily more lucrative than the others due to operational constraints. Without the power to adjust trip prices, the mechanisms we design %
use drivers' waiting times in the queue to %
align incentives, improve reliability and efficiency, and reduce the variability in drivers' total payoffs.

\paragraph{The model.}

We study a continuous time, non-atomic model, with one origin (e.g. the airport) %
and stationary arrival rates of riders and drivers. Riders request trips to a %
number of destinations with heterogeneous earnings for drivers. %
Upon the arrival of each rider, or after a rider's trip request is declined, the platform offers the rider's trip to a driver in the queue. %
Riders are willing to wait for some time %
for a match, but have finite patience and will cancel their requests %
after a certain number of declines from drivers. Drivers' waiting in the queue (as opposed to driving elsewhere in the city, for example) is costly for both the drivers and the platform. Drivers are strategic, aiming to optimize their total payoff, i.e. the earnings from trips minus the waiting costs they incur. %

We study mechanisms that are fully \emph{transparent} and \emph{flexible}  (see Footnote~\ref{fn:prop22}). At any point in time, drivers know about the supply, demand, the length of the queue and their positions in their queue. When offered trip requests, drivers are provided trip destinations and earnings \emph{upfront} so that they can decide whether to accept based on this information. %
Moreover, drivers are not penalized for any actions they take, and have the flexibility to (i) decline any number of trip dispatches %
without losing their positions in the queue, (ii) rejoin the virtual queue at the tail at any point of time, and (iii) decide to not join the queue upon arrival, or leave the queue at any point of time to perhaps relocate back to the city without a rider.

\paragraph{Main results.}
To optimize trip throughput and the platform's net revenue (i.e. total earnings from completed trips, minus the opportunity costs the platform incurs due to drivers' waiting in the queue), %
the \emph{first best} outcome has no driver in the queue, and dispatches all drivers upon arrival to destinations in decreasing order of earnings. %
However, under the status quo \emph{strict FIFO dispatching} where trips are dispatched to each and every driver starting from the head of the queue, drivers close to the head of the queue are incentivized to %
cherry-pick and wait for higher-earning trips. We analyze the equilibrium outcome under strict FIFO and show that with finite rider patience, most trips except for the highest earning ones become unfulfilled. %
Drivers' excessive waiting in the queue further reduces drivers' total payoffs as well as the net revenue of the platform.

Recognizing that the moderate and low earning trips never reach drivers in the queue who are willing to accept them, we  first present %
the \emph{direct FIFO mechanism}, which offers lower-earning trips directly to drivers further down the queue. 
We prove that %
accepting all dispatches forms a subgame perfect equilibrium among drivers, and that the equilibrium outcome achieves the first best trip throughput, and \emph{the second best net revenue} (i.e. the highest steady state net revenue achievable by any flexible and transparent mechanism). 
The direct FIFO mechanism, however, would be considered \emph{unfair} in practice since a driver may have lower priority for trips to many destinations than drivers %
further down the queue, even when all drivers are non-strategic and accept every dispatch. %
Consider the Chicago Midway airport (Figure~\ref{fig:heatmap_midway_average_fare}) as an example. A driver close enough to the head of the queue will no longer receive any trip back to downtown Chicago, since direct FIFO skips drivers at the head of the queue when dispatching lower and moderate earning trips, and all high-earning trips the driver may receive will be heading to the suburbs.

To achieve optimal throughput and revenue without the use of an unfair dispatch rule, %
we introduce a family of \emph{randomized FIFO mechanisms}. 
A randomized FIFO mechanism is specified by a set of ``bins'' in the queue %
(e.g., the top 10 positions, the $10\th$ to $20\th$ positions, and so on). %
Each trip request is first offered to a driver in the first bin uniformly at random. %
After each decline, the mechanism then offers the trip to a random driver in the next bin.
By sending trips gradually down the queue in this randomized manner, the randomized FIFO mechanisms appropriately align incentives using waiting times, achieving the first best throughput and second best net revenue: %
the option to skip the rest of the line incentivizes drivers further down the queue to accept trips with lower earnings; randomizing each dispatch among a small group of drivers %
increases each individual driver's waiting time for the next dispatch, thereby allowing the mechanism to prioritize drivers closer to the head of the queue for trips to every destination without creating incentives for excessive cherry-picking.

Extensive counterfactual simulations using %
data from the City of Chicago suggest that in comparison to %
strict FIFO dispatching, the randomized FIFO mechanism achieves substantial improvements in revenue, throughput, and driver earnings. Moreover,  %
the variance in drivers' total payoffs is small, and diminishes rapidly as riders' patience increases--- %
with higher rider patience, the mechanism can more effectively match higher-earning trips with drivers who have incurred higher waiting costs in the queue. This demonstrates the desirable balance achieved by the randomized FIFO mechanisms between efficiency, reliability, fairness, and the variability in driver earnings, and highlights the effectiveness of using waiting times in queue to align incentives and to reduce earning inequity when the flexibility to set prices is limited due to operational constraints.

\subsection{Related Work} \label{sec:related_work}

\paragraph{Ridesharing platforms.} 

The literature on pricing and matching in ridesharing platforms is rapidly growing. \citet{castillo2017surge} and \citet{yan2020dynamic} establish the importance of dynamic pricing in maintaining the spatial density of open driver supply, which reduces waiting times and improves operational efficiency. 
In the presence of spatial imbalance and temporal variation of supply and demand, \citet{bimpikis2019spatial} and \citet{besbes2020surge} study revenue-optimal pricing; %
\citet{ma2019spatio} propose origin-destination based pricing that is appropriately smooth in space and time, achieving welfare optimality and incentive compatibility; \citet{garg2019driver} show that additive instead of multiplicative ``surge'' pricing is more incentive aligned for drivers when prices need to be origin-based only. 
Considering the online arrival of supply and demand and their distribution in space, \citet{kanoria2020blind}, \citet{qin2020ride} and \citet{ozkan2020dynamic} study dynamic matching policies that dispatch drivers from areas with relatively abundant supply, and \citet{ashlagi2019edge}, \citet{dickerson2018allocation} and \citet{aouad2020dynamic} focus on the online matching between riders and drivers and the pooling of shared rides. 
In this work, we focus on a single origin where the optimal destination-based pricing is infeasible due to operation constraints such that some trips are necessarily more lucrative than the others. This leads to the need of using drivers' waiting times to align incentives and to reduce the variability in driver earnings.

The operation of ridesharing platforms is also studied using queueing-theoretic models.
\citet{BanJohRiq2015Pricing} compare optimal dynamic and static pricing policies;  
\citet{banerjee2018state} propose state-dependent dispatching policies to minimize %
unfulfilled demand; \citet{afeche2018ride} study the impact of %
admission control on platform revenue and driver income; 
\citet{besbes2018spatial} show that in comparison to traditional service settings, higher capacity is needed when spatial density of available supply affects operational efficiency; 
\citet{castro2020matching} study %
practical dispatching policies when drivers have heterogeneous compatibility with trips.
These works use queueing-theoretic frameworks %
to analyze the availability of driver supply, but study settings where drivers are spread out in space, and do not consider cherry-picking by drivers. %
In contrast, we focus on %
the matching of trips to drivers who are waiting in a \emph{virtual queue}, addressing the problem of dispatching heterogeneous trips to drivers who have incurred different waiting costs in a way that is reliable, efficient and fair.

Various empirical studies analyze the Uber platform as a two-sided marketplace, focusing on the labor market of Uber drivers~\citep{hall2017labor}, the longer-term labor market equilibration~\citep{hall2016analysis}, the value of flexible work arrangements~\citep{chen2019value}, learning-by-doing and the gender earnings gap~\citep{cook2018gender}, and the surplus of consumers~\citep{cohen2016using}. In regard to the dynamic ``surge'' pricing, \citet{hall2015effects}, \citet{chen2015dynamic}, and \citet{lu2018surge} demonstrate its effectiveness in improving reliability and efficiency, increasing driver supply during high-demand times, as well as incentivizing drivers to relocate to higher demand areas. 
In contrast, we use data from ridesharing platforms (including Uber and Lyft, made public by the City of Chicago) to %
estimate %
the heterogeneity in driver earnings by trip destination. We also demonstrate via 
counterfactual simulations the inefficiencies of %
FIFO dispatching when drivers are strategic, as well as the substantial improvements achieved by our proposed mechanisms.

\paragraph{%
Queueing mechanisms.}

The allocation of resources or jobs to participants waiting in a queue has been studied extensively in the literature. %
\citet{naor1969regulation} first demonstrates the negative externalities from waiting: %
when agents make self-interested decisions on whether to join a FIFO queue, in equilibrium more agents line up in the queue in comparison to the socially optimal outcome. 
When monetary transfers are allowed, Naor shows that %
the optimal outcome can be achieved by levying an entrance toll, and a large body of subsequent work has studied how to align incentives and improve system efficiency in various settings (see \citet{hassin2016rational} for a comprehensive review). In many practical settings including ours, however, the use of monetary incentives is restricted due to %
regulatory or business constraints.

Without the use of monetary transfers, \citet{hassin1985optimality} shows that the last-in-first-out (LIFO) queueing discipline achieves the socially optimal outcome in equilibrium, since when the agent who has waited the longest in the queue decides whether to leave, she imposes no externality on any current or future agents. 
With homogeneous agents who prefer items of higher quality (i.e. when all patients prefer kidneys from younger and healthier donors in the context of kidney transplantation), \citet{su2004patient} demonstrate the excessive organ wastage resulting from patients' cherry-picking under FIFO, and proves that LIFO dispatching optimizes organ utilization. %
These works highlight the important role of the queueing discipline in shaping participants' strategic considerations. As %
is discussed in these papers, however, LIFO is practically infeasible since the dispatch rule (i) would be perceived as unfair, and (ii) can be easily manipulated by re-joining the queue. %
In this work, we propose practical mechanisms that %
allow drivers to decline dispatches and to re-join %
the queue at any point of time. %
Moreover, we model the fact that riders' finite patience limits the number of times a trip can be dispatched, %
and prove that no transparent and flexible mechanism can achieve a better outcome than ours even when assuming infinite rider patience.%

\hmadd{On the flip side, \citet{che2021optimal} establish the optimality of FIFO when the planner has full flexibility to %
(i) prevent participants from joining the queue and remove participants from the queue, and (ii) design the information %
provided to the participants. 
The objective is to optimize a weighted sum of %
the participants' utility and the service provider's profit. 
Intuitively, when the planner has the power to ensure that the queue is not too long, FIFO dispatching is the most effective since it provides the strongest incentive for participants to join %
and to stay in the queue. 
}

\citet{su2006recipient} and \citet{ashlagi2020optimal} study settings %
where %
an agent's value for an item depends on the type of the item %
and the private type of the agent.
\citet{su2006recipient} design disjoint queue mechanisms that optimize either efficiency or equity (i.e. the minimum utility across all agent types). %
Assuming that %
the value for a match is supermodular in %
the types of the agent and the item, \citet{ashlagi2020optimal} %
establishes that a monotone disjoint queue mechanism is welfare-optimal. In both settings, %
agents cannot decline the allocated items.%
\footnote{The same optimal outcome in \citet{ashlagi2020optimal} can also be achieved by a FIFO queue that allows agents to decline undesired items, %
assuming that the items are infinitely patient, and that the mechanism %
does not have to reveal full information on the offered items to the agents. %
}
Therefore, %
the mix of items dispatched to each queue effectively determines a %
lottery over items, and the %
waiting times in the different queues function as prices %
and incentivize an agent to %
choose the lottery intended for her type.  
In contrast, instead of %
eliciting private information, we focus on %
improving reliability without using penalties %
or hiding information. %
Our mechanism effectively dispatches every trip according to ``a sequence of lotteries over positions in the queue'', %
aligning incentives using (i) the option to skip the rest of the line and (ii) the additional cost of cherry-picking introduced by randomization. %

Existing work also compare FIFO and randomized allocation rules in various settings. Assuming an overloaded queue with fixed length, \citet{bloch2017dynamic} show that agents in the queue %
prefer FIFO, %
but randomizing offers among all agents in the queue reduces waste, thus improves turnover and benefits agents who are not yet in the queue. %
Also assuming an overloaded queue, %
\citet{leshno2019dynamic} focuses on  
inefficiencies arising from the ``mismatch'', i.e. agents accepting their less preferred item 
since the wait for the more preferred item is too long. 
In a buffer queue for agents who have declined a less preferred item, %
randomizing offers reduces the variability of the expected waiting time for the more preferred item and %
reduces %
mismatches compared to FIFO. 
When agents have heterogeneous preferences over affordable housing developments, \citet{arnosti2020design} %
prove that ``individual lotteries'' (one for each development) achieves the same %
outcome as a ``wait-list without choice'', both compelling agents to accept poor matches. %
More choices (via e.g., wait-list with choice) leads to better matching, but the authors also establish a trade-off between matching and targeting 
agents with worse outside options. 
In all three settings, the randomization is among all agents in the queue. %
In contrast, our proposed mechanisms randomize each dispatch among drivers from a small segment in the queue, %
which increases the costs of cherry-picking without introducing excessive variability in drivers' total payoffs. 

\hmadd{In this work, we focus on settings where participants have the flexibility to decline dispatches without losing their positions in the queue. \citet{schummer2021influencing} analyze the impact of limiting this ``deferral right'' for various settings, where participants are risk averse or discount the future. 
}

\section{Preliminaries} \label{sec:preliminaries}

\todo{HN: I think we should comment here about the generality of our model, and then focus on the specific application.}

We study a continuous time model, with one origin (e.g. an airport) where trips are dispatched to drivers who are waiting in a queue. $\loc = \{1, 2, \dots, \numLoc\}$ denotes the set of $\numLoc \in \setZ_{>0}$ discrete trip types (e.g. trips to different destinations). Rider demand and driver supply are non-atomic and are stationary over time. For each location $i \in \loc$, $\numRiders_i > 0$ denotes the arrival rate of riders requesting trips to location $i$ (i.e. the mass of riders arriving per unit of time). %
Upon arrival, riders' trip requests need to be dispatched to the drivers. All riders have a \emph{patience level} of $\patience \in \setZ_{> 0}$, meaning that a rider may be willing to wait for a while for a driver to accept her trip request, %
but she will cancel her request and leave after the $\patience \th$ time that her trip is declined by the drivers. %
Each driver can drive any rider to her destination, and riders do not have preferences over drivers.

Let $\numDrivers > 0$ be the arrival rate of drivers. %
Upon arrival, the driver may decide whether to join the queue. %
The \emph{net earnings} of a trip to each location $i \in \loc$ is $\gb_i$, %
meaning that a driver who completes a rider trip to location $i$ gets a payoff of $\gb_i$ from the trip, %
and the payoff of a driver who does not join the queue or leaves the queue without a rider is normalized to be zero. %
For each unit of time a driver spends waiting in the queue, the driver incurs an opportunity cost of $\cost > 0$, and the platform incurs an opportunity cost of $\mechCost \in [0, \cost]$.%
\footnote{The opportunity costs for drivers captures the value of their forgone outside options, which include, for example, the potential earnings a driver can make from driving elsewhere in the city for the same platform instead of waiting in the queue. %
Having driver supply tied-up in the queue is thereby potentially costly for the platform as well.}%
\footnote{Drivers who are waiting in the queue may not drive for the same platform at all times, and the market might be oversupplied already. %
As a result, the opportunity cost for the platform $\mechCost$ may be lower than that for the drivers.
}
Drivers are strategic, aim to optimize their earnings from trips minus their waiting costs, and do not have preferences over riders or destinations.

\medskip

An informal timeline of a dispatching mechanism is as follows (see Section~\ref{sec:direct_FIFO} for the formal definition). Upon the arrival of each rider, the %
mechanism may dispatch the rider's trip request to a driver in the queue. If the driver accepts the dispatch, she leaves the queue to pick up the rider. %
Otherwise, the trip may be dispatched again, until (i) some driver accepts the trip, or (ii) the rider cancels her request when her patience runs out (after the trip is declined for $\patience$ times), or (iii) the mechanism decides to not dispatch the trip again.%
\footnote{%
It takes some time for trip dispatches to be accepted or declined by  the drivers %
(see Footnote~\ref{fn:offer_card_15_seconds}). 
Drivers in the queue will be moving forward during the time a trip is repeatedly dispatched, but this does not affect our results since (i) the dispatch rules and driver strategies we study depend on the positions in the queue instead of the identities of individual drivers, and (ii) the optimality results we establish focus on the equilibrium outcome in steady state. 
}

We consider a setting where the platform has complete information about demand, supply, opportunity costs, and the earnings from trips to different destinations. We assume drivers have the same information, and that this is common knowledge amongst the drivers. %
We study mechanisms that are fully \emph{transparent} and \emph{flexible}. At any point in time, all drivers know the total length of the queue and their positions in the queue. When offered trip dispatches, drivers are provided trip destinations and earnings \emph{upfront}, so that they can decide whether to accept a dispatch based on this information. Moreover, drivers are not penalized for actions they take, and have the options to (i) decline dispatches they do not want to accept without losing their position in the queue, (ii) rejoin the virtual queue at the tail at any point of time, and (iii) decide to not join the queue upon arrival, or leave the queue at any point of time to perhaps relocate back to the city without a rider.

\medskip

A platform's \emph{trip throughput} is the mass of trips completed per unit of time by drivers in the queue.  A platform's \emph{net revenue} is the sum of the net earnings from trips made by drivers %
per unit of time, minus the opportunity cost the platform incurs due to drivers' waiting in the queue (this opportunity cost models the platform's loss of revenue elsewhere in the city, due to driver supply being tied-up in the queue). %
When drivers are non-strategic and %
accept all dispatches from the platform, %
we refer to the highest achievable %
trip throughput and net revenue %
as \emph{the first best}.

For simplicity of notation, we assume the destinations are ordered such that $\gb_1> \gb_2 > \dots > \gb_\numLoc \geq 0$.\footnote{Combining destinations with the same net earnings does not affect the equilibrium outcome of any mechanism we study. %
For drivers who are free to decline dispatches based on trip destinations, no trip with $\gb_i < 0$ will be accepted since completing one such trip is worse than declining the dispatch and leave the queue immediately without a rider. %
\todo{Proper discussion on pricing once it's written.}
}
With stationary and infinitesimal demand and supply, a platform does not need a non-zero driver queue. In steady state, a platform that aims to optimize its net revenue should keep no driver in the queue, but dispatch drivers upon their arrival to destinations in decreasing order of $\gb_i$ until either all drivers are dispatched or all riders are picked-up. %
Denote the lowest-earning trip that is (partially) completed as
\begin{align}
	\maxJ = \max \left\lbrace i \in \loc ~\middle|~ \lambda > \sum_{j = 1}^{i-1} \mu_j  \right\rbrace. \label{equ:maxJ}
\end{align}

\begin{proposition}[The first best] \label{prop:first_best}
The steady state first best outcome has zero drivers in the queue. Upon arrival, drivers are dispatched to pick up arriving riders in decreasing order of $\gb_i$. The remaining drivers (if any) are suggested to leave without joining the queue.
The \emph{first best trip throughput} is 
\begin{align}
	\tp\fb = \min\left\{\lambda, ~\sum_{i \in \loc} \mu_i \right\}, \label{equ:fb_tp}
\end{align}
and the \emph{first best net revenue} is  
\begin{align}
	\rev\fb = \sum_{i=1}^{\maxJ-1} \gb_i \mu_i  + \gb_\maxJ  \min\left\{ \lambda - \sum_{j=1}^{\maxJ-1} \mu_i, ~\mu_\maxJ  \right\}. \label{equ:fb_revenue}
\end{align}
\end{proposition}

\subsection{Strict FIFO Dispatching} \label{sec:strict_FIFO}

The %
FIFO queue discipline is considered \emph{fair} by most, and is the default discipline in many everyday situations~\citep{larson1987Perspectives,breinbjerg2016strategic,platz2017curse}. %
We show that when drivers have the flexibility to decline undesired trips, offering each trip to every driver %
incentivizes excessive cherry-picking and leads to poor outcomes for the riders, drivers, and the platform. 
To avoid ambiguity, we refer to this mechanism as {\em strict FIFO dispatching.}

We start by analyzing the equilibrium outcome under strict FIFO dispatching. Consider a rider request for a trip to location $1$. Under strict FIFO, the trip will be accepted by the driver at the head of the queue, since a trip to location~$1$ has the highest earnings among all trips the driver may receive in the queue. Moreover, the (infinitesimal) driver at the head of the queue will be willing to accept only trips to location~$1$, since she is the first in line to receive all incoming trip dispatches, thus she does not have to wait any time for a trip dispatch to location~$1$.

Similar reasoning holds for drivers who are very close to the head of the queue, and a driver is willing to accept a trip to location~$2$ %
only if the additional waiting cost for a trip to location $1$ outweighs the earnings gap $\gb_1 - \gb_2$. 
Let $\waitGap\ot$ be the maximum additional time a driver is willing to wait for a trip to location $1$, in comparison to immediately taking a trip to location $2$. We know
\begin{align}
	\waitGap\ot \cost  = \gb_1 - \gb_2  \Rightarrow \waitGap\ot = (\gb_1 - \gb_2)/\cost. %
	\label{equ:wait_time_gap}
\end{align}

By Little's Law, the first position in the queue where the driver is willing to accept a location~$2$ trip is $\numSkip_2 \triangleq  \mu_1 \waitGap\ot = \mu_1 (\gb_1 - \gb_2)/\cost$, since the waiting time from this position to the head of the queue is $\waitGap\ot$ when all drivers ahead of this position only accept trips to location $1$. %
For a driver at position $\numSkip_2$, her \emph{continuation payoff} (i.e. net earnings from the trip the driver accepts minus the waiting costs the driver incurs from this point of time onward) is $\gb_2$, regardless of whether the driver accepted a trip to location~$2$, or if the driver continued to wait for a trip to location~$1$.

Similarly, in comparison to accepting a trip to location $i+1$, a driver is willing to wait an additional $\waitGap \ii = (\gb_i - \gb_{i+1})/\cost$ units of time for a trip to location $i$. We can compute the first positions in the queue where drivers are willing to accept trips to each location $i \in \loc$, assuming that riders have infinite patience and will not cancel their trip requests regardless of how many times their trips have been declined by the drivers (see Figure~\ref{fig:strict_FIFO_equ}).

\begin{restatable}[informal]{lemma}{LemStrictFIFOEqu} \label{lem:strict_FIFO_eq} 
Assume that riders have infinite patience. Under strict FIFO dispatching, it is an %
equilibrium for a driver to accept trip dispatches to each location $i \in \loc$ %
only if the driver is at position $\queue \geq \numSkip_i$ in the queue, where $\numSkip_1 \triangleq 0$ and %
\begin{align}
	\numSkip_i & \triangleq  \sum_{j = 1}^{i-1} \left( \frac{\gb_j - \gb_{j+1}}{ \cost}  \sum_{k=1}^j \mu_k \right), ~\forall i \geq 2.
	 \label{equ:first_accpet_positions}
\end{align}
\end{restatable}

\newcommand{\lamdmarkHeight}{0.7}
\newcommand{\patienceLocation}{7}

\begin{figure}[t!]
\centering
\begin{tikzpicture}[scale = 0.7] %

\draw[-, line width=0.25mm] (-0,0) -- (15.5,0); %

\draw(-1.8, -0.25) node[anchor=south] {Riders};

\draw[-, line width=0.5mm] 	(0,0.1) -- (0,-0.1);
\draw(0, 1.2) node[anchor=south] {Head of the queue};
\draw(0, \lamdmarkHeight) node[anchor=south] {$\numSkip_1$};
\draw[->] (0, \lamdmarkHeight - 0.1) -- (0, 0.15);

\draw [decorate, decoration={brace, mirror, amplitude=8 pt}, xshift=0.4pt,yshift=-0.4pt] (0,0) -- (5,0); %
\node[text width=2.3cm] at (2.8, -0.8) {$\waitGap \ot$ periods};

\draw[-, line width=0.5mm] 	(5,0.1) -- (5,-0.1);
\draw(5, \lamdmarkHeight) node[anchor=south] {$\numSkip_2$};
\draw[->](5, \lamdmarkHeight - 0.1) -- (5, 0.15);

\draw [decorate, decoration={brace, mirror, amplitude=8 pt}, xshift=0.4pt,yshift=-0.4pt] (5,0) -- (11,0); %
\node[text width=2.3cm] at (8.5, -0.8) {$\waitGap_{2,3}$ periods};

\draw[-, line width=0.5mm] 	(11,0.1) -- (11,-0.1);
\draw(11, \lamdmarkHeight) node[anchor=south] {$\numSkip_3$};
\draw[->](11, \lamdmarkHeight - 0.1) -- (11, 0.15);    
    
\draw(13, 0.6) node[anchor=south] {$\dots$};

\end{tikzpicture}
\caption{The equilibrium outcome under strict FIFO dispatching, assuming infinite rider patience. %
\label{fig:strict_FIFO_equ}} 
\end{figure}
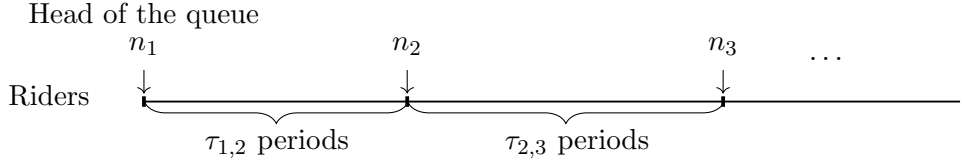

We provide in Appendix~\ref{appx:proof_strict_FIFO} the formal statement and the proof of the equilibrium outcome under strict FIFO dispatching. Briefly, we prove by induction that assuming infinite rider patience, for each location $i \in \loc$, a driver at position $\numSkip_i$ in the queue gets a continuation payoff of $\gb_i$ regardless of whether she accepted a trip to location $i$ or not. Drivers at positions earlier than $\numSkip_i$ in the queue are, however, better off waiting for trips with higher earnings.

When riders have a finite patience level $\patience$, however, trip requests to locations $i \in \loc$ with $\numSkip_i > \patience$ will not reach a driver who is willing to accept this trip before the rider's patience runs out. As a result, trips to these destinations become unfulfilled, leading to %
poor efficiency and reliability.
The following example demonstrates that drivers' excessive waiting in the queue further reduces drivers' total payoffs as well as the net revenue of the platform.

\begin{example} \label{exmp:model_and_fb}
Consider an airport queue, where riders request trips to three destinations $\loc = \{1, 2, 3\}$. The arrival rate of riders to each destination is $\mu_1 = 1$, $\mu_2 = 6$, and $\mu_3 = 3$, and the net earnings from these trips are $\gb_1 = 75$, $\gb_2 = 25$, $\gb_3 = 15$, respectively. Intuitively, trips to location~$1$ represent the rare but high-earning long trips from the airport. Location~$2$ can be considered as the downtown area with high trip volumes and medium earnings, and think about location~$3$ trips as short rides to the hotels and towns surrounding the airport with low earnings.

Drivers arrive at a rate of $\lambda = 5$ per unit of time, and the opportunity costs for the drivers and the platform are $\cost = \mechCost = 1/3$. Considering each unit of time as one minute, this corresponds to a scenario where %
a driver driving for the platform elsewhere in the city will make \$20 per hour. Riders have a patience level of $\patience = 12$. When it takes an average of 10 seconds for each driver to decline a dispatch, this corresponds to the riders' being willing to wait for two minutes for a match before canceling their trip requests. 

\medskip

\noindent{}\emph{The first best.} 
The first best outcome accepts all trips to location $1$, and dispatches the remaining $4$ units of drivers to trips to location~$2$. No driver waits in the line. The first best trip throughput is $\tp \fb = 5$, and the first best net revenue is $\rev \fb = \gb_1 + 4\gb_2 = 175$. 
This outcome can be implemented, for example, by forcing each driver to always accept the first trip dispatch she receives. This, however, introduces a high variance in drivers' total payoffs (net earnings from trip minus waiting costs): %
the average total payoff of a driver who arrived at the queue is $35$, and the 
variance is $400$. %

\medskip

\noindent{}\emph{Strict FIFO dispatching.} Under strict FIFO dispatching, when drivers have the flexibility to decide which trips to take, the driver at the head of the queue is only willing to accept trips to location~$1$. A driver with a location~$2$ trip in hand is willing to wait an additional $\waitGap\ot = (\gb_1-\gb_2)/\cost = 150$ minutes for a trip to location $1$. Therefore, the first position in the queue where the driver is willing to go to location $2$ is $\numSkip_2 = \waitGap \ot \mu_1 = 150$. 
With a patience level of $12$, riders requesting trips to location $2$ will cancel their trip requests %
after their requests are declined by the $12\th$ driver in the queue. %
Location~$3$ trips are similarly unfulfilled, %
thus strict FIFO dispatching achieves a trip throughput of only $\tp\strict = 1$ per minute.

The remaining $4$ units of drivers will need to leave the queue without a rider trip in steady state. The drivers, however, will not leave %
if the payoff from joining the queue at the tail is better than that from relocating without a rider. Drivers are willing to wait for $\gb_1/\cost = 225$ minutes for a trip to location~$1$, thus the steady state queue length will be $\mu_1 \gb_1/\cost = 225$ by Little's Law. In equilibrium, drivers get a payoff of zero regardless of whether they joined the queue or left without a rider. The large number of drivers waiting in the queue is also very costly for the platform, which achieves in this example a net revenue of zero: $\rev \strict = \gb_1 \mu_1 - \mechCost (\mu_1 \gb_1/\cost) = 0$. 
\qed
\end{example}

Strict FIFO dispatching is \emph{fair} in the sense that drivers who are closer to the head of the queue have higher priority for trips to every destination. However, as we have seen in the above example, dispatching each trip to each and every driver in the queue leads to poor reliability for the riders, low trip throughput and net revenue for the platform, and zero earnings for the drivers %
despite their strategizing for better earnings. 
In the next section, we will see that by deprioritizing drivers at the head of the queue for trips to certain destinations (thereby violating what is typically perceived as fair dispatching), we are able to substantially improve the outcome for the riders, drivers, and the platform, even without the power to adjust trip prices.

\section{The Direct FIFO Mechanism} \label{sec:direct_FIFO}

In this section, we introduce \emph{the direct FIFO mechanism}. The mechanism is based largely on FIFO dispatching, but %
sends lower-earning trips starting from positions further down the queue where drivers are willing to accept the dispatches for the option to skip the rest of the line. Accepting all trips forms a \emph{subgame perfect equilibrium} among drivers, and the mechanism achieves the highest possible revenue and throughput under any mechanism that is flexible and transparent. %

\subsection{A Dispatching Mechanism} \label{sec:dynamic_mech}

We first formally define a dispatching mechanism.  Let $\Queue \geq 0$ denote the length of the queue, and let $\queue \in [0, \Queue]$ be a particular position in the queue. $\queue = 0$ and $\queue = \Queue$ are the \emph{head} and the \emph{tail} of the queue, respectively, i.e. positions where the drivers have waited the longest and the shortest time in line. 
Let $\history$ denote the past dispatching history of a particular rider's trip request. %
This represents the positions in the queue to which the trip was offered (if any). 
Finally, we use $\phi$ %
to denote the decision to not dispatch a rider's trip request to any driver.

\begin{definition}[Dispatching mechanism] \label{defn:dynamic_mech}
Given the queue length $\Queue$, the past dispatching history $\history$ of a trip, and the trip's destination, a dispatching mechanism %
determines a probability distribution over $[0, \Queue] \cup \{\phi\}$.
Upon the arrival of a rider, or after a rider's trip is declined by some driver, the mechanism %
either (i) dispatches the trip to a driver at some position $\queue \in [0, \Queue]$ in the queue, or (ii) decides to not dispatch the trip (which we denote as $\phi$).
\end{definition}

The queue length $\Queue$ represents the \emph{state} of the queue. The dispatching mechanisms we study %
make dispatch decisions for each trip based on the state and the past dispatch history of this particular trip, but not on other factors such as how the state had evolved over time, or what actions the drivers had taken in the past.%
\footnote{When a mechanism is allowed to make dispatch decisions based on drivers' actions in the past, the mechanism may easily align incentives by no longer sending any trip offers to a driver who had declined a dispatch, for example. %
}
Similarly, we focus on driver strategies that %
depend on the queue length and a driver's position in the queue, and we denote a \emph{strategy} as a tuple $\sigma = (\alpha, \beta, \gamma)$.  For any queue length $\Queue \geq 0$, and at any position $\queue \in [0, \Queue]$ in the queue, 
\begin{enumerate}[(i)]
	\item $\alpha(\queue, \Queue, i) \in [0, 1]$ for each location $i \in \loc$ is the probability with which the driver at position $\queue$ in the queue accepts the trip dispatch %
	if she is dispatched a trip to location $i$,
	\item $\beta(\queue, \Queue)  \in [0, 1]$ determines the probability with which the driver  at position $\queue$ in the queue re-joins the queue at the tail (by going offline and online again, for example), and 
	\item $\gamma(\queue, \Queue) \in [0, 1]$ is the probability for the driver at %
	$\queue$ to leave the queue without a rider. 
	\end{enumerate}

Let $\Util(\queue, \Queue, \sigma, \sigma')$ denote the random variable representing the \emph{continuation payoff} of the driver at position $\queue$ in the queue, when the current length of the queue is $\Queue$, when this driver adopts strategy $\sigma$, and when all other drivers employ strategy $\sigma'$ (including those drivers who will arrive in the future). This includes the net earnings from the trip the driver may complete in the future, minus the total opportunity cost she incurs from this point of time onward waiting in the queue. Denote $\pi(\queue, \Queue, \sigma, \sigma') \triangleq \E{\Util(\queue, \Queue, \sigma, \sigma')}$ as the driver's \emph{expected continuation payoff} from position $\queue$ onward, where the expectation is taken over randomness in both the mechanism's decisions and the strategies of drivers. $\pi(\Queue, \Queue, \sigma, \sigma')$ thus represents the expected payoff of a driver with strategy $\sigma$, who joins the queue when the queue length is $\Queue$, and when all other drivers %
employ strategy $\sigma'$.

\medskip

We define the following properties. \todo{Consider defining an economy.}

\begin{definition}[Subgame-perfect equilibrium] %
\label{defn:spic}
A strategy $\sigmast$ forms a \emph{subgame perfect equilibrium} (SPE) among drivers under a mechanism if for any economy and any feasible strategy $\sigma$, %
\begin{align}
	\pi(\queue, \Queue, \sigmast, \sigmast) \geq \pi(\queue, \Queue, \sigma, \sigmast) , ~ \forall \Queue \geq 0, ~\forall \queue \in [0, \Queue]. %
	\label{equ:defn_spe}
\end{align}
\end{definition}

\begin{definition}[Individual rationality] \label{defn:IR_in_SPE}
A mechanism is \emph{%
individually rational in SPE}  
if under a strategy $\sigmast$ that forms an SPE among drivers, for any economy, 
\begin{align}
	\pi(\queue, \Queue, \sigmast, \sigmast) \geq 0, ~\forall \Queue \geq 0, ~\forall \queue \in [0, \Queue]. \label{equ:exante_IRSPE}
\end{align}
\end{definition}

\begin{definition}[Envy-freeness] \label{defn:EF_in_SPE}
A mechanism is \emph{%
envy-free in SPE} if %
under a strategy $\sigmast$ that forms an SPE among drivers, for any economy, 
\begin{align}
	\pi(\queue_1, \Queue, \sigmast, \sigmast) \geq \pi(\queue_2, \Queue, \sigmast, \sigmast) , ~\forall \Queue \geq 0, ~\forall \queue_1, \queue_2 \in [0, \Queue] \txtst \queue_1 \leq \queue_2. \label{equ:exante_EFSPE}
\end{align}
\end{definition}

Intuitively, under a mechanism that is individually rational and envy-free in SPE, a driver anywhere in the queue always gets non-negative continuation payoff, and does not envy the expected continuation earnings of any driver who is further down the queue.%

Given a mechanism $\mech$ %
and a strategy $\sigmast$ that forms an SPE under $\mech$, let $\equQL$ denote the length of the queue under $\sigmast$ in steady state. This is the case if the number of drivers joining the queue per unit of time is equal to the number of drivers dispatched from the queue when (i) the length of the queue is $\equQL$ and (ii) all drivers adopt strategy $\sigmast$. %
Moreover, let $\completion_i(\sigmast)$ denote the fraction of trips to location $i \in \loc$ that are completed  in steady state when all drivers adopt $\sigmast$.

The \emph{trip throughput} of mechanism $\mech$ is the amount %
of trips completed per unit of time under $\sigmast$ in steady state:
\begin{align}
	\tp_\mech(\sigmast) \triangleq \sum_{i \in \loc} \completion_i(\sigmast) \mu_i. \label{equ:trip_throughput}
\end{align}

The \emph{net revenue} achieved by mechanism $\mech$ is the total net earnings all drivers made from trips per unit of time under $\sigmast$ in steady state, minus the total opportunity costs the platform incurs due to drivers' waiting in the queue:
\begin{align}
	\rev_\mech(\sigmast) \triangleq %
	\sum_{i \in \loc} \completion_i(\sigmast) \mu_i \gb_i - \equQL  \mechCost. \label{equ:net_revenue}
\end{align} 
When $\mechCost = \cost$, the net revenue of the platform is $\rev_\mech(\sigmast) = \sum_{i \in \loc} \completion_i(\sigmast) \mu_i \gb_i - \equQL  \cost$, %
i.e. the total net payoffs to all drivers who arrive at the queue.

The objective of a mechanism is to optimize trip throughput and net revenue achieved %
in equilibrium in steady state.\footnote{When the platform takes as commission a fixed fraction of the earnings made by the drivers (%
from the queue %
as well as from driving elsewhere in the city), the problem of maximizing a platform's total commission is equivalent to that of maximizing the net revenue as defined in \eqref{equ:net_revenue}.}
We say a mechanism is \emph{optimal} if in equilibrium in steady state (i) the mechanism achieves the first best trip throughput, and (ii) the mechanism achieves the \emph{second best net revenue} i.e., the highest steady sate equilibrium net revenue that is achievable by any dispatching mechanism that is flexible and transparent, provides trip information to drivers upfront, and does not penalize drivers for any actions they take.%
\footnote{As we shall see later in this section, %
a platform may not be able to achieve the first best net revenue in certain settings, despite achieving the first best trip throughput. This is the case when a mechanism completes the same set of trips as those under the first best outcome, but drivers' strategically waiting for higher earning trips leads to a non-zero equilibrium queue length, thus increasing opportunity cost and reducing the net revenue of the platform.}

\subsection{Optimality of Direct FIFO} \label{sec:direct_FIFO_optimality}

\begin{definition}[Direct FIFO] \label{defn:direct_FIFO} Under \emph{the direct FIFO mechanism},  trips to each location $i \in \loc$ are dispatched in a FIFO manner to drivers starting from position $\numSkip_i$  (as defined in \eqref{equ:first_accpet_positions}) in the queue, when the length of the queue is $\Queue \geq \numSkip_i$. When $\Queue < \numSkip_i$, trips to location $i$ are not dispatched. %
\end{definition}

Under the direct FIFO mechanism,  the highest earning trips to location~$1$ are dispatched to the head of the queue, where the driver have waited for the longest time (thus have incurred the highest waiting costs). For a trip to location $i > 1$, %
the mechanism skips drivers close to the head of the queue who will be unwilling to accept, and dispatches the trip starting from the $\numSkip_i\th$ position--- the first position in the queue where the driver is willing to accept a trip to location $i$ under strict FIFO dispatching assuming infinite rider patience. 
The following theorem proves that this option to ``skip the rest of the line'' incentivizes drivers to accept all dispatches they receive.

\begin{restatable}[Incentive compatibility of direct FIFO]{theorem}{thmFIFOSkipSPE} \label{thm:fifo_skip_spe} 
It is a subgame-perfect equilibrium for drivers to accept all dispatches from the direct FIFO mechanism, and to join the queue if and only if the length of the queue is at most 
\begin{align}
	\maxQL \triangleq \numSkip_\numLoc + \frac{\gb_\numLoc}{\cost} \sum_{i \in \loc} \mu_i.
 \label{equ:max_eq_queue_length}
\end{align}
Moreover, the equilibrium outcome is individually rational and envy-free. 
\end{restatable}

The proof is via induction on queue positions, and is provided in Appendix~\ref{appx:proof_direct_FIFO}. %
Intuitively, this is a ``direct implementation'' of the equilibrium outcome under strict FIFO dispatching when riders have infinite patience (see Lemma~\ref{lem:strict_fifo_spe_formal}). Trips are dispatched starting from the positions in the queue where the drivers are indifferent towards accepting the trip or continuing to wait, and the equilibrium payoff from joining the queue is non-negative when the queue length is at most $\maxQL$.

\medskip

When there are more drivers than needed to complete all trips to location~$1$, %
the direct FIFO mechanism does not achieve the first best net revenue--- drivers are willing to spend time waiting for trips with higher earnings,  leading to a queue of non-zero length and lowering the net revenue of the platform. This kind of ``strategic waiting'' is, however, not avoidable. %
We prove in the following theorem that the outcome under direct FIFO %
achieves the \emph{second best net revenue}, %
i.e. the highest equilibrium net revenue achievable in steady state by any dispatching mechanism that provides trip destinations upfront and does not penalize drivers for declining dispatches.

\begin{restatable}[Optimality of direct FIFO]{theorem}{thmFIFOSB} \label{thm:fifo_skip_second_best} %
For every economy, the direct FIFO mechanism achieves in SPE the first best %
trip throughput. Moreover, the equilibrium outcome achieves the first best net revenue when $\mechCost = 0$, and the second best net revenue when $\mechCost \in (0, \cost]$. 
\end{restatable}

\begin{figure}[t!]
\centering
\begin{subfigure}[t]{\textwidth}
\centering
\begin{tikzpicture}[scale = 0.8] %

\draw[-, line width=0.25mm] (-0,0) -- (15, 0); %

\draw(-1.8, -0.25) node[anchor=south] {Riders};

\draw[-, line width=0.5mm] 	(0,0.1) -- (0,-0.1);
\draw(0, 1.2) node[anchor=south] {Head of the queue};
\draw(0, \lamdmarkHeight) node[anchor=south] {$\numSkip_1$};
\draw[->] (0, \lamdmarkHeight - 0.1) -- (0, 0.15);

\draw[-, line width=0.5mm] 	(4,0.1) -- (4,-0.1);
\draw(4, \lamdmarkHeight) node[anchor=south] {$\numSkip_2$};
\draw[->](4, \lamdmarkHeight - 0.1) -- (4, 0.15);

\draw[-, line width=0.5mm] 	(7,0.1) -- (7,-0.1);
\draw(7, \lamdmarkHeight) node[anchor=south] {$\numSkip_3$};
\draw[->](7, \lamdmarkHeight - 0.1) -- (7, 0.15);    
  
\draw(9, 0.6) node[anchor=south] {$\dots$};    

\draw[-, line width=0.5mm] 	(11,0.1) -- (11,-0.1);

\draw[-, line width=0.5mm] 	(15,0.1) -- (15,-0.1);

\draw(11, \lamdmarkHeight) node[anchor=south] {$\numSkip_{\maxJ- 1}$};
\draw[->](11, \lamdmarkHeight - 0.1) -- (11, 0.15);    

\draw(14.5, 1.2) node[anchor=south] {Tail of the queue};
\draw(15, \lamdmarkHeight) node[anchor=south] {$\numSkip_{\maxJ}$};
\draw[->](15, \lamdmarkHeight - 0.1) -- (15, 0.15);    

\end{tikzpicture}
\caption{The under-supplied scenario, with $\lambda \leq \sum_{i \in \loc} \mu_i$.
\label{fig:direct_FIFO_equ_under_supplied}}
\end{subfigure}%

\medskip

\begin{subfigure}[t]{\textwidth}
\centering
\begin{tikzpicture}[scale = 0.8] %

\draw[-, line width=0.25mm] (-0,0) -- (15, 0); %

\draw(-1.8, -0.25) node[anchor=south] {Riders};

\draw[-, line width=0.5mm] 	(0,0.1) -- (0,-0.1);
\draw(0, 1.3) node[anchor=south] {Head of the queue};
\draw(0, \lamdmarkHeight) node[anchor=south] {$\numSkip_1$};
\draw[->] (0, \lamdmarkHeight - 0.1) -- (0, 0.15);

\draw[-, line width=0.5mm] 	(3,0.1) -- (3,-0.1);
\draw(3, \lamdmarkHeight) node[anchor=south] {$\numSkip_2$};
\draw[->](3, \lamdmarkHeight - 0.1) -- (3, 0.15);

\draw[-, line width=0.5mm] 	(5.5,0.1) -- (5.5,-0.1);
\draw(5.5, \lamdmarkHeight) node[anchor=south] {$\numSkip_3$};
\draw[->](5.5, \lamdmarkHeight - 0.1) -- (5.5, 0.15);    
  
\draw(7.5, 0.6) node[anchor=south] {$\dots$};    

\draw[-, line width=0.5mm] 	(9.5,0.1) -- (9.5,-0.1);

\draw(9.5, \lamdmarkHeight) node[anchor=south] {$\numSkip_\ell$};
\draw[->](9.5, \lamdmarkHeight - 0.1) -- (9.5, 0.15);    

\draw(14.5, 1.3) node[anchor=south] {Tail of the queue};
\draw[-, line width=0.5mm] 	(15,0.1) -- (15,-0.1);
\draw(15, \lamdmarkHeight - 0.1) node[anchor=south] {$\maxQL$};
\draw[->](15, \lamdmarkHeight - 0.1) -- (15, 0.15);    

\end{tikzpicture}
\caption{The over-supplied scenario, with $\lambda > \sum_{i \in \loc} \mu_i$. \label{fig:direct_FIFO_equ_over_supplied}}
\end{subfigure}%
\caption{The steady-state equilibrium outcome under the direct FIFO mechanism. \label{fig:direct_FIFO_equ}} 
\end{figure}
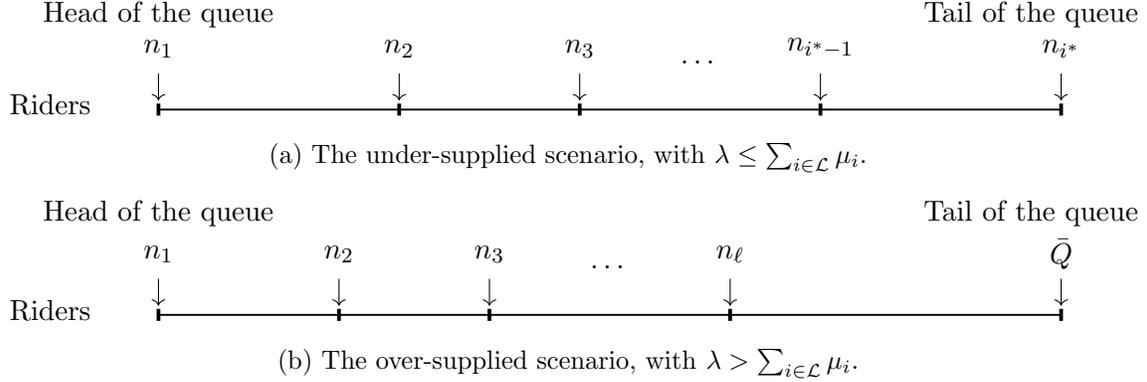

We prove this theorem in Appendix~\ref{appx:proof_direct_FIFO}. Briefly, we first show that the steady state equilibrium outcome under direct FIFO is as illustrated in Figure~\ref{fig:direct_FIFO_equ}. 
When $\lambda \leq \sum_{i\in \loc}\mu_i$, $\maxJ$ as defined in \eqref{equ:maxJ} is the lowest-earning trip that is (partially) completed in equilibrium. Drivers will line up for trips to locations $j < \maxJ$ (which have higher earnings than $\gb_\maxJ$), but the equilibrium queue length is $\equQL = \numSkip_\maxJ$ and there is no wait for a trip to location $\maxJ$. See Figure~\ref{fig:direct_FIFO_equ_under_supplied}. Every driver gets a total payoff of $\gb_\maxJ$ regardless of which trip they take, and all trips that are completed under the first-best outcome are completed.
When $\lambda > \sum_{i\in \loc}\mu_i$, the queue is over-supplied such that all trips are accepted and completed, and the equilibrium queue length is $\equQL = \maxQL$ (see Figure~\ref{fig:direct_FIFO_equ_over_supplied}). At this point, the drivers are indifferent between joining the queue and leaving, and all drivers get a zero payoff.

Given that all trips completed under the first best outcome are completed, direct FIFO achieves the first best trip throughput, and also the first best net revenue if $\mechCost = 0$ (i.e. when the platform does not incur any opportunity cost due to drivers' waiting in the queue).
To prove that the direct FIFO mechanism achieves the second best net revenue when $\mechCost > 0$, we first establish that no mechanism can achieve in equilibrium a strictly higher total payoff for all drivers combined. %
This implies that if the same set of trips are completed, reducing the equilibrium queue length in comparison to that under direct FIFO (thereby reducing the total opportunity costs for the drivers as well as the platform) is not possible. %
We then prove that completing a different set of trips in return for a shorter queue cannot be an improvement.

We now revisit the economy analyzed in Example~\ref{exmp:model_and_fb} in Section~\ref{sec:preliminaries}. 

\addtocounter{example}{-1}
\begin{example}[Continued] \label{exmp:outcome_direct_FIFO} 
Consider the economy in Example~\ref{exmp:model_and_fb}, %
for which strict FIFO dispatching achieves trip throughput $\tp\strict = 1$ and net revenue $\rev\strict = 0$. 
Under direct FIFO, trips to each location will be dispatched to drivers in the queue starting at positions $\numSkip_1 = 0$, $\numSkip_2 = 150$, and $\numSkip_3 = 360$, respectively. With $\lambda = 5$, $\mu_1 = 1$ and $\mu_2 = 6$, the lowest earning trip accepted in equilibrium will be trips to location $\maxJ = 2$, and the steady state queue length is $\equQL = \numSkip_2 = 150$. %

Upon arrival at the tail of the queue at $\numSkip_2$, one unit of driver moves on to wait for trips to location~$1$, and the remaining $4$ units of drivers immediately accept trips to location~$2$ and leave.
In equilibrium, all drivers get the same total payoff of $\gb_2 = 25$. %
The platform achieves a trip throughput of $\tp \direct = 5$ and a net revenue of $\rev\direct = 1\cdot \gb_1 + 4 \cdot \gb_2 - \equQL \mechCost = 125$ per unit of time. Since $\cost = \mechCost$ this net revenue is also the total payoff for all drivers combined.  
\qed
\end{example}

In comparison to strict FIFO dispatching, the direct FIFO mechanism substantially improves driver earnings, trip throughput, and the net revenue of the platform. The mechanism is, however, \emph{not fair} because even when all drivers are non-strategic and accept all dispatches from the platform, a driver closer to the head of the queue may still receive trips to certain destinations at a lower rate than drivers further down the queue. %
Take the Midway airport as an example. A driver who has waited long enough in the queue will never receive a trip back to downtown Chicago again--- as we can see from Figure~\ref{fig:heatmap_midway_average_fare}, all high-earning trips direct FIFO dispatches to her will be heading to the suburbs. %
This %
renders the direct FIFO mechanism ill-suited for practice.

\section{The Randomized FIFO Mechanism} \label{sec:rand_fifo}

In this section, we introduce a family of \emph{randomized FIFO mechanisms}, which achieve optimal equilibrium throughput and revenue without using unfair dispatch rules--- when drivers are straightforward and accept all dispatches, a driver closer to the head of the queue receives trip dispatches to \emph{every} destination at a (weakly) higher rate than any driver further down the queue.  

\medskip

To demonstrate the effectiveness of randomization for aligning incentives, we first analyze the steady state Nash equilibrium under \emph{random dispatching}, where every trip request is simply dispatched to drivers in the queue uniformly at random.%
\footnote{%
For simplicity of analysis, we work in this section with Nash equilibrium in steady state because drivers' equilibrium strategy depends on the length of the queue when dispatches are randomized.}

\begin{definition}[Nash equilibrium in steady state] A strategy $\sigmast$ forms a Nash equilibrium among drivers in steady state under a mechanism if there exists a queue length $\equQL \geq 0$ such that
\begin{enumerate}[(i)]
	\item for any feasible strategy $\sigma$ and any position in the queue $\queue \in [0, \equQL]$, 
	\begin{align}
		\pi(\queue, \equQL, \sigmast, \sigmast) \geq \pi(\queue, \equQL, \sigma, \sigmast), \label{equ:Nash}
	\end{align}
	\item when all drivers adopt strategy $\sigmast$, the steady state queue length is $\equQL$.
\end{enumerate}
\end{definition}

\begin{restatable}[Optimality of random dispatching]{proposition}{propPureRandOpt} \label{prop:pure_random}
In Nash equilibrium in steady sate, dispatching every trip to all drivers in the queue uniformly at random achieves  the first best trip throughput and the second best net revenue. When $\mechCost = 0$, the equilibrium net revenue is also the first best. 
\end{restatable}

See Appendix~\ref{appx:proof_pure_random} for the proof of this result. 
Briefly, we show that under random dispatching, every driver in the queue is willing to accept a trip to location $\maxJ$ (the lowest earning trip accepted in equilibrium under direct FIFO) despite the fact that the drivers may still receive higher-earning trips %
later. This is different from the outcome under strict FIFO, because in comparison to the driver at the head of the queue under strict FIFO, a driver who declines a dispatch under random dispatching will need to wait for a much longer time to receive her next dispatch.

This additional waiting time introduced by randomization
increases drivers' costs of cherry-picking, and allows random dispatching to align incentives %
without deprioritizing drivers at the head of the queue when dispatching trips to any location.
Naively randomizing among all drivers in the queue, however, introduces substantial uncertainty in drivers' waiting times. This contributes to the variability in drivers' total payoffs, on top of the variability in the net earnings %
from trips. This is in stark contrast to direct FIFO, which matches lower-earning trips with drivers who have waited less time in the queue, thereby reducing the variation in drivers' total payoffs.

\medskip

\emph{The randomized FIFO mechanisms} we now introduce make proper use of drivers' waiting times in the queue in both ways. By gradually sending declined trips down the queue in a randomized manner, a randomized FIFO mechanism aligns incentives, and also guarantees that drivers in earlier segments in the queue (who have incurred higher waiting costs) will take trips with higher earnings.

\begin{definition}[Randomized FIFO] \label{defn:randomized_FIFO} A \emph{randomized FIFO mechanism} is specified by $\numBins \geq 1$ \emph{bins} in the queue $([\binLB\1, \binUB\1], ~[\binLB\2, \binUB\2], \dots, [\binLB \supPar{\numBins}, ~ \binUB\supPar{\numBins}])$.
For the $k\th$ time a trip is dispatched, the mechanism dispatches the trip to a driver in the $k\th$ bin $[\binLB\supk, \binUB\supk]$ uniformly at random. 
\end{definition}

See Figure~\ref{fig:randomized_FIFO} for an illustration. Intuitively, all rider requests are first dispatched to drivers in the first bin $[\binLB\1, \binUB\1]$ uniformly at random. If a dispatch is declined, the mechanism will then dispatch the trip to a random driver in the next bin. %

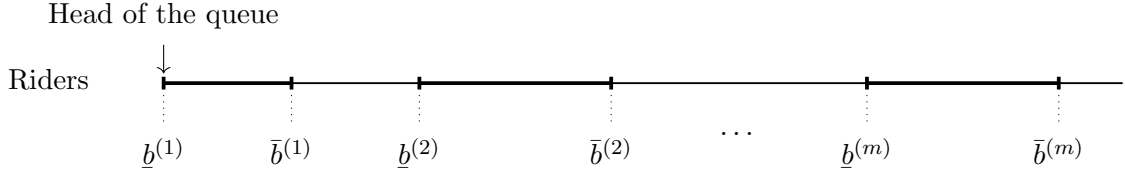
\begin{figure}[t!]
\centering
\begin{tikzpicture}[scale = 0.85] %

\draw[-, line width=0.25mm] (-0,0) -- (15, 0); %

\draw(-1.8, -0.25) node[anchor=south] {Riders};

\draw[-, line width=0.5mm] 	(0,0.1) -- (0,-0.1);
\draw(0,  \lamdmarkHeight) node[anchor=south] {Head of the queue};
\draw[->] (0, \lamdmarkHeight - 0.1) -- (0, 0.15);

\draw[-, line width=0.5 mm] 	(2,0.1) -- (2,-0.1);
\draw[-, line width = 0.5 mm] 	(0, 0) -- (2, 0);
\draw[dotted](0, -\lamdmarkHeight + 0.1) -- (0, 0-.15);  
\draw[dotted](2, -\lamdmarkHeight + 0.1) -- (2, 0-.15);  
\draw(0, - \lamdmarkHeight) node[anchor = north] {$\binLB\1$};
\draw(2, - \lamdmarkHeight) node[anchor = north] {$\binUB\1$};

\draw[-, line width = 0.5 mm] 	(4, 0) -- (7, 0);
\draw[-, line width=0.5mm] 	(4,0.1) -- (4,-0.1);
\draw[-, line width=0.5mm] 	(7,0.1) -- (7,-0.1);
\draw[dotted](4, -\lamdmarkHeight + 0.1) -- (4, 0-.15);    
\draw[dotted](7, -\lamdmarkHeight + 0.1) -- (7, 0-.15);    
\draw(4, -\lamdmarkHeight) node[anchor=north] {$\binLB\2$};
\draw(7, - \lamdmarkHeight) node[anchor = north] {$\binUB\2$};

\draw(9,- 0.6) node[anchor=north] {$\dots$};    

\draw[-, line width = 0.5 mm] 	(11, 0) -- (14, 0);
\draw[-, line width=0.5mm] 	(11,0.1) -- (11, -0.1);
\draw[-, line width=0.5mm] 	(14,0.1) -- (14, -0.1);
\draw[dotted](11, -\lamdmarkHeight + 0.1) -- (11, 0-.15);    
\draw[dotted](14, -\lamdmarkHeight + 0.1) -- (14, 0-.15);    
\draw(11, -\lamdmarkHeight) node[anchor=north] {$\binLB\supK$};
\draw(14, - \lamdmarkHeight) node[anchor = north] {$\binUB\supK$};

\end{tikzpicture}

\caption{Illustration of a randomized FIFO mechanism. \label{fig:randomized_FIFO} }
\end{figure}

With a rider patience level of $\patience$, each trip may be dispatched a maximum of $\patience$ times before the rider cancels her request. 
Recall that $\maxJ$ as defined in \eqref{equ:maxJ} is the lowest-earning trip that is (partially) completed under the first best outcome. Given any economy, let $(\loc\1, \loc\2, \dots, \loc^{(\numBins)})$ for some $\numBins \leq \min\{ \maxJ, ~\patience\}$ be an \emph{ordered partition} of the top $\maxJ$ destinations $\{1, 2, \dots, \maxJ \} \subseteq \loc$, i.e.
\begin{enumerate}[(i)]
    \item (collectively exhaustive) $\bigcup_{k=1}^{\numBins} \loc\supk = \{1, 2, \dots, \maxJ \}$, and for each $k = 1, \dots, \numBins$, %
    $\loc\supk \neq \emptyset$,
	\item (mutually exclusive) for all $k_1, k_2 \leq \numBins$ s.t. $k_1 \neq  k_2$, $\loc^{(k_1)} \cap \loc^{(k_2)} = \emptyset$, %
	\item (monotone) \label{enumitem:ordered} for all $k_1, k_2 \leq \numBins$ s.t. $k_1 < k_2$, we have $ i < j $ for all $i \in \loc^{(k_1)}$ and all $j  \in \loc^{(k_2)}$. 
\end{enumerate}

Condition (\ref{enumitem:ordered}) %
requires that trips in an earlier partition have higher earnings than those in a later partition. Given an economy and any ordered partition $(\loc\1, \dots, \loc^{(\numBins)})$ of the top $\maxJ$ destinations, we construct a corresponding set of $\numBins$ bins in the queue $([\binLB\1, \binUB\1], %
\dots, [\binLB \supPar{\numBins}, ~ \binUB\supPar{\numBins}])$ as follows: 
\begin{align}
	\binLB\supk &  \triangleq  \sum_{i \in \cup_{k' < k} \loc\supkprime }  \left( \gb_i-  \min_{i' \in \loc\supk} \{\gb_{i'}\} \right)  \mu_i  / \cost, \label{equ:defn_of_binLB} \\
	\binUB\supk & \triangleq \sum_{i \in \cup_{k' \leq k} \loc\supkprime} \left(\gb_i - \min_{i' \in \loc\supk} \{ \gb_{i'} \} \right)  \mu_i/ \cost. \label{equ:defn_of_binUB}
\end{align}

In Lemma~\ref{lem:bin_properties} in Appendix~\ref{appx:proof_thm_randFIFO_second_best} we show that $\binLB\1 = 0$, $\binUB\supk \geq \binLB\supk \geq 0$ for each $k \leq \numBins$, and $\binLB\supkpo \geq \binUB\supk$ for all $ k \leq \numBins - 1$. This guarantees that the bins start from the head of the queue, are well defined, and do not overlap with each other. 

\medskip

We now present the main result of this paper, that the family of randomized FIFO mechanisms constructed in this way achieves the optimal steady state outcome in Nash equilibrium.

\begin{restatable}[Optimality of randomized FIFO]{theorem}{thmRandomFIFOSB} \label{thm:randFIFO_second_best}
For any economy and any ordered partition of the top $\maxJ$ destinations $(\loc\1, \dots, \loc^{(\numBins)})$ with $\numBins \leq \min \{\maxJ, \patience\}$,  a randomized FIFO mechanism corresponding to \eqref{equ:defn_of_binLB} and \eqref{equ:defn_of_binUB} %
achieves the first best trip throughput and the second best net revenue in Nash equilibrium in steady state. %
When $\mechCost = 0$, the net revenue is also the first best. %
\end{restatable}

We provide the proof of this theorem in Appendix~\ref{appx:proof_thm_randFIFO_second_best}. 
Briefly, %
let $\equQL$ denote the steady state equilibrium queue length under the direct FIFO mechanism. %
We first show that under a randomized FIFO mechanism, %
it is a Nash equilibrium in steady state that (i) all drivers in the $k\th$ bin %
in the queue accept all and only trips in the top $k$ partitions $\cup_{k' = 1}^k \loc \supkprime$ (with potential randomization over trips to location $\maxJ$), (ii)  after joining the queue, no driver %
leaves the queue without a rider trip, or rejoins the queue at its tail, (iii) drivers join the queue with probability $\min\{1,~\sum_{i \in \loc} \mu_i / \lambda\}$ upon arrival, and (iv) the length of the queue remains constant at $\equQL$.
Under this equilibrium outcome, %
all trips that are completed under the first best are also completed,  so that the randomized FIFO and direct FIFO mechanisms complete the same set of trips in steady sate.
The queue lengths being the same %
then implies that randomized FIFO also achieves the optimal net revenue, given the optimality of the direct FIFO mechanism we have established in Theorem~\ref{thm:fifo_skip_second_best}.

More formally, let $\pist(\queue) \triangleq \pi(\queue,\equQL,\sigmast,\sigmast)$ be the continuation payoff of a driver at position $\queue \in [0, \equQL]$ in the queue, when the queue length is $\equQL$ and when all drivers adopt the equilibrium strategy described above (which we denote as $\sigmast$).
We show that $\pist(\queue)$ is non-negative, piece-wise linear, and monotonically non-increasing in $\queue$. %
Moreover, $\pist(\queue) = \min_{i \in \loc\supk} \{\gb_i\}$ for all $\queue \in [\binLB\supk, \binUB\supk]$ for each $k \leq \numBins$, i.e. the continuation earning of any driver in the $k\th$ bin is equal to the net earning of the lowest earning trip in the $k\th$ partition. 
This allows us to prove by induction on $k$ that $\sigmast$ forms a Nash equilibrium. %
The non-negativity and monotonicity of %
$\pist(\queue)$ %
imply that the randomized FIFO mechanism is \emph{individually rational} and \emph{envy free in Nash equilibrium in steady state}. 

\medskip

We now %
demonstrate via the following example that the randomized FIFO mechanism %
substantially reduces the variability in drivers' earnings in comparison to dispatching %
every request to all drivers in the queue uniformly at random.

\begin{example} \label{exmp:running_example_rand_vs_rand_FIFO}
Consider an economy with three destinations $\loc = \{1,2,3\}$, rider arrival rates $\mu_1 = 1$, $\mu_2 = 6$, $\mu_3 = 3$, and net earnings from trips $\gb_1 = 75$, $\gb_2 = 25$, $\gb_3 = 15$. The opportunity costs per minute are $\cost = \mechCost = 1/3$, and drivers arrive at a rate of $\lambda = 8$ per minute. %
Moreover, assume for simplicity that riders have a patience level of $\patience = 2$, i.e. each trip can be dispatched only twice before riders cancel their requests.

\smallskip

\noindent{}\emph{Strict FIFO dispatching.} Trips to locations~$2$ and $3$ cannot reach drivers in the queue who are willing to accept them. $\tp\strict = 1$ as a result. Moreover, $\rev\strict = 0$ since queue is long enough that drivers get a payoff of zero regardless of whether they had joined the queue, and when $\mechCost = \cost$ the platform's net revenue is equal to the total payoff of all drivers. %

\smallskip

\noindent{}\emph{Random dispatching.} When the length of the queue is $\equQL = 360$, it is an equilibrium for drivers to accept all trips to locations $1$ and $2$, and randomize on trips to location~$3$. In steady state, all location $1$ and $2$ trips, and a third of location $3$ trips are completed. The throughput is $\tp\rand = 8$, and the platform achieves a net revenue of $\rev\rand = 120$ per minute. The average waiting time in the queue is $\equQL / \tp\rand = 45$ minutes, and the drivers get an average total payoff of  $15$. %
However, due to the high level of variability in (i) a driver's waiting time for a trip, and (ii) the net earnings from the trip a driver %
may accept, the variance of drivers' total payoffs is $500$ (see Appendix~\ref{appx:eq_under_pure_random} for the computation of the equilibrium outcome and earning variance).

\smallskip

\noindent{}\emph{Randomized FIFO.} Consider a randomized FIFO mechanism corresponding to the ordered partition $(\loc\1, \loc\2) = (\{1\}, \{2, 3 \})$. The corresponding bins are given by $\binLB\1 = \binUB\1 = 0$, $\binLB\2=180$, and $\binUB\2 = 360$. 
All trips are first sent to drivers in $[\binLB\1, \binUB\1] = \{0\}$, i.e. the head of the queue. In equilibrium, %
drivers at the head of the queue accept only trips to location~$1$.%
The remaining trips to locations~$2$ and $3$ will then be randomly dispatched to drivers at positions $180$ to $360$ in the queue.

In equilibrium, the length of the queue is $\equQL = 360$. Compared to random dispatching, the randomized FIFO mechanism achieves the same trip throughput, net revenue, average driver waiting time, and average driver payoff. %
In contrast, the variance of the total payoffs of the drivers is reduced from $500$ to $75$ (see Appendix~\ref{appx:eq_under_randomized_FIFO} for more details). 
By matching higher earning trips with drivers in earlier bins who have incurred a higher waiting cost, the randomized FIFO mechanism is able to substantially reduces the variability in drivers' total payoffs. %
\qed
\end{example}

A higher patience level of riders increases the number of times a trip can be dispatched before the rider cancels her request. This allows the randomized FIFO mechanisms to use a larger number of bins and better match higher-earning trips with drivers who have waited longer in the queue. When riders are sufficiently patient, the randomized FIFO mechanism is able to achieve zero variance in drivers' total payoffs.
Consider an economy where riders' patience level is higher than the number of trip types completed in equilibrium, %
i.e. when $\patience \geq \maxJ$. %
The randomized FIFO mechanism corresponding to $\numBins = \maxJ$ partitions has a single trip in each partition, and offers a trip to the driver at position $\numSkip_k$ in the queue if it is the $k\th$ time that the trip is dispatched.
In equilibrium, trips to each location $k \leq \maxJ$ are accepted by the drivers at $\numSkip_k$, and the equilibrium outcome %
is the same as that under direct FIFO, where all drivers have the same total payoff.

\subsection{Discussion} \label{sec:discussion}

In real-life systems with the richness of a ridesharing platform, there typically exist multiple notions of fairness. %
In the context of the airport queues, a platform could be %
perceived as not treating drivers fairly if some drivers receive much more lucrative trips than others after waiting a similar amount of time, or if drivers who arrived later in time have higher priority for trips to certain destinations. %

Under the randomized FIFO mechanisms, the small variance in drivers' total payoffs and the envy-freeness of the equilibrium outcome (which guarantees that no driver would want to swap positions with any other driver who joined the queue later in time) can both be considered as fairness properties for \emph{the equilibrium outcome} (see \citet{avi2004measuring}, \citet{platz2017curse}, and \citet{wierman2011fairness}). %
The direct FIFO mechanism achieves zero earning variance in theory, but severely violates what is typically required of a fair \emph{dispatch rule} since even when all drivers are straightforward and accept every dispatch, drivers closer to the head of the queue may still receive trips to certain destinations at a lower rate than drivers %
further down the queue. %

Under a randomized FIFO mechanism, trips are only dispatched to drivers closest to the head of the queue when all drivers are straightforward. With strategic drivers, in equilibrium, it is possible for drivers in later bins to receive certain low-earning trips at a higher rate than drivers in an earlier bin, after these trips are first dispatched to and declined by drivers in the earlier bins.%
\footnote{There may also be segments in the queue where the drivers do not receive any dispatches under the randomized FIFO mechanisms. It is possible for $\binUB \supkmo < \binLB \supk$, meaning that a driver who have just moved past the $k\th$ bin may need to wait for some time before she reaches the $k-1\th$ bin, and in the mean time will not receive any dispatches. The existence of such segments in between bins is to guarantee that a driver who is about to reach $\binLB\supk$ is not getting a continuation payoff that is too high such that the driver will decline the lower-earning trips in $\loc\supk$.}
As we have seen from the analysis of strict FIFO dispatching, improving efficiency and average driver earnings %
does require that %
lower-earning trips be quickly dispatched to
drivers further down the queue who are willing to accept them, before riders' patience runs out.\footnote{The trade-off between efficiency and fairness has been discussed in various contexts~\citep{su2004patient,su2006recipient}. In our setting, %
the very limited rider patience leads to substantially higher efficiency loss under strict FIFO dispatching.}

\medskip

In addition to the optimality of revenue and throughput, 
the proof of Theorem~\ref{thm:fifo_skip_second_best} also establishes that no mechanism can achieve a better total payoff for drivers, if drivers are provided trip %
details upfront as well as the flexibility to freely %
decline any dispatches. 
It is tempting to think that a mechanism that %
imposes penalties could easily achieve a better outcome. However, the %
same proof also implies that even if a mechanism is allowed %
to %
move drivers to the tail of the queue for declining trip dispatches, no such mechanism can achieve a better throughput, %
revenue, or driver payoff.%
\footnote{\citet{su2004patient} suggest that a mechanism could impose penalties such that patients who decline an organ offer would expect a decrease in their priority position. 
In today's ridesharing platforms, Uber and Lyft drivers may lose their positions in line and move back to the tail of the queue after declining (multiple) trip dispatches. See \url{https://help.lyft.com/hc/en-us/articles/115012922787-Receiving-Airport-FIFO-pickup-requests} and  \url{https://www.uber.com/us/en/drive/san-francisco/airports/san-francisco-international/}, accessed 09/27/2020. %
Such penalties can improve the outcome under %
strict FIFO dispatching.
}
Intuitively, all trips that are accepted under the first best outcome are also accepted under randomized FIFO. The only inefficiency arises because drivers %
strategically wait for better trips, which leads to a higher-than-necessary amount of driver supply being ``tied-up'' at the queue in equilibrium.
Reducing the length of the %
queue %
lowers the opportunity cost for the platform, but  %
also decreases drivers' cost of being moved to the tail of the queue at the same time. This renders such penalties less effective. %
All things considered, the randomized FIFO mechanisms achieve a desirable balance between flexibility, efficiency, fairness and variability in drivers' earnings.

\section{Simulation Results} \label{sec:simulations}

In this section, we present counterfactual simulation results for the Chicago O'Hare International Airport. As we vary the level of driver supply or rider patience, we compare various mechanisms and benchmarks in the equilibrium net revenue, trip throughput, queue length, and drivers' average waiting time, average earnings, and earning variance. Additional simulations for O'Hare and for the Chicago Midway International Airport are provided in Appendix~\ref{appx:additional_simulations}.

To estimate the distribution of trips and the net earnings from trips by destination, we make use of trip-level data from ridesharing platforms (including Uber and Lyft) made public by the City of Chicago.%
\footnote{\url{https://data.cityofchicago.org/Transportation/Transportation-Network-Providers-Trips/m6dm-c72p}, accessed December 12, 2020. This dataset contains all trips in Chicago from November 2018 onward, %
but we use data up to mid March of 2020, before the COVID-19 pandemic substantially changed the dynamics of the market. During this time of consideration, drivers in Chicago \emph{do not} have trip destinations up front.%
}
The dataset provides the fare for each trip (rounded to the nearest \$2.50), the origin and destination of each trip on the Census Tract level, as well as the timestamps at the beginning and the end of each trip (rounded to the nearest 15 minutes). There are a total of 801 census tracts within the City of Chicago, which we consider as the set of destinations.%
\footnote{\url{https://data.cityofchicago.org/Facilities-Geographic-Boundaries/Boundaries-Census-Tracts-2010/5jrd-6zik}, accessed September 14, 2020. The destination census tracts are not available for trips ending outside of the City of Chicago, and may also be hidden %
due to privacy considerations when trips are sparse. 
Overall, 42.6\% of trips originating from O'Hare do not have a destination census tract, thus we cannot take these trips into consideration for our simulations. %
We do not expect any qualitative change in the simulation results if trips to all destinations are included. In fact, incorporating the long trips to the suburbs with very high earnings (which are currently missing) will likely lead to a worse outcome under strict FIFO dispatching, since these trips provide strong incentives for drivers at the head of the queue to wait and cherry-pick.}

\begin{figure}[t!]
\centering
\begin{subfigure}[t]{\subfigWidthTwo \textwidth}
	\centering
    \includegraphics[height=\heatMapHeight in]{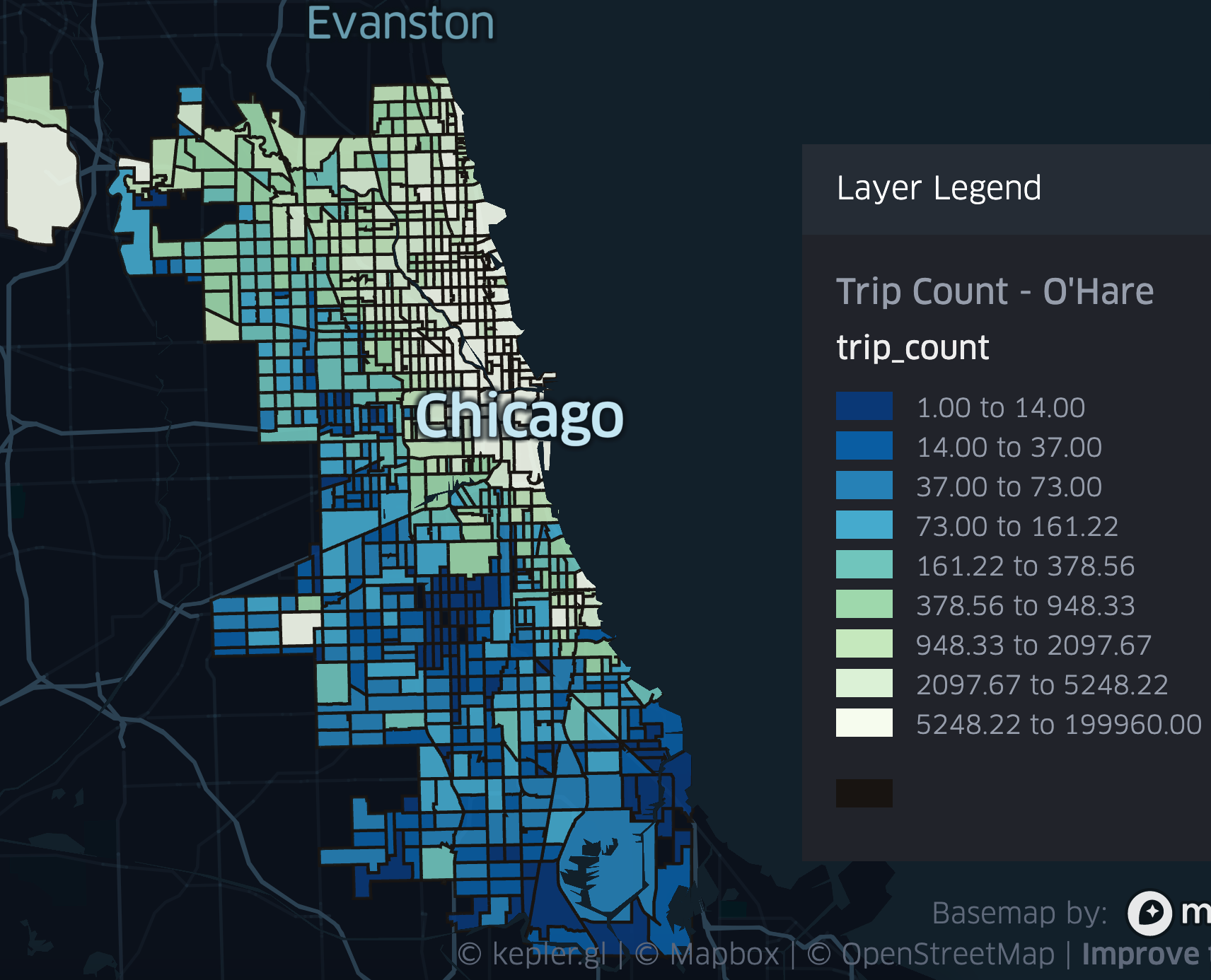}
    \caption{Trip count by destination. \label{fig:heatmap_ohare_trip_count}}
\end{subfigure}%
\begin{subfigure}[t]{\subfigWidthTwo \textwidth}
	\centering
    \includegraphics[height=\heatMapHeight in]{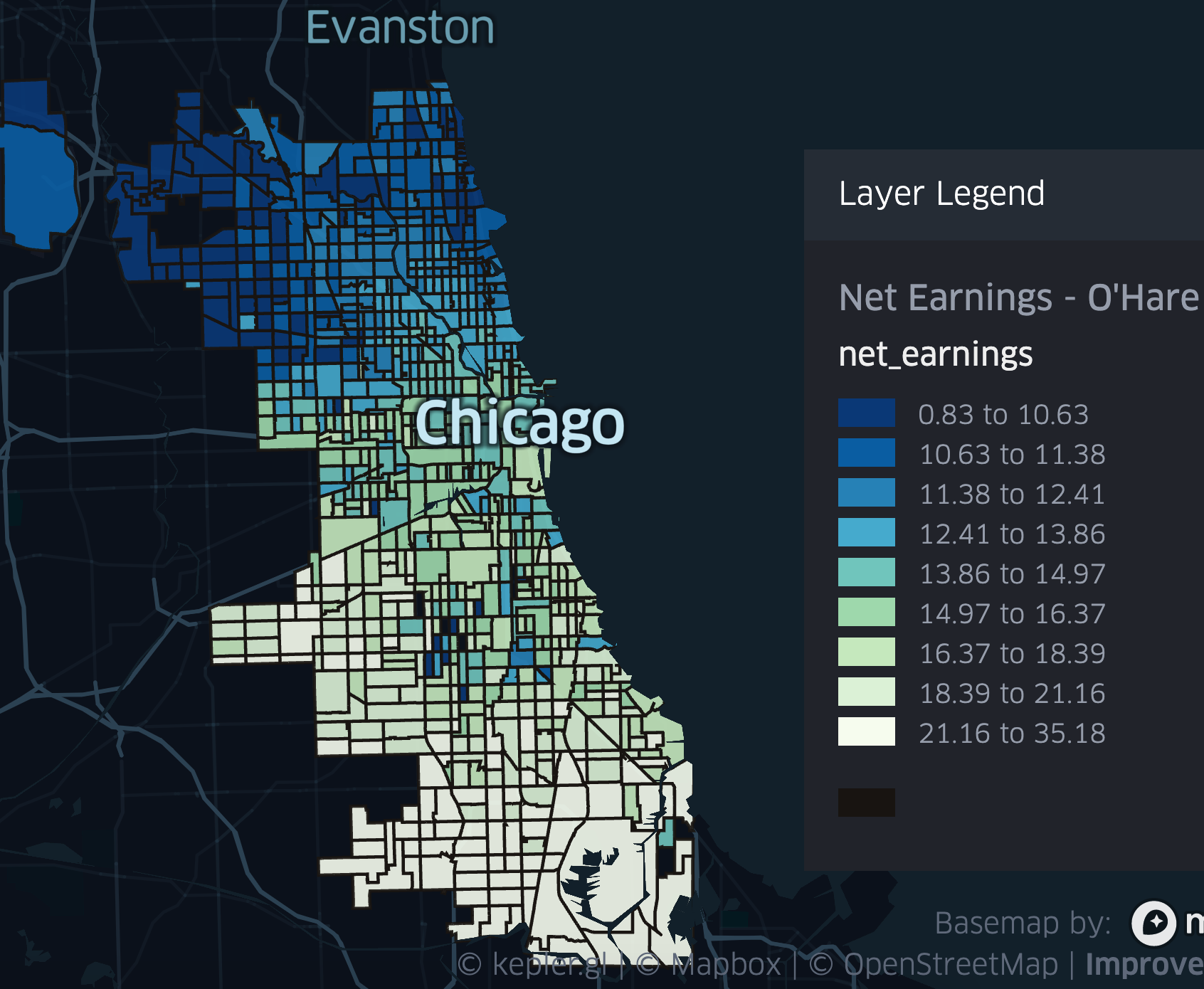}
    \caption{Net earnings by destination. \label{fig:heatmap_ohare_net_earnings}}
\end{subfigure}%
\caption{Trip volume and estimated net earnings (assuming $\cost = 1/3$) by destination Census Tract in Chicago, for trips originating from the Chicago O'Hare International Airport. %
\label{fig:heatmaps_from_ohare_1}}   
\end{figure}

From November 1, 2018 to March 11, 2020, there are a total of 4.53 million trips originating from O'Hare (see Figure~\ref{fig:ohare_daily_trips_from_airport} in Appendix~\ref{appx:additional_simulations_ohare}). The number of trips by destination census tract is as shown in Figure~\ref{fig:heatmap_ohare_trip_count}, and the average trip fare by destination is shown in Figure~\ref{fig:heatmap_ohare_average_fare} in Section~\ref{sec:intro}. Without driver identifiers, we are unable to estimate the average hourly earnings of a driver in Chicago. We assume throughout this section that the opportunity cost of a driver is $\cost = 1/3$, representing the scenario that an average driver driving in the city makes \$20 per hour.%
\footnote{The simulation results are not sensitive to the choice of $\cost$. A proposal from Uber in 2019 (see \url{https://p2a.co/H9gttWA}, accessed September 14, 2020) discussed ensuring drivers are paid an average of \$21 per hour while \emph{on trip}, the earnings per hour online could be currently slightly lower, depending on the average utilization level.} 
Combining the average fare, average trip duration (see Figure~\ref{fig:heatmap_ohare_duration}), and the opportunity cost, we estimate the net earnings by trip destination as shown in Figure~\ref{fig:heatmap_ohare_net_earnings}.%
\footnote{See Appendix~\ref{appx:net_earnings} for more details. Note that without driver identifiers, we are not able to appropriately estimate the continuation payoff of drivers after arriving at different destinations. As a result, the net earnings used in our simulations incorporate only payments from the immediate trip, effectively assuming that there is no heterogeneity in the continuation earnings from different locations onward.%
}

Throughout this section, we fix the total arrival rate of riders at $\sum_{i \in \loc} \mu_i = 12$ per minute. This is roughly equal to the rate of \emph{completed} trips during early evening hours on weekdays (see Figure~\ref{fig:ohare_how_trips_from_airport} in Appendix~\ref{appx:additional_simulations_ohare} for the average number of completed trips by hour-of-week). %
We assume that the platform's opportunity cost of drivers' time is $\mechCost = \cost = 1/3$ per minute, which corresponds to the scenario where %
the gap between the first best and the second best net revenue (which is achieved by the mechanisms we propose) is the largest. 
Finally, the randomized FIFO mechanism we evaluate in this section %
corresponds to an ordered partition of the set of completed trips into at most $\patience$ subsets, each containing (approximately) the same number of destinations. %

\paragraph{Varying Driver Supply.} 

We first compare the different mechanisms and benchmarks as the arrival rate of drivers $\lambda$ varies from zero to twenty percent over the total rider arrival rate. We fix the rider patience level at $\patience = 12$, representing the scenario where each driver decline takes $10$ seconds on average, and where riders are willing to wait for $2$ minutes for a match. %
Figure~\ref{fig:varying_lambda_ohare_1} presents the steady state net revenue, trip throughput, and queue length achieved in equilibrium.

When the arrival rate of drivers is very low, the outcome under direct FIFO, randomized FIFO, strict FIFO and random dispatching coincide, and all mechanisms achieve a net revenue very close to that under the first best outcome. This is because all drivers are able to accept trips with high earnings, and do not spend much time lining up in the queue. As the arrival rate of drivers increases, the length of the queue increases, and so does the gap between the first best  and the second best net revenue (which is achieved by direct FIFO, randomized FIFO, and random dispatching).

\newcommand{\figHeight}{1.66}
\newcommand{\subfigWidthThree}{0.33}

\begin{figure}[t!]
\centering
\begin{subfigure}[t]{\subfigWidthThree \textwidth}
	\centering
    \includegraphics[height= \figHeight in]{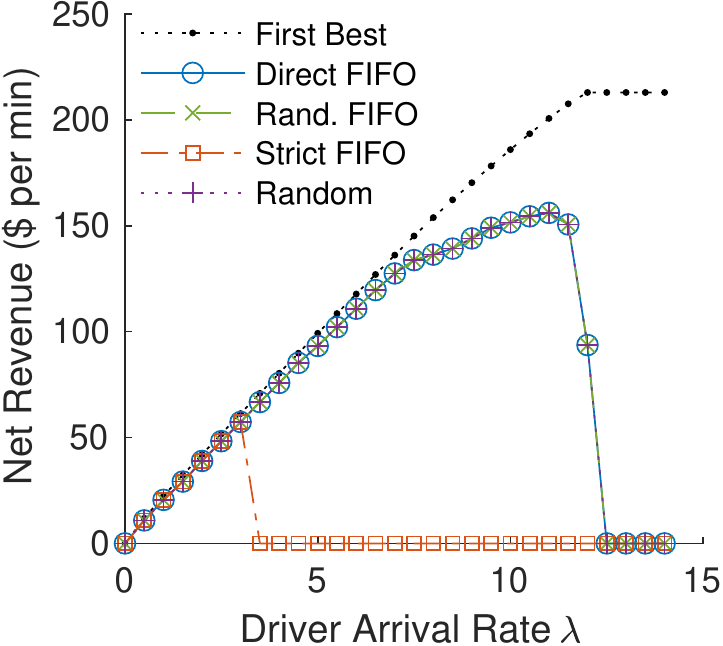}
    \caption{Net revenue.\label{fig:varying_m_gb_ohare}}
\end{subfigure}%
\begin{subfigure}[t]{\subfigWidthThree \textwidth}
	\centering
    \includegraphics[height=\figHeight in]{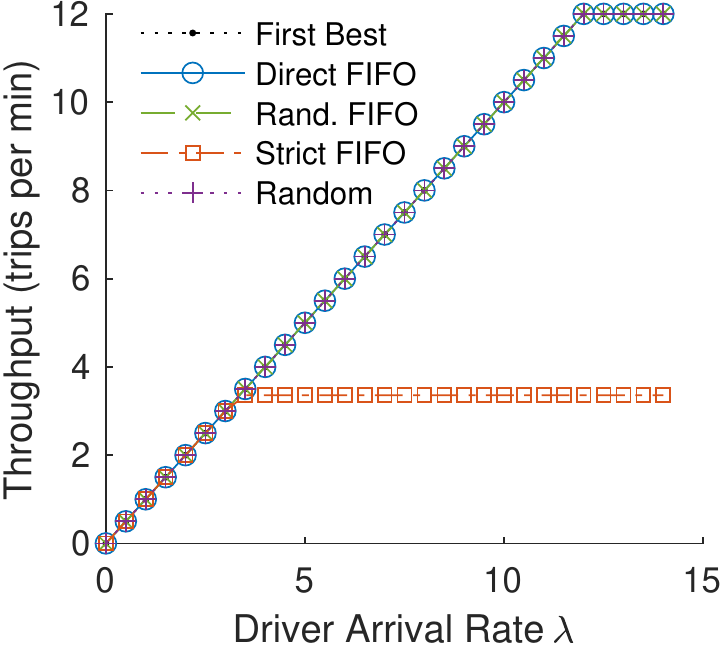}
    \caption{Trip throughput. \label{fig:varying_m_throughput_ohare}}
\end{subfigure}%
\begin{subfigure}[t]{\subfigWidthThree \textwidth}
	\centering
    \includegraphics[height=\figHeight in]{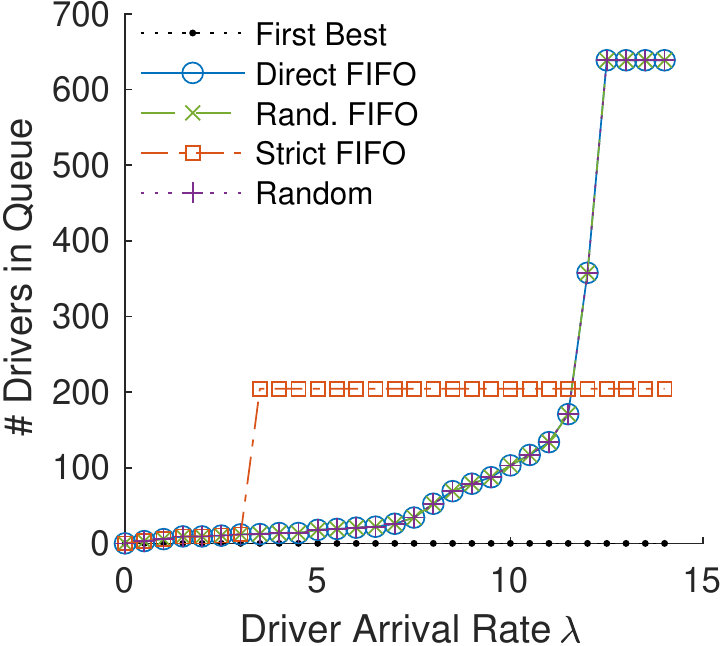}
    \caption{Equilibrium queue length. \label{fig:varying_m_max_q_length_ohare}}
\end{subfigure}%
\caption{Equilibrium net revenue, trip throughput, and length of the queue in steady state, as the arrival rate of drivers varies. Chicago O'Hare. \label{fig:varying_lambda_ohare_1}}    
\end{figure}

In contrast to the other mechanisms, the trip throughput under strict FIFO dispatching quickly plateaus despite the increasing driver supply, since rider requests for lower earning trips cannot reach drivers in the queue who are willing to accept them. These trips become unfulfilled, and at the same time, some drivers will have to deadhead back to the city without a rider. As a result, the net revenue under strict FIFO (which is equal to the total payoff of all drivers combined when $\mechCost = \cost$) drops to zero--- drivers will continue to join the queue until the queue is so long that joining is no better than leaving without a rider, thus in equilibrium all drivers %
get a zero total payoff.

Once the queue is over-supplied, i.e. when the driver arrival rate exceeds the total rider demand, the net revenue under %
all mechanisms %
drop to zero. This is inevitable, since no driver is willing to leave the airport without a rider as long as %
joining the queue and wait leads to a strictly positive payoff, but some driver has to deadhead in steady state. Nevertheless, %
we can see from Figure~\ref{fig:varying_m_ave_wait_ohare} that the average waiting time %
under randomized FIFO is still shorter than that under strict FIFO dispatching despite the longer queue length, since the trip throughout is substantially higher.

\begin{figure}[t!]
\centering
\begin{subfigure}[t]{\subfigWidthThree \textwidth}
	\centering
    \includegraphics[height=\figHeight in]{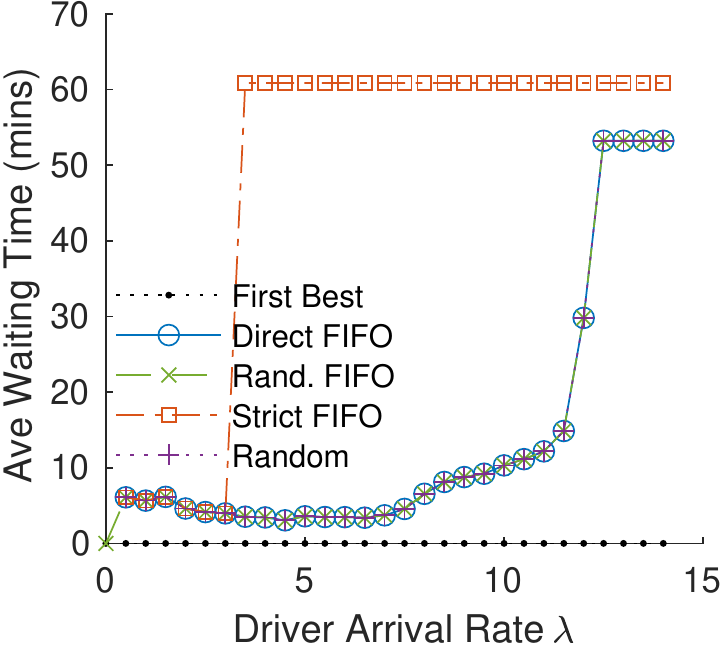}
    \caption{Average waiting time. \label{fig:varying_m_ave_wait_ohare}}
\end{subfigure}%
\begin{subfigure}[t]{\subfigWidthThree \textwidth}
	\centering
    \includegraphics[height=\figHeight in]{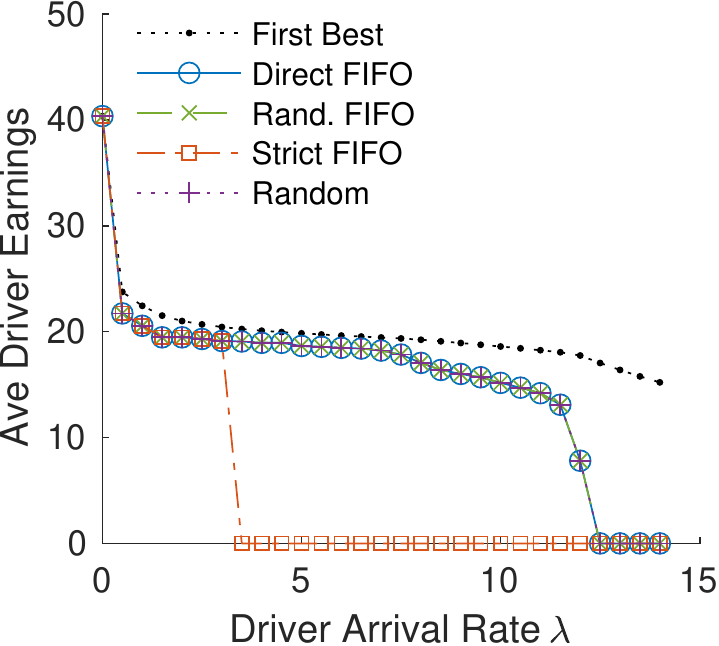}
    \caption{Average driver payoff. \label{fig:varying_m_ave_driver_earning_ohare}}
\end{subfigure}%
\begin{subfigure}[t]{\subfigWidthThree \textwidth}
	\centering
    \includegraphics[height=\figHeight in]{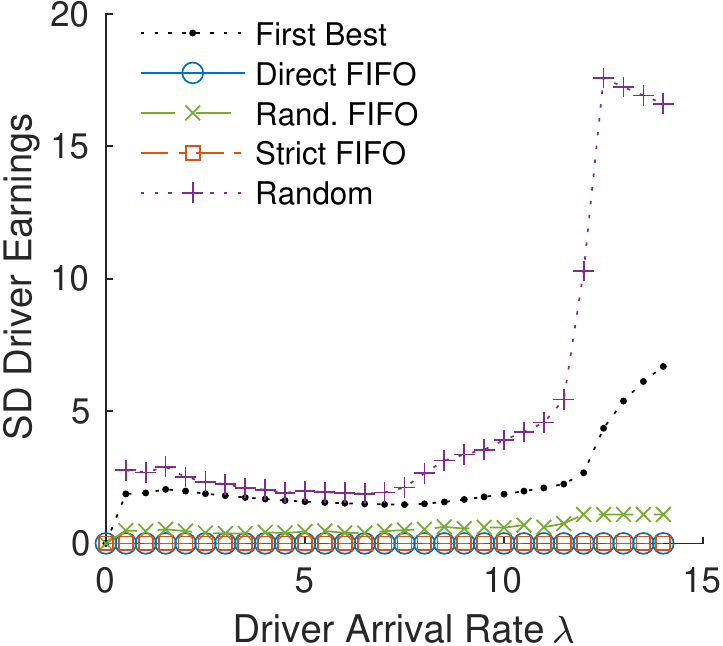}
    \caption{SD of driver payoffs. \label{fig:varying_m_sd_driver_earning_ohare}}
\end{subfigure}%
\caption{Drivers' average waiting times, average payoff, and the standard deviation (SD) in drivers' payoff in equilibrium in steady state, as the arrival rate of drivers varies. Chicago O'Hare. \label{fig:varying_lambda_ohare_2}}   
\end{figure}

In Figures~\ref{fig:varying_m_ave_driver_earning_ohare} and~\ref{fig:varying_m_sd_driver_earning_ohare}, we compare the average payoff (i.e. the net earnings from trips %
minus the waiting costs) of all drivers who arrived at the airport, and also the standard deviation of drivers' payoffs. 
As expected, random dispatching introduces substantial uncertainty in drivers' payoffs. In contrast, by matching higher-earning trips with drivers who have waited longer in the queue, the randomized FIFO mechanism achieves a much smaller variance in drivers' payoffs, in comparison to random dispatching as well as the first best outcome. %

\paragraph{Varying Rider Patience.}

Fixing the arrival rate of drivers at $\lambda = 10$, we %
compare the equilibrium, steady state outcomes under different mechanisms when riders' patience level increases from $\patience = 1$ to $\patience = 120$.  %
Figure~\ref{fig:varying_patience_ohare_1} shows the net revenue, trip throughput, and the length of the queue, and Figure~\ref{fig:varying_patience_ohare_2} shows drivers' average waiting times in queue, drivers' average payoff after arriving at the queue, and the standard deviation of drivers' payoffs.

\begin{figure}[t!]
\centering
\begin{subfigure}[t]{\subfigWidthThree \textwidth}
	\centering
    \includegraphics[height= \figHeight in]{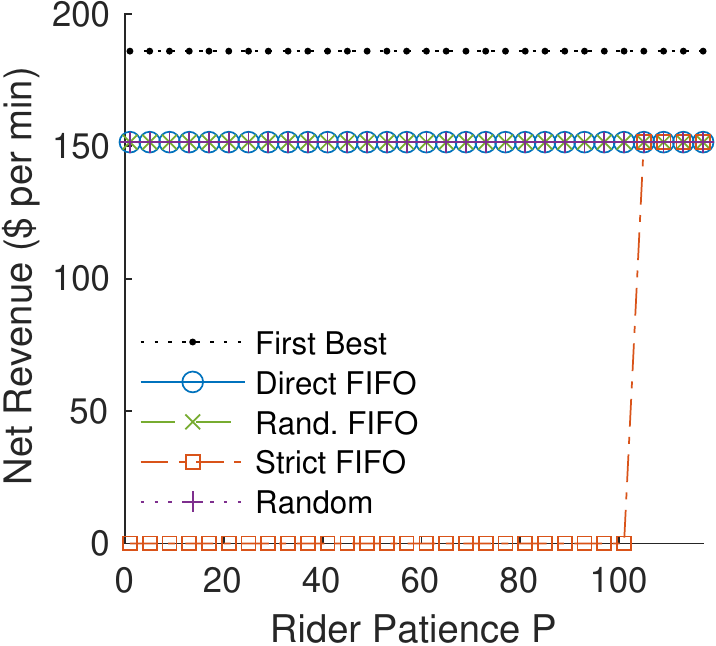}
    \caption{Net revenue.\label{fig:varying_P_gb_ohare}}
\end{subfigure}%
\begin{subfigure}[t]{\subfigWidthThree \textwidth}
	\centering
    \includegraphics[height=\figHeight in]{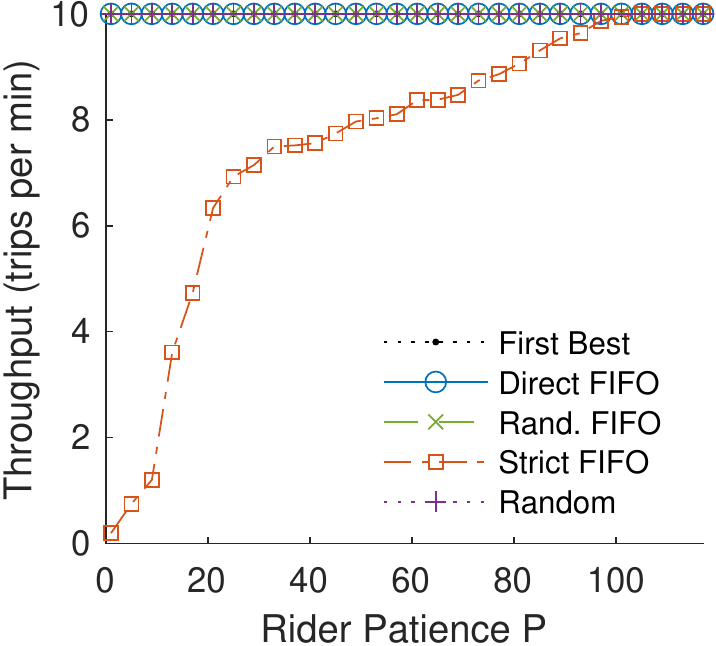}
    \caption{Trip throughput. \label{fig:varying_P_throughput_ohare}}
\end{subfigure}%
\begin{subfigure}[t]{\subfigWidthThree \textwidth}
	\centering
    \includegraphics[height=\figHeight in]{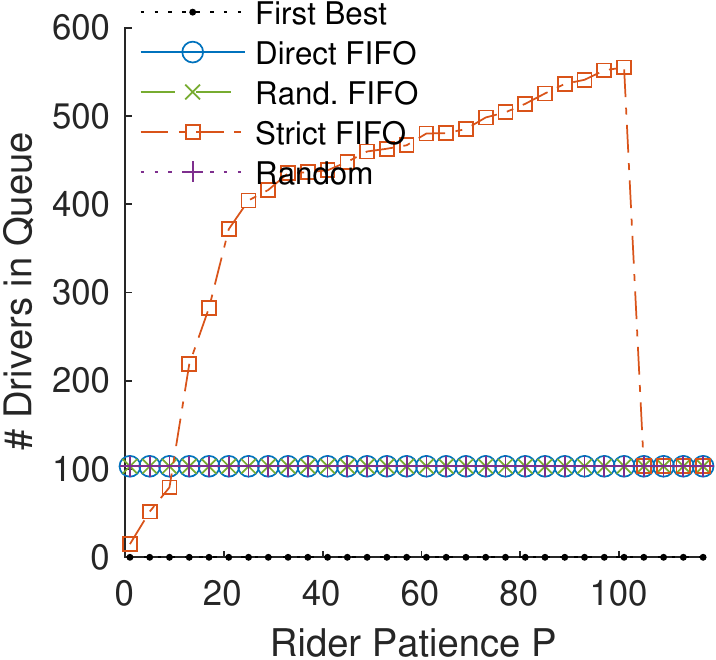}
    \caption{Equilibrium queue length. \label{fig:varying_P_max_q_length_ohare}}
\end{subfigure}%
\caption{Equilibrium net revenue, trip throughput, and length of the queue in steady state, as riders' patience level varies. Chicago O'Hare. \label{fig:varying_patience_ohare_1}}    
\end{figure}

The equilibrium outcomes under the direct FIFO mechanism and random dispatching are not affected by riders' patience level. %
Both mechanisms achieve the first best trip throughput, a high net revenue for the platform, and a low waiting time for the drivers. The randomized FIFO mechanism achieves the same throughput, revenue, and average driver waiting time. Moreover, we see from Figure~\ref{fig:varying_P_sd_driver_earning_ohare} that (i) the variance in drivers' total payoffs is substantially lower than that under random dispatching, and (ii) this variance diminishes rapidly as riders' patience level increases. 
Intuitively, riders' patience level $\patience$ determines the number of times a trip can be dispatched, hence the number of bins a randomized FIFO mechanism may employ. As $\patience$ increases, the mechanism is able to better match trips with higher earnings with drivers who have waiting longer in the queue. %

\begin{figure}[t!]
\centering
\begin{subfigure}[t]{\subfigWidthThree \textwidth}
	\centering
    \includegraphics[height=\figHeight in]{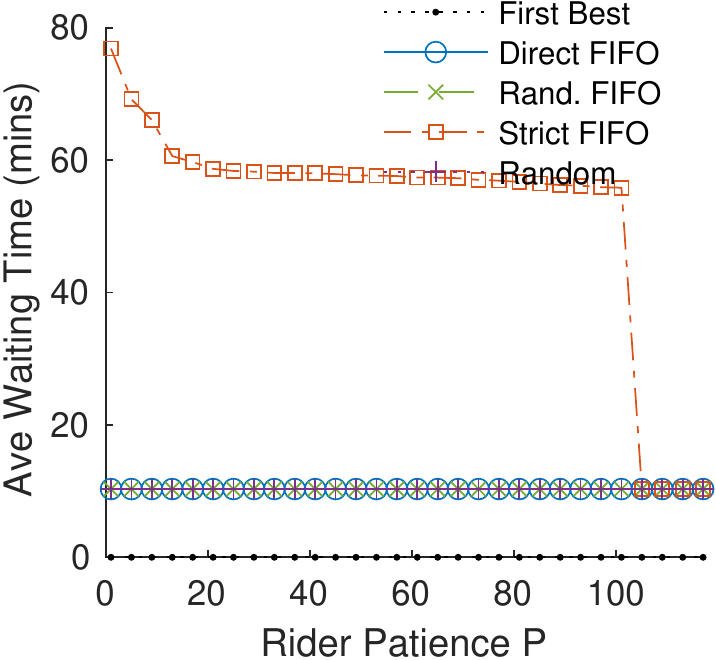}
    \caption{Average waiting time. \label{fig:varying_P_ave_wait_ohare}}
\end{subfigure}%
\begin{subfigure}[t]{\subfigWidthThree \textwidth}
	\centering
    \includegraphics[height=\figHeight in]{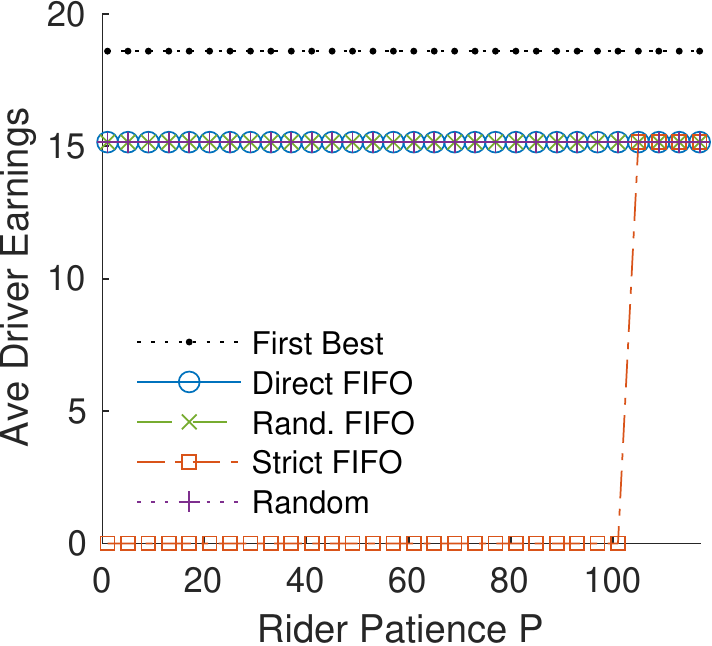}
    \caption{Average driver payoff. \label{fig:varying_P_ave_driver_earning_ohare}}
\end{subfigure}%
\begin{subfigure}[t]{\subfigWidthThree \textwidth}
	\centering
    \includegraphics[height=\figHeight in]{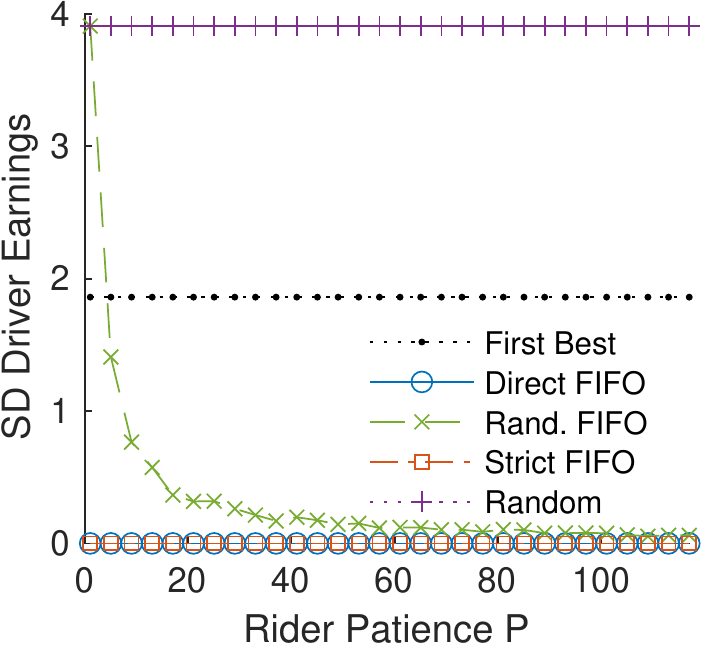}
    \caption{SD of driver payoffs. \label{fig:varying_P_sd_driver_earning_ohare}}
\end{subfigure}%
\caption{Drivers' average waiting times, average payoff, and the standard deviation (SD) in drivers' payoff in equilibrium in steady state, as riders' patience level varies. Chicago O'Hare. \label{fig:varying_patience_ohare_2}}   
\end{figure}

Strict FIFO dispatching, on the other hand, performs poorly. As the patience level increases, %
trips to more destinations can reach drivers in the queue who are willing to accept them, thus the %
throughput increases. The net revenue %
and the average driver payoff remain at zero, however, because drivers continue to join the queue until the payoff from joining is no better than that from leaving without a rider.
Once $\patience$ exceeded $100$, strict FIFO is finally able to dispatch all drivers that arrive at the airport, %
achieving the second best net revenue.
This level of rider patience is not practical, however, since even when each driver decline takes only 10 seconds, $\patience > 100$ requires that riders wait for over fifteen minutes to get matched to a driver. %

\section{Conclusion} \label{sec:conclusion}

We study the dispatching of trips to drivers in a queue, where some trips are necessarily more lucrative than the others due to operational constraints. We propose a family of randomized FIFO mechanisms, which send declined trips gradually down the queue in a randomized manner, and
achieve in equilibrium the %
highest possible revenue and throughput %
under any mechanism that is transparent and flexible.
Extensive counterfactual simulations %
demonstrate substantial improvements of throughput and revenue in comparison to the status quo strict FIFO dispatching, highlighting the effectiveness of using drivers' waiting times in the queue to align incentives, improve efficiency and reliability, and reduce the variability in driver earnings.

From a technical perspective, our setting generalizes existing work in the literature by modeling rider impatience and endogenizing drivers' decisions to join, leave, or re-join the platform. The randomized FIFO mechanisms we propose are also appealing for practice since drivers are provided trip destination and earnings information upfront, as well as the flexibility to freely accept or decline any dispatches. Even when a mechanism is allowed to %
impose penalties such that drivers would lose their position in the queue after declining a dispatch (i.e., moving back to the tail of the queue), no such mechanism can achieve a higher better throughput, revenue, or total driver earnings. %
All things considered, the randomized FIFO mechanism achieves a desirable balance between efficiency, flexibility, fairness and variability in driver earnings.

\bibliography{Refs-randomized-fifo}

\newpage

\appendix

\noindent{}\textbf{\huge{Appendix}}

\bigskip

\noindent{}%
Appendix~\ref{appx:proofs} provides proofs omitted from the body of the paper. We derive the equilibrium %
outcome under various mechanisms and benchmarks in Appendix~\ref{appx:var_of_earnings_derivation}. Additional examples and discussions are provided in Appendix~\ref{appx:additional_discussion}, and we include in Appendix~\ref{appx:additional_simulations} detailed description of the data from the City of Chicago as well as additional simulation results. %
 
\section{Proofs} \label{appx:proofs}

\subsection{Equilibrium Outcome Under Strict FIFO} \label{appx:proof_strict_FIFO}

Before formally stating and proving the equilibrium outcome under strict FIFO dispatching, we first provide the following result on necessary conditions of best-response strategies.

Recall from Section~\ref{sec:dynamic_mech} that given a mechanism, $\pi(\queue, \Queue, \sigma, \sigma')$ denotes the expected continuation payoff (net earnings from trip minus waiting costs) of a driver at position $\queue \geq 0$ in the queue, when the length of the queue is $\Queue \geq \queue$, when this driver adopts strategy $\sigma$, and when every other driver employs strategy $\sigma'$.
Moreover, under a strategy $\sigma = (\alpha, \beta, \gamma)$, $\alpha(\queue,\Queue,i)$, $\beta(\queue,\Queue)$ and $\gamma(\queue,\Queue)$ denote the probability for a driver to (i) accept a trip to location $i \in \loc$, (ii) re-joins the queue at the tail, and (iii) leave the queue without a rider, when the length of the queue is $\Queue \geq 0$ and when the driver is at some position $\queue \in [0, \Queue]$.

\begin{lemma} \label{lem:best_response_necessary_conditions}
Fix a strategy $\sigmast$ adopted by the rest of the drivers. $\sigma = (\alpha, \beta, \gamma)$ is a \emph{best-response strategy} only if for any queue length $\Queue \geq 0$ and at any position in the queue $\queue \leq \Queue$, 
\begin{enumerate}[(i)]
	\item the driver accepts (or declines) with probability one trips for which the net earnings is strictly above (or below) the continuation payoff, i.e. for all $i \in \loc$, $\gb_i > \pi(\queue, \Queue, \sigma, \sigmast) \Rightarrow \alpha(\queue, \Queue, i)=1$, and $\gb_i < \pi(\queue, \Queue, \sigma, \sigmast) \Rightarrow \alpha(\queue, \Queue, i) = 0$,
	\item the driver rejoins at the tail of the queue with probability one (or zero) if the continuation payoff at the tail of the queue is strictly higher (or lower), i.e. $ \pi(\queue, \Queue, \sigma, \sigmast) < \pi(\Queue, \Queue, \sigma, \sigmast) \Rightarrow  \beta(\queue, \Queue) = 1$ and $ \pi(\queue, \Queue, \sigma, \sigmast) > \pi(\Queue, \Queue, \sigma, \sigmast) \Rightarrow  \beta(\queue, \Queue) = 0$, and
	\item the driver leaves the queue without a rider trip with probability one (or zero) if the continuation payoff is strictly negative (or positive), i.e. $\pi(\queue, \Queue, \sigma, \sigmast) < 0 \Rightarrow \gamma(\queue, \Queue) = 1$ and $\pi(\queue, \Queue, \sigma, \sigmast) > 0 \Rightarrow \gamma(\queue, \Queue) =0$.
\end{enumerate}
\end{lemma}

When the length of the queue is $\Queue$, a driver at location $\queue$ who is dispatched a trip to location $i \in \loc$ faces the decision of whether to accept the trip and get a continuation payoff of $\gb_i$, or to decline the trip and remain in the queue. The continuation payoff from remaining in the queue given strategy $\sigma$ is $\pi(\queue,\Queue,\sigma,\sigmast)$, thus a best response must satisfy condition (i) in Lemma~\ref{lem:best_response_necessary_conditions}. 
Similarly, it is easy to see that a violation of either condition (ii) or (iii)  leads to a useful deviation that improves the driver's payoff, contradicting the assumption that $\sigma$ is a best-response strategy. 

Condition (i) also implies that an optimal acceptance strategy $\alpha$ must have a cut-off structure, such that for any $\Queue \geq 0$ and any $\queue \in [0, \Queue]$, $\alpha( \queue, \Queue, i) > 0$ for  location $i \in \loc$ implies $\alpha(\queue, \Queue, j) =1$ for all $j < i$, since trips to these destinations have higher net earnings. 

\bigskip

Recall that $\numSkip_1 \triangleq 0$%
, and observe that for each $i \geq 2$, $\numSkip_i$ as defined in \eqref{equ:first_accpet_positions} can be rewritten as:
\begin{align}
	\numSkip_i & \triangleq  \sum_{j = 1}^{i-1} \left( \frac{\gb_j - \gb_{j+1}}{ \cost}  \sum_{k=1}^j \mu_k \right) = \sum_{j = 1}^{i-1} \frac{\gb_j - \gb_i}{\cost} \mu_j.
	 \label{equ:first_accpet_positions_appx}
\end{align}
Similarly, $\maxQL$ as defined in \eqref{equ:max_eq_queue_length} can be rewritten as: %
\begin{align}
		\maxQL \triangleq \numSkip_\numLoc + \frac{\gb_\numLoc}{\cost} \sum_{i = 1}^\numLoc \numRiders_i = \sum_{i \in \loc} \gb_i \mu_i / \cost, \label{equ:max_eq_queue_length_appx}
\end{align}
and it is straightforward to see that $\numSkip_i \leq \maxQL$ for all $i \in \loc$. 
We now formally state and prove Lemma~\ref{lem:strict_FIFO_eq}, on the equilibrium outcome of strict FIFO dispatching when riders have infinite patience and never cancel their trip requests.

\newcount\lemmaCount
\lemmaCount = \thelemma

\addtocounter{lemma}{-\thelemma}

\begin{lemma}[SPE under strict FIFO with infinite rider patience] 
\label{lem:strict_fifo_spe_formal}
Assume that riders are infinitely patient. Under strict FIFO dispatching, it is a subgame-perfect equilibrium for drivers to:\footnote{Note that the strategy prescribed by this lemma is a particular SPE. There exist other strategies that may form an SPE among the drivers, depending on how drivers break ties between alternatives with equal continuation payoffs. %
}
\begin{enumerate} [$\bullet$]
	\item accept trips to each location $i \in \loc$ if and only if the driver is at position $\queue \geq \numSkip_i$ in the queue, and
	\item join the queue if and only if the length of the queue is weakly below $\maxQL$, and never leave the queue or move to the tail after joining.
\end{enumerate}
\end{lemma}
\addtocounter{lemma}{-\thelemma+\lemmaCount}

\begin{proof}
Let $\sigmast = (\alphast, \betast, \gammast)$ be the strategy specified by the lemma, i.e. for any queue length $\Queue \geq 0$ and any position in the queue $\queue \in [0, \Queue]$, %
\begin{align*}
	\alphast(\queue, \Queue, i) = & \one{\queue \geq \numSkip_i}, \\
	\betast(\queue,\Queue) = & 0, \\
	\gammast(\queue, \Queue) = & \one{\queue > \maxQL}.  
\end{align*}
Here, $\one{\cdot}$ is the indicator function. $\gammast(\queue, \Queue) = \one{\queue > \maxQL}$ means that the driver will leave (or not join) the queue  if and only if the driver's position is (or will be) later than $\maxQL$. 
What we need to show is that starting from any queue length $\Queue\geq 0$, assuming that the rest of the drivers all adopt strategy $\sigmast$, it is a best response for a driver at any position $\queue \in [0, \Queue]$ in the queue to also employ strategy $\sigmast$. We prove this by induction on (segments of) positions in the queue, starting from the head of the queue.

\medskip

\noindent{}\emph{The base case.}
First, consider the driver at the head of the queue (i.e. at position $\queue =0$). The (infinitesimal) driver does not have to wait any time for a dispatch to location $1$. The  continuation payoff for the driver at $\queue = 0$ under $\sigmast$ is therefore
\begin{align}
	\pi(0, \Queue, \sigmast, \sigmast) = \gb_1. \label{equ:strict_FIFO_pi_expression_1}
\end{align}
This is the highest net earnings a driver may get from any trip, thus no other strategy may achieve a higher payoff, and $\sigmast$ is a best response for the driver at the head of the queue.  

\medskip

\noindent{}\emph{The induction step.}
Now assume that for some $i \geq 2$, it is a best response for drivers at positions $\queue \leq \numSkip_{i-1}$ in the queue to employ strategy $\sigmast$, and that a driver's optimal continuation payoff starting from position $\queue = \numSkip_{i-1}$ onward is $ \pi(\numSkip_{i-1}, \Queue, \sigmast, \sigmast) = \gb_{i-1}$ for any $\Queue \geq \numSkip_{i-1}$. 
We prove the induction step by showing that:
\begin{enumerate}[(i)]
	\item $\sigmast$ is a best response for a driver at any position $\queue \in (\numSkip_{i-1},  \numSkip_{i}]$ in the queue, and
	\item %
	the optimal continuation payoff from position $\numSkip_i$ onward is $\pi(\numSkip_{i}, \Queue, \sigmast, \sigmast) = \gb_i$.
\end{enumerate}

We first compute the continuation payoff for drivers at positions $(\numSkip_{i-1}, \numSkip_{i}]$ in the queue, assuming that all drivers adopt strategy $\sigmast$. 
First, consider a driver at some position $\queue \in (\numSkip_{i-1}, \numSkip_{i})$. Under $\sigmast$, the driver accepts trips to locations $j < i$, does not leave the queue without a rider trip, or move to the tail of the queue. When all drivers adopt $\sigmast$, trips to locations $1$ through $i-1$ are accepted by drivers at or ahead of $\numSkip_{i-1}$, thus the driver at $\queue$ will not accept any trip she receives from the platform. 
By Little's Law, a driver at $\queue$ will wait $(\queue - \numSkip_{i-1})/\sum_{j \leq i-1} \mu_j$ units of time before she reaches position $\numSkip_{i-1}$ in the queue. By the induction assumption, the driver gets an optimal continuation payoff of $\gb_{i-1}$ starting from  $\numSkip_{i-1}$. As a result, for any $\queue \in (\numSkip_{i-1}, \numSkip_{i})$, 
\begin{align*}
	\pi(\queue, \Queue, \sigmast, \sigmast) = \pi(\numSkip_{i-1}, \Queue, \sigmast, \sigmast) - \cost (\queue - \numSkip_{i-1})/\sum_{j \leq i-1} \mu_j  = \gb_{i-1} -  \cost (\queue - \numSkip_{i-1})/\sum_{j \leq i-1} \mu_j. 
\end{align*}

Now consider a driver at position $\queue = \numSkip_i$ in the queue. If the driver is dispatched and accepts a trip to location $i$, she gets $\gb_i$. If not, the driver moves forward in the queue and her continuation payoff is again 
\begin{align*}
	\lim_{\queue \rightarrow \numSkip_i-} \pi(\queue, \Queue, \sigmast, \sigmast) = \gb_{i-1} - \cost (\numSkip_i - \numSkip_{i-1})/\sum_{j \leq i-1} \mu_j  = \gb_i.
\end{align*}
Combining the two cases, we know 
\begin{align}
	\pi(\queue, \Queue, \sigmast, \sigmast) = \gb_{i-1} -  \cost (\queue - \numSkip_{i-1})/\sum_{j \leq i-1} \mu_j, ~\forall \queue \in (\numSkip_{i-1}, \numSkip_{i}],
	 \label{equ:strict_FIFO_pi_expression_2}
\end{align}
and we have $\pi(\queue, \Queue, \sigmast, \sigmast) > \gb_i$ when $\queue < \numSkip_i$ and  $\pi(\numSkip_i, \Queue, \sigmast, \sigmast) = \gb_i$.

\smallskip

We now prove that $\sigmast$ is a best response for drivers at $(\numSkip_{i-1}, \numSkip_{i}]$ in the queue. Assume towards a contradiction, that there exists some strategy $\sigma$, and some $\queue \in (\numSkip_{i-1}, \numSkip_{i}]$, such that $\pi(\queue, \Queue, \sigma, \sigmast) > \pi(\queue, \Queue, \sigmast, \sigmast)$ for some $\Queue \geq \queue$. %
Note that the driver does not get dispatched any trip with net earnings higher than $\gb_i$ until the driver reaches position $\numSkip_{i-1}$ in the queue. Consider the following two scenarios:
\begin{enumerate}[$\bullet$]
	\item If the driver left the queue (with or without a rider) under $\sigma$ before she reaches $\numSkip_{i-1}$, the driver's payoff is upper-bounded by $\gb_i \leq  \pi(\queue, \Queue, \sigma, \sigmast)$. 
	\item When the driver did reach $\numSkip_{i-1}$ under $\sigma$, her %
	optimal continuation payoff from $\numSkip_{i-1}$ onward is %
	$\gb_{i-1}$ given the induction assumption. Moreover, the driver will incur a waiting cost of at least $ \cost (\queue - \numSkip_{i-1})/\sum_{j \leq i-1} \mu_j$ before reaching $\numSkip_{i-1}$ thus driver's continuation payoff starting from $\queue$ is again upper bounded by $\pi(\queue,\Queue, \sigmast, \sigmast)$.
\end{enumerate}
Combining the two cases, we know that $\pi(\queue, \Queue, \sigma, \sigmast) > \pi(\queue, \Queue, \sigmast, \sigmast)$  is not achievable for any $\queue \in (\numSkip_{i-1}, \numSkip_{i}]$ under any strategy $\sigma$, thus $\sigmast$ is a best response. This completes the proof of the induction step.

\medskip

\noindent{}\emph{End of the queue.} What we have proved by induction is that $\sigmast$ is a best response for any driver at positions $\queue \in [0, \numSkip_\ell]$ in the queue, and that the optimal continuation payoff from $\numSkip_\ell$ onward under any strategy is $\gb_\ell$. 
Now consider drivers at positions $\queue \in ( \numSkip_\numLoc, \maxQL]$ in the queue. Following $\sigmast$ implies waiting until reaching $\numSkip_\numLoc$ in the queue, thus the continuation payoff is:
\begin{align}
	\pi(\queue, \Queue, \sigmast, \sigmast) = \gb_\ell -  \cost (\queue - \numSkip_\ell)/\sum_{j \leq \ell} \mu_j \geq \gb_\ell - \cost (\maxQL - \numSkip_\ell)/\sum_{j \leq \ell} \mu_j= 0, ~\forall \queue \in (\numSkip_\ell, \maxQL]. \label{equ:strict_FIFO_pi_expression_3}
\end{align}
An argument very similar to the proof of the induction step shows that regardless of whether a driver reached $\numSkip_\ell$ in the queue or not, achieving continuation payoff higher than $\pi(\queue, \Queue, \sigmast, \sigmast)$ is not possible and $\sigmast$ is a best response. 
Moreover, $\pi(\maxQL, \maxQL, \sigmast, \sigmast) = 0$ holds, thus it is a best response to leave (or not join) the queue when at position $\queue > \maxQL$ (or when the length of the queue is longer than $\maxQL$). This completes the proof of this lemma. 
\end{proof}

\subsection{Incentive Compatibility and Optimality of Direct FIFO} \label{appx:proof_direct_FIFO}

In this section, we provide proofs for the incentive compatibility and optimality of the direct FIFO mechanism. 

\thmFIFOSkipSPE*

\begin{proof} 

Let $\sigmast$ denote the strategy of %
(i) always accepting trip dispatches from the direct FIFO mechanism, (ii) join the queue if and only if the length of the queue is $\Queue \leq \maxQL$, and (iii) once in the queue, never leave the queue without a rider or move to the tail of the queue.

To establish the incentive compatibility of the direct FIFO mechanism, we need to show that starting from any queue length $\Queue \geq 0$, and assuming the rest of the drivers all adopt strategy $\sigmast$, it is a best response for a driver %
to also employ strategy $\sigmast$. This can be established using a very similar (and in fact, slightly simpler) argument as in the proof of Lemma~\ref{lem:strict_fifo_spe_formal}.
We do not repeat the same arguments here but refer the readers to Appendix~\ref{appx:proof_strict_FIFO}, where we established %
the SPE under strict FIFO (assuming infinite rider patience) by induction on segments of the queue. %

\smallskip

It is also straightforward to show that the equilibrium continuation payoff under direct FIFO is also identical to that under strict FIFO dispatching where riders are infinitely patient. Combining equations \eqref{equ:strict_FIFO_pi_expression_1}, \eqref{equ:strict_FIFO_pi_expression_2} and \eqref{equ:strict_FIFO_pi_expression_3}, we have the following expression for the equilibrium continuation payoff of a driver at position $\queue \in [0, \Queue]$ in the queue for any queue length $\Queue \geq 0$ under the direct FIFO mechanism: 
\begin{align}
	\pi(\queue, \Queue, \sigmast, \sigmast) = \pwfun{
		\gb_1, & \txtif \queue = 0, \\
		\gb_{i-1} - \cost (\queue - \numSkip_{i-1})/\sum_{j \leq i-1}\mu_j , & \txtif \queue \in (\numSkip_{i-1}, \numSkip_i],~\forall i \geq 2, \\
		\gb_\numLoc -  \cost (\queue - \numSkip_{\numLoc})/\sum_{j \in \loc}\mu_j , & \txtif \queue \in (\numSkip_{\numLoc}, \maxQL], \\
		0,  & \txtif q > \maxQL.
		} \label{equ:pist_direct_fifo}
\end{align}
For any $\queue \leq \maxQL$, $\pi(\queue, \Queue, \sigmast, \sigmast)$ is non-negative, continuous, and monotonically decreasing in $\queue$, thus the direct FIFO mechanism is individually rational and envy-free (i.e., no drivers envies the other drivers in positions behind her in the queue).

Also observe that there is no randomness at all  in a driver's continuation payoff starting from any position in the queue, since at each point $\numSkip_i$, the driver gets precisely $\gb_i$ regardless of whether the driver accepted a trip to location $i$ and left, or if the driver moved forward in the queue. As a result, the individual rationality and envy-freeness properties also hold \emph{ex post}. This completes the proof of the theorem.
\end{proof}

\thmFIFOSB*

\begin{proof}
Let $\sigmast$ denote the equilibrium strategy of accepting all dispatches from the direct FIFO mechanism, and joining the queue if and only if the length of the queue is at most $\maxQL$ (see Theorem~\ref{thm:fifo_skip_spe}). 
We prove the optimality of the direct FIFO mechanism with the following three steps:
\begin{enumerate}[{Step} 1.]
    \item Establish the steady state outcome when all drivers adopt strategy $\sigmast$, and prove that the same set of trips that are completed under the first best outcome are also completed under direct FIFO.
    \item Show that in equilibrium, no transparent and flexible  mechanism is able to achieve a higher %
    total payoff than that under direct FIFO %
    for all drivers who arrive at the queue. %
    \item Complete the proof that no mechanism is able to achieve a better net revenue. 
\end{enumerate}
We start from Step~1.

\medskip

\noindent{}\emph{Step 1.}
We first establish the steady state equilibrium outcome under the direct FIFO mechanism. There are two cases, depending on whether the platform is over or under-supplied.

\smallskip 
 
\noindent{}\emph{Step 1.1: $\lambda > \sum_{i \in \loc} \mu_i$.}
We first show that in the over-supplied case, when all drivers adopt strategy $\sigmast$, $\equQL = \maxQL$ is a steady-state queue length. 
To prove this, first observe that with $\equQL = \maxQL \geq \numSkip_i$ for all $i \in \loc$, all rider trips are accepted. The rate at which drivers are dispatched from the queue is $\sum_{i \in \loc} \mu_i < \lambda$, thus drivers effectively join the queue with probability $ \sum_{i \in \loc} \mu_i/ \lambda $ and the length of the queue remains constant at $\equQL = \maxQL$. 
Observe that the total payoff achieved by all drivers who arrive at the queue is zero, because a driver gets a zero payoff regardless of whether she join the queue upon arrival, or left immediately without joining.

We also show that $\equQL = \maxQL$ is the unique steady queue length, by proving that starting from any queue length $\Queue \neq \maxQL$, the length of the queue will converge to $\maxQL$ within a finite amount of time. First, we know from \eqref{equ:pist_direct_fifo} that the equilibrium continuation payoff of a driver at any position $\queue < \maxQL$ in the queue is strictly positive. 
If the length of the queue $\Queue$ is strictly shorter than $\maxQL$, a driver strictly prefers to join the queue upon arrival, and drivers join the queue at a rate of $\lambda$ under $\sigmast$. 
This cannot be the steady state outcome, since the rate at which drivers are dispatched from the queue is at most $\sum_{i \in \loc} \mu_i < \lambda$, and even lower when $\Queue < \numSkip_\ell$. As a result, the queue length will grow at a rate of at least $\lambda - \sum_{i \in \loc} \mu_i$, whenever $\Queue < \maxQL$. 
Moreover, any queue length $\Queue > \maxQL$ cannot be an steady sate either, since \eqref{equ:pist_direct_fifo} implies that the drivers at positions $\queue > \maxQL$ have strictly negative continuation payoffs, thus will leave the queue immediately.

\smallskip

\noindent{}\emph{Step 1.2: $\lambda \leq \sum_{i \in \loc} \mu_i$.}
Recall that when a platform is not over-supplied, $\maxJ \in \loc$ as defined in \eqref{equ:maxJ} denotes the lowest-earning (i.e. highest index) trip that is (partially) completed under the first best outcome, when the $\lambda$ units of drivers are dispatched to destinations in decreasing order of $\gb_i$. 

We first show that $\equQL = \numSkip_{\maxJ}$ is a steady state equilibrium queue length. When the length of the queue is $\numSkip_{\maxJ}$, all trips to locations $i < \maxJ$ will be dispatched and accepted by drivers in the queue. $\sum_{i < \maxJ} \mu_i$ out of the $\lambda$ drivers move forward in the queue upon arrival, and the remaining $\lambda - \sum_{i < \maxJ} \mu_i$ drivers leave the queue immediately with trips to location $\maxJ$ that are dispatched to the tail of the queue $\equQL = \numSkip_\maxJ$. In this way, rate at which drivers join the queue is the same as the rate at which drivers are dispatched from the queue, and the length of the queue remains at $\numSkip_\maxJ$.

Observe that the set of trips completed in steady state under direct FIFO is the same as those completed under the first best outcome. Moreover, a driver gets a payoff of $\gb_\maxJ$ regardless of whether the driver accepted a trip to location $\maxJ$ immediately after arrival. As a result, the total payoff of all drivers is $\lambda \gb_\maxJ$ per unit of time.

We also show that $\equQL = \numSkip_\maxJ$ is the unique steady state queue length for all non-degenerate economies, meaning that $\lambda \neq \sum_{j = 1}^i \mu_j$ for any $i \in \loc$. Consider the following two scenarios:  
	\begin{enumerate}[$\bullet$]
		\item When the length of the queue is $\Queue < \numSkip_{\maxJ}$, trips to locations $j \geq \maxJ$ are not dispatched under the direct FIFO mechanism. The excess drivers, however, will still join the queue under $\sigmast$ (at $\Queue < \maxQL$, the payoff from joining is strictly positive). As a result, the length of the queue will grow at a rate at least $\lambda - \sum_{j < \maxJ} \mu_ > 0j$, as long as it is strictly below $\numSkip_\maxJ$.
		\item When $\Queue > \numSkip_\maxJ$, all trips to locations $j \leq \maxJ$ are dispatched and accepted under direct FIFO. As a result, the length of the queue will decrease at a rate of $\sum_{j \leq \maxJ} \mu_j - \lambda$  when $\lambda < \sum_{j \leq \maxJ} \mu_j$, until it reaches $\Queue = \numSkip_{\maxJ}$. %
		In the degenerate case where $\lambda = \sum_{j \leq \maxJ} \mu_j$, any queue length between $\numSkip_\maxJ$ and $\numSkip_{\maxJ + 1}$ may be a steady  state queue length, and we break ties in favor of shorter queues under the direct FIFO mechanism. %
	\end{enumerate}
\smallskip

Combining the two settings in Step 1.1 and 1.2, we know that the same set of trips that are completed under the first best outcome are also completed under direct FIFO. This implies that the direct FIFO mechanism achieves in equilibrium the first best steady state trip throughput of $\tp\direct = \min\{\sum_{i \in \loc} \mu_i,~\lambda\}$. Moreover, when $\mechCost = 0$, the outcome under direct FIFO also achieves the first best revenue, since the total net earnings from trips is the same as that under the first best, and drivers' lining up in the queue is not costly for the platform.

\medskip

\noindent{}\emph{Step 2.}
We now prove that it is not possible to improve the total payoff of all drivers who had arrived at the queue, when  drivers have access to trip destinations upfront and have the option to decline trips and to re-join the queue at the tail at any point of time. 
Again we discuss the under-supplied and the over-supplied cases separately. 

\smallskip

\noindent{}\textit{Step 2.1: $\lambda > \sum_{i \in \loc} \mu_i$.} 
We need to prove that in equilibrium, under any mechanism that is transparent and flexible, drivers cannot get a strictly positive average payoff after arriving at the queue.
To show this, consider a mechanism $\mech$ that is flexible and transparent. %
It cannot be a steady state equilibrium under $\mech$ for every driver to leave the queue with a rider trip. %
As a result, some driver must willingly leave without a rider, and the net payoff of such drivers is non-positive.

Assume towards a contradiction that $\mech$ achieves a strictly positive average driver payoff, and let $\sigma'$ and $\Queue'$ denote the equilibrium strategy under $\mech$, and the steady state queue length under $\mech$, respectively. 
The expected continuation payoff of a driver who joined the queue at the tail must be strictly positive: $\pi_\mech(\Queue', \Queue', \sigma', \sigma') > 0$. This is because the drivers who did not join the queue upon arrival (if any) have zero net earnings thus if $\pi_\mech(\Queue', \Queue', \sigma', \sigma') \leq 0$, the average payoff of all drivers who arrived at the queue will be non-positive. 
$\pi_\mech(\Queue', \Queue', \sigma', \sigma') > 0$, however, contradicts the assumption that the outcome forms an equilibrium. In this case, no driver will be willing to leave the queue without a rider trip, since it is a useful deviation to join the queue at the tail and get a strictly positive payoff. 

\smallskip

\noindent{}\textit{Step 2.2: $\lambda \leq \sum_{i \in \loc} \mu_i$.} 
As we've shown in Step 1.2, in this case drivers have an average payoff of $\gb_\maxJ$ after arriving at the queue, where $\maxJ$ is the lowest earning trip that is (partially) completed in equilibrium. What we need to prove is that under any mechanism $\mech$ that does not penalize drivers for declining dispatches or rejoining the queue at the tail, the average payoff of a driver who arrived at the queue cannot exceed $\gb_\maxJ$.

First, by definition of $\maxJ$, it cannot be a steady state equilibrium under $\mech$ for every driver who arrive at the virtual queue to leave the queue with a rider trip to a location $j < \maxJ$. As a result, some driver must leave with a trip to some location $j \geq \maxJ$, or leave without a rider. In both cases, the driver's continuation payoff after accepting a trip or leaving the queue is upper bounded by $\gb_\maxJ$. 
This cannot form an equilibrium when $\pi_\mech(\Queue', \Queue', \sigma', \sigma') > \gb_\maxJ$ (since a driver is better off re-joining the queue at the tail instead, %
therefore $\pi_\mech(\Queue', \Queue', \sigma', \sigma') \leq \gb_\maxJ$ must hold. 

This completes the proof of Step 2.

\medskip

\noindent{}\emph{Step 3.} 
We now prove that no mechanism can achieve a higher net revenue in equilibrium than that under the direct FIFO mechanism. The case of $\mechCost = 0$ was already discussed in Step 1. The case where $\mechCost = \cost$ is also straightforward, since in this case the net revenue of the platform is equal to the total net payoff of all drivers combined (see discussions in Section~\ref{sec:direct_FIFO}), thus Step~2 implies that no mechanism can achieve a higher net revenue.

What is left to prove is the case where $\mechCost \in (0, \cost)$. %
Consider an alternative mechanism $\mech$, and let $\{ \mutilde_i \}_{i \in \loc}$ be the rate at which mechanism $\mech$ completes trips to each destination in equilibrium in steady state. %
We are going to prove that the net revenue under $\mech$ is optimized when the outcome under $\mech$ is the same as that under direct FIFO, and we again discuss the over and under-supplied cases separately.

\smallskip

\noindent{}\emph{Step 3.1: $\lambda > \sum_{i \in \loc} \mu_i$.} Given Step~2, drivers get a total payoff of zero under $\mech$. Assuming that the equilibrium queue length is $\equQL_\mech$, we have:
    \begin{align}
        \sum_{i \in \loc} \mutilde_i \gb_i - \cost \equQL_\mech = 0. \label{equ:thm_1_equality_condition_1}
    \end{align}
    The platform, however, may still get a non-zero net revenue
    \begin{align*}
        \rev_\mech = \sum_{i \in \loc} \mutilde_i \gb_i - \mechCost \equQL_\mech  = (\cost - \mechCost)  \equQL_\mech \geq 0,
    \end{align*}
    and it is straightforward to see that %
    $\rev_\mech$ is optimized when $\equQL_\mech$ is the maximized. With \eqref{equ:thm_1_equality_condition_1}, $\equQL_\mech = \sum_{i \in \loc} \mutilde_i \gb_i/\cost$ is maximized when $\mutilde_i = \mu_i$ for all $i \in \loc$. This is the same outcome as that under the direct FIFO mechanism, thus no mechanism can achieve a better net revenue.

\smallskip

\noindent{}\emph{Step 3.2: $\lambda \leq \sum_{i \in \loc} \mu_i$.} In this case, drivers get an average payoff of $\gb_\maxJ$ under direct FIFO, and the equilibrium queue length is $\equQL\direct = \numSkip_\maxJ$. 
Let $\tp_\mech \triangleq \sum_{i \in \loc} \mutilde_i$ denote the trip throughput under mechanism $\mech$, and let $\equtil_\mech$ be the average equilibrium payoff of drivers achieved under $\mech$. Consider the following two cases:
\begin{enumerate}[$\bullet$]
    \item $\tp_\mech < \lambda$, in which case not all drivers receive rider trips in equilibrium under $\mech$. 
    An argument very similar to that in Step~2 shows that in this case, the average payoff of a driver who joined the queue upon arrival must be zero, thus $\equtil_\mech = 0$. Similar to the over-suppllied setting, we have
    \begin{align*}
    	\sum_{i \in \loc} \mutilde_i \gb_i - \cost \equQL_\mech = 0,
    \end{align*}
    which implies 
    \begin{align*}
      \rev_\mech = \sum_{i \in \loc} \mutilde_i \gb_i - \mechCost \equQL_\mech  = (\cost - \mechCost)  \equQL_\mech. 
    \end{align*}
    $ \rev_\mech$ is again optimized when $\equQL_\mech$ is the longest. For any fixed throughput $ \tp_\mech = \sum_{i \in \loc} \mutilde_i  < \lambda$,  the queue length $\equQL_\mech = \sum_{i \in \loc} \mutilde_i \gb_i / \cost$ is maximized when the $\tp_\mech$ units of drivers are dispatched to  trips in decreasing order of $\gb_i$, and this implies that the net revenue $\rev_\mech =  (\cost - \mechCost)  \equQL_\mech$ is upper bounded by:
    \begin{align*}
    	\rev_\mech  \leq & (\cost - \mechCost) \left( \sum_{i < \maxJ} \mu_i \gb_i / \cost + (\lambda - \sum_{i < \maxJ}\mu_i) \gb_\maxJ / \cost \right) \\
    	= & \sum_{i < \maxJ} \mu_i \gb_i +(\lambda - \sum_{i < \maxJ}\mu_i) \gb_\maxJ - \frac{\mechCost}{\cost} \left(\sum_{i < \maxJ} \mu_i \gb_i + (\lambda - \sum_{i < \maxJ}\mu_i) \gb_\maxJ \right)
	\end{align*}     
	This is weakly below the net revenue under direct FIFO, which can be written as:
	\begin{align*}
		\rev\direct =& \sum_{i < \maxJ} \mu_i \gb_i + (\lambda - \sum_{i < \maxJ}\mu_i) \gb_\maxJ - \mechCost\numSkip_\maxJ \\
		= &  \sum_{i < \maxJ} \mu_i \gb_i +(\lambda - \sum_{i < \maxJ}\mu_i) \gb_\maxJ - \frac{\mechCost}{\cost} \left(\sum_{i < \maxJ} \mu_i (\gb_i - \gb_\maxJ) + (\lambda - \sum_{i < \maxJ}\mu_i) (\gb_\maxJ - \gb_\maxJ) \right).
	\end{align*}
    \item Consider now the case where $\tp_\mech = \lambda$. Drivers' getting an average payoff of $\equtil_\mech$ implies:
    \begin{align}
        \sum_{i \in \loc} \mutilde_i \gb_i - \cost \equQL_\mech = \lambda \equtil_\mech. \label{equ:thm_1_equality_condition_2}
    \end{align}
	For each $i \in \loc$, denote $ \Delta_i \triangleq (\gb_i - \gb_\maxJ) / \cost$. The equilibrium queue length can  be written as:
    \begin{align}
        \equQL_\mech = & \frac{1}{\cost} \left(  \sum_{i \in \loc} \mutilde_i \gb_i -  \sum_{i \in \loc} \mutilde_i \equtil_\mech \right)
        	=  \sum_{i \in \loc} \mutilde_i \left(\Delta_i + (\gb_\maxJ - \equtil_\mech)/\cost \right). 
    \end{align}
    The net revenue under $\mech$ is therefore of the form:
    \begin{align}
        \rev_\mech = & \sum_{i \in \loc} \mutilde_i \gb_i - \mechCost \equQL_\mech  %
        =  \sum_{i \in \loc} \mutilde_i \left(\gb_i - \mechCost \Delta_i  \right) - \lambda (\gb_\maxJ - \equtil_\mech) \mechCost/\cost. \label{equ:proof_of_thm_2_mech_rev}
    \end{align}
    For the first term in \eqref{equ:proof_of_thm_2_mech_rev}, $\gb_i - \mechCost \Delta_i = \gb_i -  (\gb_i - \gb_\maxJ) \mechCost/\cost = \gb_i (1-\mechCost/\cost) + \gb_\maxJ \mechCost/\cost$ is higher for smaller $i$ with higher $\gb_i$. 
    The second term $- \lambda (\gb_\maxJ - \equtil_\mech) \mechCost/\cost$ is non-positive given Step~2, therefore achieves its maximum when $\equtil_\mech = \gb_\maxJ$.   
    Putting the two parts together, we know that $\rev_\mech$ is optimized when when $\mutilde_i$ is maximized for smallest $i \in \loc$ first (until we have $\sum_{i\in\loc}\mutilde_i=\lambda$), in which case the average payoff achieves $\equtil_\mech = \gb_\maxJ$. 
    This is, again, the same outcome as that under direct FIFO.  
\end{enumerate}
This completes the proof of Step~3, and concludes the proof of the optimality of direct FIFO.
\end{proof}

\subsection{Optimality of Random Dispatching} \label{appx:proof_pure_random}

Before proving the optimality of random dispatching, we first provide the following lemma on the best response strategy of a driver in a stationary environment. %

\begin{lemma}\label{lem:single_driver_best_response_stationary} 
Consider a driver in a stationary environment, where she receives trip offers to each location $i \in \loc$ at a rate of $\eta_i \geq 0$. The highest achievable net payoff from any feasible strategy is $\max \left\lbrace \max_{j \in \loc} \rho_j, ~ 0 \right\rbrace$, 
where 
\begin{align}
	\rho_j \triangleq \left(\sum_{i = 1}^j \gb_i \eta_i - \cost \right) \bigg/ \sum_{i = 1}^j \eta_i. \label{equ:util_top_j}
\end{align}
Moreover, $\bestJ$ is a maximizer of $\rho_j$ if and only if $\rho_\bestJ \leq \gb_\bestJ$ and $\rho_\bestJ \geq \gb_{\bestJ+1}$.
\end{lemma}
\begin{proof} 
Lemma~\ref{lem:best_response_necessary_conditions} implies that any best response strategy on acceptance in this setting must have a cutoff structure, meaning that if the driver accepts a trip to some location $j \in \loc$ with non-zero probability, then she must accept any trip to locations $i < j$ with probability $1$. Moreover, the driver will decide to leave the queue only if the expected continuation payoff from the optimal acceptance strategy is non-positive. %
We now show that the highest achievable net payoff under any best-response strategy is this stationary environment is $\max\{ \max_{j \in \loc} \rho_j, ~ 0 \}$.

Consider for now a deterministic strategy such that the driver stays in the queue, and accepts all trips to locations $1$ through $j$ if offered. We denote this strategy as $\sigma_j$.  
The average net earnings the driver gets from the an average trip she accepts is $\sum_{i = 1}^j \gb_i \eta_i / \sum_{i = 1}^j \eta_i$, and in expectation, the driver will wait $1 / \sum_{i = 1}^j \eta_i$ units of time to receive a trip dispatch she will accept. 
Therefore, the expected net payoff (i.e. the net earnings from trip a driver accepts minus her expected waiting cost) under strategy $\sigma_j$ is
\begin{align*}
	\sum_{i = 1}^j \gb_i \eta_i \bigg/ \sum_{i = 1}^j \eta_i - \cost \bigg/ \sum_{i = 1}^j \eta_i =  \left(\sum_{i = 1}^j \gb_i \eta_i - \cost \right) \bigg/ \sum_{i = 1}^j \eta_i  = \rho_j. 
\end{align*}
Among all deterministic strategies such that the driver does not leave, the highest achievable net payoff is therefore $\max_{j \in \loc} \rho_j$.  %

The cutoff structure proved by Lemma~\ref{lem:best_response_necessary_conditions} also implies that the only potentially useful randomization in a driver's acceptance strategy is on the lowest earning trip that is accepted. 
Consider a strategy where the driver accepts all trips to locations $1$ through $j-1$, but accepts location $j$ trips with probability $\theta  \in [0, 1]$. The expected net payoff in this setting is:
\begin{align*}
	& \left(\sum_{i = 1}^{j-1} \gb_i \eta_i + \gb_j \theta \eta_j - \cost \right) \bigg/ \left(\sum_{i = 1}^{j-1} \eta_i + \theta \eta_j \right) %
	= \left( \rho_{j-1} \sum_{i = 1}^{j-1} \eta_i + \gb_j \theta \eta_j \right) \bigg/ \left(\sum_{i = 1}^{j-1} \eta_i + \theta \eta_j \right).
\end{align*}
This is a weighted average of $\rho_{j-1}$ and $\gb_j$, thus can be optimized at $\theta = 0$ (or $\theta = 1$) if $\rho_{j-1} \geq \gb_j$ (or if $\rho_{j-1} \leq \gb_j$). Therefore, for a driver who does not choose to immediately leave the queue, the highest achievable net payoff can be achieved by a deterministic acceptance strategy, and the optimal payoff under any acceptance strategy is equal to $\max_{j \in \loc} \rho_j$. When this payoff is negative, the driver is better off leaving the queue instead of waiting for any trip dispatches. As a result, a driver's highest possible payoff a driver may achieve in this stationary environment is $\max \left\lbrace \max_{j \in \loc} \rho_j, ~ 0 \right\rbrace$.

\bigskip

What is left to show is that %
$\bestJ$ is a maximizer of $\rho_j$ if and only $\rho_\bestJ \leq \gb_\bestJ$ and $\rho_\bestJ \geq \gb_{\bestJ+1}$. To prove this, first observe that for any $j > 1$, $\rho_j$ is a weighted average of $\rho_{j-1}$ and $\gb_j$:
\begin{align}
	\rho_j = \left( \rho_{j-1} \sum_{i = 1}^{j-1} \eta_i + \gb_j  \eta_j \right) \bigg/ \sum_{i = 1}^j \eta_i. \label{equ:rho_j_as_a_weighted_average}
\end{align}
This implies (i) when $ \rho_j \geq \rho_{j-1}$, it must be the case that $\gb_j  \geq \rho_{j}  \geq \rho_{j-1}$, and (ii) $\rho_j \geq \rho_{j+1} \Rightarrow \rho_j \geq \gb_{j+1}$. 
Therefore, if $\bestJ$ is a maximizer of $\rho_j$, we must have $\rho_\bestJ \geq \rho_{\bestJ - 1} \Rightarrow \gb_\bestJ \geq \rho_\bestJ$, and $\rho_\bestJ \geq \rho_{\bestJ+1} \Rightarrow \rho_\bestJ \geq \gb_{\bestJ+1}$.

On the other hand, if $\rho_\bestJ \leq \gb_\bestJ$ and $\rho_\bestJ \geq \gb_{\bestJ+1}$ both hold, we now prove that $\bestJ$ must be a maximizer of $\rho_j$. 
Denote $\hat{j} \in \loc$ as the first location for which $\rho_{j} > \gb_{j+1}$, i.e. 
\begin{align}
	\hat{j} \triangleq \min\{ j \in \loc ~|~ \rho_{j} > \gb_{j+1}\}. \label{equ:min_j_with_higher_rho_than_w}
\end{align}

We first claim that $\rho_j$ must be monotonically non-decreasing when $j \leq \hat{j}$, i.e. for all $j < \hat{j}$, $\rho_{j} \leq \rho_{j+1}$. This is because for any $j < \hat{j}$,  $\rho_{j} \leq  \gb_{j+1}$ holds by definition of $\hat{j}$, thus by \eqref{equ:rho_j_as_a_weighted_average} we have $\rho_{j} \leq \rho_{j+1}$. 
Moreover, given \eqref{equ:rho_j_as_a_weighted_average} and the fact that $\gb_j$ is monotonically decreasing, we can prove by a simple induction ($\rho_{\hat{j}} > \gb_{\hat{j} + 1} \Rightarrow \rho_{\hat{j}} > \rho_{\hat{j} + 1} > \gb_{\hat{j} + 1} > \gb_{\hat{j} + 2}$ and so on)  that (i) $\rho_j$ must be monotonically decreasing for all $j \geq \hat{j}$, i.e. $\forall j \geq \hat{j}$, $\rho_{j} \geq \rho_{j+1}$, and (ii) $\rho_j > \gb_{j + 1}$ for all $j \geq \hat{j}$.
Combining the two cases, we know that $\hat{j}$ is a maximizer of $\rho_j$.

For $\bestJ$, we know from \eqref{equ:rho_j_as_a_weighted_average} that $\rho_\bestJ \leq \gb_\bestJ \Rightarrow \gb_\bestJ \geq  \rho_\bestJ \geq \rho_{\bestJ - 1}$, therefore $\bestJ -1 < \hat{j}$. %
Given $\rho_\bestJ \geq \gb_{\bestJ+1}$, consider the two possible scenarios. 
\begin{enumerate}[$\bullet$]
	\item If $\rho_\bestJ > \gb_{\bestJ+1}$, we must have $\bestJ \geq \hat{j}$, thus $\bestJ = \hat{j}$ holds and 
$\bestJ$ is a maximizer of $\rho_j$.
	\item If $\rho_\bestJ = \gb_{\bestJ+1}$, we have $\bestJ < \hat{j}$. Moreover, \eqref{equ:rho_j_as_a_weighted_average} implies $\rho_{\bestJ + 1} = \rho_{\bestJ} = \gb_{\bestJ+1} > \gb_{\bestJ+2}$, which means $\bestJ + 1 \geq \hat{j}$. As a result, $\bestJ = \hat{j} - 1$, and $\bestJ$ is still a maximizer of $\rho_j$ because $\rho_\bestJ = \rho_{\bestJ + 1} = \rho_{\hat{j}}$. 
\end{enumerate}

This completes the proof of this lemma. 
\end{proof}

With Lemma~\ref{lem:single_driver_best_response_stationary} at hand, we now prove the result on the equilibrium outcome under a mechanism that dispatches every trip request to all drivers in the queue, uniformly at random. 

\propPureRandOpt*

\begin{proof}
We prove this result by showing that the equilibrium outcome under random dispatching has the same queue length $\equQL$ as that under direct FIFO, and that the same set of trips that are completed under direct FIFO is also completed under random dispatching. Theorem~\ref{thm:fifo_skip_second_best} then implies the same optimality results for random dispatching. 

We discuss the over-supplied and under-supplied settings separately.

\medskip

\noindent{}\emph{Case 1: $\lambda > \sum_{i \in \loc} \mu_i$.} When the platform is over-supplied, we have proved in Theorem~\ref{thm:fifo_skip_second_best} that all rider trips are completed under direct FIFO, and that the equilibrium queue length is $\equQL = \maxQL$ (as defined in \eqref{equ:max_eq_queue_length}). 
We now prove that under random dispatching, when the queue length is $\maxQL$, it is a Nash equilibrium for drivers to (i) join the queue with probability $\sum_{i \in \loc} \mu_i / \lambda$ upon arrival, (ii)  accept all trip dispatches while in the queue, and (iii) never move to the tail of the queue or leave the queue after joining.

More formally, we prove that the strategy $\sigmast = (\alphast, \betast, \gammast)$ defined as follows forms a Nash equilibrium among the drivers when the queue length is $\maxQL$:
\begin{align*}
	\alphast(\queue, \maxQL, i) &= 1, ~\forall i \in \loc, ~\forall \queue \in [0, ~\maxQL], \\
	\betast(\queue, \maxQL) & = 0, ~\forall \queue \in [0, ~\maxQL], \\
	\gammast(\queue, \maxQL) & = \pwfun{0, & \txtif \queue < \maxQL \\
		1 - \sum_{i \in \loc} \mu_i / \lambda, & \txtif \queue = \maxQL.}	
\end{align*}
When all drivers adopt strategy $\sigmast$, the length of the queue remains at $\maxQL$, since the numbers of drivers who join the queue and who are dispatched from the queue are both $\sum_{i \in \loc} \mu_i $ per unit of time. All rider trips are completed, implying the same steady state revenue and trip throughput as those under direct FIFO.

\smallskip

We now prove that $\sigmast$ forms a Nash equilibrium among the drivers under random dispatching when the queue length is $\maxQL$. First, observe that when the queue length is $\maxQL$ and when the rest of the driver adopts $\sigmast$, (i) each rider trip is dispatched once since the probability of declines is zero, and (ii) a driver's position in the queue has no impact on the rate at which she receives dispatches to each destination. This is therefore a stationary setting we have analyzed in Lemma~\ref{lem:single_driver_best_response_stationary}. 
For a driver anywhere in the queue, the rate at which she receives dispatches to each location $i \in \loc$ is:
\begin{align*}
	\eta_i = \mu_i / \maxQL. 
\end{align*}

Recall from \eqref{equ:max_eq_queue_length_appx} that $\maxQL$ can be rewritten as $ \maxQL = \left(\sum_{i =1}^\ell \gb_i \mu_i  \right) / \cost$.  
Therefore, the expected payoff $\rho_j$ from accepting only trips to locations $1$ through $j$ (as defined in \eqref{equ:util_top_j}) is of the form:
\begin{align*}
	\rho_j = \left(\sum_{i = 1}^j \gb_i \eta_i - \cost \right) \bigg/ \sum_{i = 1}^j \eta_i =  \left(\sum_{i = 1}^j \gb_i \mu_i - \cost \maxQL \right) \bigg/ \sum_{i = 1}^j \mu_i =  \left(\sum_{i = 1}^j \gb_i \mu_i - \sum_{i =1}^\ell  \gb_i \mu_i  \right) \bigg/ \sum_{i = 1}^j \mu_i. 
\end{align*}

This implies $\rho_\ell = 0$, and $\rho_j < 0$ for all $j < \ell$. 
By Lemma~\ref{lem:single_driver_best_response_stationary}, we know that the best acceptance strategy is to accept all dispatches, which is aligned with $\alphast$. This also implies that when all drivers adopt $\sigmast$,  the continuation payoff of a driver anywhere in the queue is  $\pi(\queue, \maxQL, \sigmast, \sigmast) = \rho_\ell = 0$.

The drivers' being indifferent towards being in the queue and leaving the queue %
means that there is no useful deviation from joining the queue with probability $\sum_{i \in \loc} \mu_i / \lambda$ (hence the probability of not joining the queue is $\gammast(\maxQL, \maxQL)  = 1 - \sum_{i \in \loc} \mu_i / \lambda$). Moreover, re-joining the queue at the tail is not useful since a driver's position in the queue has no impact on the rate at which the driver receives trip dispatches. This completes the proof that $\sigmast$ forms a Nash equilibrium among the drivers, thus concludes the discussion for Case 1, the over-supplied setting.

\medskip

\noindent{}\emph{Case 2: $\lambda \leq \sum_{i \in \loc} \mu_i$.} In the case without excess drivers, $\maxJ$ as defined in \eqref{equ:maxJ} denotes the index of the lowest-earning trip that is (partially) completed in equilibrium under the direct FIFO mechanism and the first best. 
We know from Theorem~\ref{thm:fifo_skip_second_best} that the equilibrium queue length under direct FIFO is $\equQL = \numSkip_\maxJ$, and the drivers complete all trips to locations $j < \maxJ$, and in each unit of time the drivers also complete $\lambda - \sum_{i = 1}^ {\maxJ-1} \mu_i$ out of the $\mu_\maxJ$ trips to location $\maxJ$.

We now prove that random dispatching achieves the same equilibrium outcome (queue length and set of trips completed). 
Fix the length of the queue at $\equQL = \numSkip_\maxJ$, and consider the strategy $\sigmast = (\alphast, \betast, \gammast)$ such that for all $\queue \in [0, ~\maxQL]$,
\begin{align}
	\alphast(\queue, \maxQL, i) &= \pwfun{1, & \txtif i < \maxJ, \label{equ:rand_under_supplied_alphast} \\
		1 -  \left( 1-  (\lambda - \sum_{i =1}^{\maxJ-1}  \mu_i )/\mu_\maxJ\right)^{1 / \patience}, & \txtif i = \maxJ, \\
		0, & \txtif i < \maxJ, }  \\
	\betast(\queue, \maxQL) & = 0, \label{equ:rand_under_supplied_betast}  \\
	\gammast(\queue, \maxQL) & =0. \label{equ:rand_under_supplied_gammast} 
\end{align}
For simplicity of notation, let $\theta_\maxJ \triangleq 1 -  \left( 1-  (\lambda - \sum_{i =1}^{\maxJ-1}  \mu_i )/\mu_\maxJ\right)^{1 / \patience}$. When every driver adopts strategy $\sigmast$, each trip to locations $i  < \maxJ$ is dispatched once, the trip to location $\maxJ$ is dispatched $\sum_{k = 1}^\patience (1-\theta_\maxJ)^{(k-1)}\theta_\maxJ k + (1 - \theta_\maxJ)^\patience \patience = (1 - (1 - \theta_\maxJ)^\patience) / \theta_\maxJ$ times, and each trip to locations $i > \maxJ$ is dispatched $\patience$ times. Given the queue length $\equQL = \numSkip_\maxJ$, the rate at which a driver anywhere in the queue receives trip dispatches to each location is:
\begin{align*}
	\eta_i = \pwfun{
		\mu_i / \numSkip_\maxJ, & \txtif i < \maxJ, \\
		\mu_i  (1 - (1 - \theta_\maxJ)^\patience) / (\theta_\maxJ \numSkip_\maxJ), & \txtif i = \maxJ, \\
		\mu_i \patience / \numSkip_\maxJ, & \txtif i > \maxJ.
	}
\end{align*}

As we observed in \eqref{equ:first_accpet_positions_appx}, $\numSkip_\maxJ = \sum_{i =1 }^{\maxJ-1} (\gb_i - \gb_\maxJ) \mu_i/\cost $. %
For each $j < \maxJ$, the expected payoff from accepting only the top $j$ trips can be written as:
\begin{align*}
	\rho_j = %
	\left(\sum_{i = 1}^j \gb_i \mu_i - \cost \numSkip_\maxJ \right) \bigg/ \sum_{i = 1}^j \mu_i = \left(\sum_{i = 1}^j \gb_i \mu_i - \sum_{i =1}^{\maxJ-1} (\gb_i - \gb_\maxJ) \mu_i  \right) \bigg/ \sum_{i = 1}^j \mu_i. 
\end{align*}
This implies that:
\begin{align*}
	\rho_{\maxJ - 1} = \left(\sum_{i = 1}^{\maxJ - 1} \gb_i \mu_i - \sum_{i =1}^{\maxJ - 1} (\gb_i - \gb_\maxJ) \mu_i  \right) \bigg/ \sum_{i = 1}^{\maxJ - 1} \mu_i = \gb_\maxJ.
\end{align*}

We know from \eqref{equ:rho_j_as_a_weighted_average} that $\rho_{\maxJ} = \gb_\maxJ$ must hold as well since $\rho_{\maxJ}$ is  a weighted average of $\rho_{\maxJ - 1}$ and $\gb_\maxJ$.Moreover, since $\gb_i$ is strictly decreasing in $i$, we have $\rho_{\maxJ - 1} < \gb_{\maxJ - 1}$. 
Applying Lemma~\ref{lem:single_driver_best_response_stationary}, we know that the highest possible expected payoff a driver may receive in this stationary setting is $\gb_\maxJ$, and this can be achieved by accepting all trips to location $ i < \maxJ$, and accepting trips to location $\maxJ$ with any probability in $[0, 1]$. 
$\alphast$ is therefore an optimal acceptance strategy. 
It is also straightforward to see that no strategy that involves not joining the queue the queue, and moving to the tail of the queue, or leave the queue without a rider trip, could achieve a higher expected payoff than $\gb_\maxJ$, thus $\sigmast$ forms a Nash equilibrium when the queue length is $\equQL = \numSkip_\maxJ$.

What is left to prove is that the length of the queue remains at $\equQL = \numSkip_\maxJ$ when all drivers adopt $\sigmast$. To show this, we prove that the rate at which drivers are dispatched from the queue is equal to $\lambda$, the rate at which drivers join the queue. 
First, all trips to locations $i \leq \maxJ - 1$ are accepted, %
so we only need to prove that $\lambda - \sum_{i =1}^{\maxJ-1}  \mu_i $ drivers accept trips to location $\maxJ$ per unit of time. 
For trips to location $\maxJ$, each time a trip is dispatched, it is \emph{not} accepted with probability $( 1-  (\lambda - \sum_{i =1}^{\maxJ-1}  \mu_i )/\mu_\maxJ)^{1 / \patience}$. Thus the probability for the trip to be unfulfilled after $\patience$ dispatches is $1-  (\lambda - \sum_{i =1}^{\maxJ-1}  \mu_i )/\mu_\maxJ$. This implies that the probability for a trip to location $\maxJ$ to be completed is $(\lambda - \sum_{i =1}^{\maxJ-1} \mu_i )/\mu_\maxJ$, so that $\lambda - \sum_{i =1}^{\maxJ-1}  \mu_i $ drivers accept trips to location $\maxJ$ per unit of time. 
This completes the proof of the under-supplied case, and concludes the proof of this proposition.
\end{proof}

\subsection{Optimality of Randomized FIFO} \label{appx:proof_thm_randFIFO_second_best}

In this section, we prove the optimality of the randomized FIFO mechanisms. 
We first provide the following lemma, which shows that the bins constructed as in \eqref{equ:defn_of_binLB} and \eqref{equ:defn_of_binUB} given any ordered partition are well-defined and not overlapping.

\begin{restatable}{lemma}{lemmaBinProperties} \label{lem:bin_properties}
For any ordered partition $(\loc \1, \loc\2, \dots, \loc\supK)$ of the top $\maxJ$ %
destinations $\{1, 2, \dots, \maxJ\}$, the corresponding set of bins satisfies: 
\begin{enumerate}[(i)]
	\item $0 \leq \binLB\supk \leq \binUB \supk$ for each $k = 1, \dots, \numBins$, and $\binUB \supk = \binLB\supk$ if $|\loc\supk| = 1$,
	\item $\binUB\supkmo < \binLB \supk$ for all $k = 2, 3, \dots, \numBins$. %
\end{enumerate}
\end{restatable}

\begin{proof}
For part (i), $\binUB\1 \geq \binLB\1 = 0$ trivially holds. For all $k = 2, 3, \dots, \numBins$, we have
\begin{align*}
	& \binUB \supk - \binLB \supk \\
	=&  \frac{1}{\cost} \left( \sum_{i \in \cup_{k' \leq k} \loc\supkprime} (\gb_i - \min_{i' \in \loc\supk} \{ \gb_{i'} \} )  \mu_i  \right) -
     \frac{1}{\cost}\left( \sum_{i \in \cup_{k' < k} \loc\supkprime }  \left( \gb_i-  \min_{i' \in \loc\supk} \{\gb_{i'}\} \right)  \mu_i \right)   \\
	= &  \frac{1}{\cost}  \left( \sum_{i \in \loc \supk} \left(\gb_i - \min_{i' \in \loc\supk} \{\gb_{i'} \} \right) \mu_i  \right) %
	\geq  0.
\end{align*}
Note that when $\loc\supk$ contains a single location, $\binLB\supk = \binUB\supk$ holds, meaning that for the $k\th$ time each trip is dispatched, the trip will be offered to the driver at position $\queue = \binLB\supk$ in the queue. 
This completes the proof of part (i). 

\medskip

For part (ii), first observe that for any $k > 1$, $\min_{i \in \loc \supkmo} \{ \gb_i \} > \min_{i \in \loc \supk} \{ \gb_i \}$, since the partition is ordered thus $\gb_i > \gb_j$ for all $i \in \loc \supkmo$ and all $j \in \loc \supk$. 
As a result, 
\begin{align*}
	\binUB \supkmo & = \frac{1}{\cost} \left( \sum_{i \in \cup_{k' < k} \loc \supkprime} \hspace{-0.2em} \left(\gb_i - \hspace{-0.2em} \min_{i' \in \loc \supkmo} \{\gb_{i'}\} \right) \mu_i \right) 
	< 
	\frac{1}{\cost} \left( \sum_{i \in \cup_{k' < k} \loc \supkprime} \left( \gb_i - \min_{i' \in \loc \supk } \{ \gb_{i'}\} \right) \mu_i  \right) 
	= \binLB \supk.
\end{align*}
This completes the proof of this lemma.
\end{proof}

\bigskip

We now prove the main result of our paper on the optimality of randomized FIFO.

\thmRandomFIFOSB*

\begin{proof} 
We first show that given a randomized FIFO mechanism corresponding to an ordered partition of the top $\maxJ$ locations, under the Nash equilibrium in steady state, (i) the length of the queue is equal to the equilibrium queue length under the direct FIFO mechanism, and (ii) the same set of trips completed under direct FIFO are also completed. Theorem~\ref{thm:fifo_skip_second_best} then implies that the equilibrium outcome under randomized FIFO is optimal. 
We also establish individual rationality and envy-freeness under randomized FIFO by showing that a driver's continuation payoff as a function of the driver's position in the queue is non-negative and monotonically non-increasing. %

\medskip

\todo{%
Also, this proof is over 10 pages long. Any idea about how to split it up into e.g. lemmas and claims?}

Recall that $\maxJ$ (defined in \eqref{equ:maxJ}) is the index of the lowest-earning trip that is (partially) completed in equilibrium under direct FIFO. We discuss the following cases: \vspace{-0.3em}
\begin{enumerate}[{Case} 1.]
	\item The total number of bins $\numBins = \min\{\maxJ, ~\patience\} = 1$, in which case %
	all trips are dispatched to drivers in the first bin. There are again two scenarios:
	\begin{enumerate}[{Case 1.}1]
		\item $\maxJ = 1$, and  in which case only trips to location $1$ are (partially) completed under the direct FIFO mechanism.
		\item $\maxJ > 1$, but $\patience = 1$, meaning that riders are impatient, and will cancel their trip request after any driver decline. 
	\end{enumerate}
	\item The number of bins  $\numBins  = \min\{\maxJ, ~\patience\} > 1$ , in which case trips may be dispatched multiple times, and we establish the equilibrium result by induction. 
\end{enumerate}

\bigskip

\noindent{}\textbf{Case 1.1: $\numBins = \maxJ = 1$.} In this case, there is a single partition under randomized FIFO: $\loc\1 = \{ 1 \}$, and we have $\binLB\1 = \binUB\1 = 0$. As a result, all trips are dispatched (only once) to the driver at the head of the queue. Recall that no driver will decline a trip to location $1$, since there is no other trip with better earnings that the driver would like to wait for. 
Consider two cases:
\begin{enumerate}[$\bullet$]
	\item When $\lambda \leq \mu_1$, it is straightforward to verify that %
	(i) the length of the queue is zero, and (ii) all drivers accept a trip to location $1$ immediately upon arrival, forms a Nash equilibrium among drivers in steady state. This is the same outcome as that under direct FIFO. %
	\item When $\lambda > \mu_1$ but $\maxJ = 1$, the number of locations must be $\ell = 1$, and all trips are accepted at the head of the queue. The equilibrium outcome is again the same as that under direct FIFO, where it is also the case that all trips are dispatched to and accepted by the driver at the head of the queue. %
	In steady state, drivers join the queue with probability $\mu_1/\lambda$, all trips are completed, and the length of the queue is $\equQL = \mu_1 \gb_1/\cost = \maxQL$ (at which point a driver is indifferent toward joining the queue and leaving without a rider trip). %
\end{enumerate}
Combining the two cases, we know that when $\maxJ = 1$, randomized FIFO achieves the same optimal outcome achieved by direct FIFO in equilibrium. Every driver gets a payoff of $\gb_1$ when $\lambda \leq \mu_1$, and when $\lambda > \mu_1$, the continuation payoff decreases linearly in the driver's position in the queue at takes value zero at $\maxQL$. The equilibrium outcome is, therefore, individually rational and envy-free. 

\bigskip

\noindent{}\textbf{Case 1.2: $\maxJ > 1$, $\numBins = \patience = 1$.} With $\patience = 1$, riders  cancel their trip requests after a single driver decline, and all trips are dispatched by the randomized FIFO mechanism to the first (and only) bin of drivers uniformly at random. $\loc\1 = \{1, \dots, \maxJ$, and the first bin is given by $\binLB\1 = 0$ and $\binUB\1 = \numSkip_\maxJ > 0$. There are two cases, depending on whether the queue is under or over-supplied. 

\medskip

\noindent{}\textit{Case 1.2.1: $\lambda \leq \sum_{i \in \loc} \mu_i$.} 
Assume that the length of the queue is $\equQL = \numSkip_\maxJ$. 
Under the randomized FIFO mechanism, all trips are randomly dispatched to drivers in $[\binLB\1, ~\binUB\1] = [0, \numSkip_\maxJ]$, i.e. all drivers in the queue. This is the same scenario as the under-supplied setting under random dispatching, which we analyzed in Case~2 in the proof of  Proposition~\ref{prop:pure_random}. 
It is straightforward to verify that the same strategy (specified by \eqref{equ:rand_under_supplied_alphast}, \eqref{equ:rand_under_supplied_betast} and \eqref{equ:rand_under_supplied_gammast}) forms a Nash equilibrium in steady state under randomized FIFO, and the equilibrium queue length remains constant at $\equQL = \numSkip_\maxJ$. We refer the readers to the proof of Proposition~\ref{prop:pure_random}, and do not repeat the same arguments here. Drivers at any position $\queue \in [0, \equQL]$ in the queue has the same continuation payoff $\gb_\maxJ \geq 0$, thus the equilibrium outcome is individually rational and envy-free. 

\medskip 

\noindent{}\textit{Case 1.2.2: $\lambda > \sum_{i \in \loc} \mu_i$.} When the queue is over-supplied, all trips are completed under direct FIFO, and $\maxJ = \ell$. 
The randomized FIFO mechanism dispatches each trip once (since $\patience = 1$) to drivers in $[\binLB\1, \binUB\1] = [0, \numSkip_\ell]$ uniformly at random.%
\footnote{The equilibrium queue length is $\equQL = \maxQL > \binUB\1$, as a result, drivers at positions $\queue \in (\binUB\1, \maxQL]$ do not receive any dispatches under randomized FIFO. This is, therefore, a different setting from the over-supplied setting under random dispatching (Case 1 of Proposition~\ref{prop:pure_random}), where trips are dispatched to all drivers in the queue  at random.}
Consider the strategy $\sigmast = (\alphast, \betast, \gammast)$:
\begin{align*}
	\alphast(\queue, \maxQL, i) & =  1, ~\forall i \in \loc, \forall \queue \in [0, \maxQL] \\
	\betast(\queue, \maxQL) & = 0, \forall \queue \in [0, \maxQL]  \\
	\gammast(\queue, \maxQL) & = \pwfun{0, & \txtif \queue < \maxQL, \\
		1 - \sum_{i \in \loc}\mu_i / \lambda, & \txtif \queue = \maxQL.}
\end{align*}
Here, $\gammast(\maxQL, \maxQL) = 1 - \sum_{i \in \loc}\mu_i / \lambda$ means that the drivers join the queue at the tail $\queue = \maxQL$ with probability $ \sum_{i \in \loc}\mu_i / \lambda$. It is clear that when $\sigmast$ is adopted by all drivers, the length of the queue remains at $\maxQL$. We now prove that with $\equQL = \maxQL$, it is a Nash equilibrium for all drivers to adopt strategy $\sigmast$. %
In other words, when the length of the queue is $\maxQL$ and when $\sigmast$ is adopted by the rest of the drivers, $\sigmast$ is a best response strategy for a driver at any $\queue \in [0, \maxQL]$

Let us first consider a driver at some position $\queue \in [0, \numSkip_\ell]$ in the queue. If the driver does not leave the queue or move to the tail of the queue, this again is a stationary environment analyzed in Lemma~\ref{lem:single_driver_best_response_stationary}. When every other driver adopts $\sigmast$, every trip is dispatched only once, thus the rate at which a driver receives trip dispatches to each location $i \in \loc$ is $\eta_i = \mu_i / \numSkip_\ell$. 
Since $\numSkip_\ell = \sum_{i =1 }^{\ell-1} (\gb_i - \gb_\ell) \mu_i/\cost = \sum_{i =1 }^\ell (\gb_i - \gb_\ell) \mu_i/\cost$ (see~\eqref{equ:first_accpet_positions_appx}), a driver's expected payoff from accepting only the top $j$ trips $\rho_j$ %
can therefore be rewritten as:
\begin{align*}
	\rho_j =  & %
	\left(\sum_{i = 1}^j \gb_i \mu_i - \cost \numSkip_\ell \right) \bigg/ \sum_{i = 1}^j \mu_i = \left(\sum_{i = 1}^j \gb_i \mu_i - \sum_{i =1}^\ell (\gb_i - \gb_\ell) \mu_i  \right) \bigg/ \sum_{i = 1}^j \mu_i \\
	= & \left( \sum_{i = j+1}^\ell (\gb_\ell -  \gb_i) \mu_i +  \gb_\ell \sum_{i =1}^j  \mu_i  \right) \bigg/ \sum_{i = 1}^j \mu_i = \gb_\ell  + \left( \sum_{i = j+1}^\ell (\gb_\ell -  \gb_i) \mu_i \right) \bigg/ \sum_{i = 1}^j \mu_i.
\end{align*}

Since $\gb_\ell - \gb_i \leq 0$ for all $i \in \loc$, %
$\rho_j$ is maximized at $j = \ell$, and we also have $\rho_\ell = \gb_\ell$. Lemma~\ref{lem:single_driver_best_response_stationary} implies that $\sigmast$, i.e. accepting all trips, is the optimal acceptance strategy for a driver at $\queue \in [0, \binUB\1]$. Under $\sigmast$, the 
continuation payoff is of the form:
\begin{align}
	\pi(\queue, \maxQL, \sigmast, \sigmast) = \rho_\ell = \gb_\ell, ~\forall \queue \in [0, \numSkip_\ell]. 
\end{align}
Since $\gb_\ell \geq 0$, there is no incentive for a driver to leave the queue without a rider trip, hence there is no useful deviation from $\gamma(\queue, \maxQL) = 0$. To see why moving to the tail of the queue is not a useful deviation either, observe that a driver will not receive any trip dispatch until she moves back to position $\binUB\1 =  \numSkip_\ell$ in the queue. Once the driver moves here (after incurring a non-negative waiting cost), the driver is in the exact same position as she is before, achieving an optimal payoff of $\gb_\ell$ when accepting all trips from the platform. This completes the proof that $\sigmast$ is a best response for a driver at some $\queue \in [0, \binUB\1]$ in the queue, when $\sigmast$ is adopted by the rest of the drivers.

\smallskip

Now consider any driver at some position $\queue \in (\numSkip_\ell, \maxQL]$ in the queue. The driver will not receive any trip dispatches until she reaches position $\numSkip_\ell$ in the queue, thus there is no useful deviation from the acceptance strategy $\sigmast$. From $\binUB\1$ onward, the driver gets an optimal continuation payoff of $\gb_\ell$ as we have shown above.
As a result, the driver's continuation payoff under $\sigmast$ is of the form:
\begin{align}
	\pi(\queue, \maxQL, \sigmast, \sigmast) = \gb_\ell - \cost (\queue - \numSkip_\ell)  / \sum_{i \in \loc} \mu_i,
	 ~\forall \queue \in (\binUB\1, \maxQL]. 
\end{align}
We can verify that $\pi(\queue, \maxQL, \sigmast, \sigmast) > 0$ for all $\queue < \maxQL$ and that $\pi(\queue, \maxQL, \sigmast, \sigmast) = 0$. Leaving the queue without a rider (and get zero) is therefore not a useful deviation. Moreover, at $\maxQL$ the drivers are indifferent towards being in the queue or leaving without a rider, thus a randomized joining decision is a best response. 
Individual rationality and envy-freeness both hold, since $\pi(\queue,\maxQL,\sigmast, \sigmast)$ is non-negative  and monotonically non-increasing for all $\queue \in \maxQL$. This completes Case 1.2.

\bigskip

\noindent{}\textbf{Case 2: $\numBins > 1$.} 
In this case %
we prove by induction on $k$ (the index of bins, starting from the first bin)
that in Nash equilibrium given the steady state queue length $\equQL$, drivers in the $k\th$ bin accept all dispatches for trips in the first $k$ partitions $\cup_{k' = 1}^k \loc\supkprime$, and decline all lower earning trips in $\cup_{k' > k} \loc\supkprime$. %
We then establish individual rationality and envy-freeness by checking that the continuation payoff is non-negative and monotonically non-increasing. 

\smallskip

Before analyzing the base case of the induction, we first provide some notations.
Denote $\minInBin\supk \in \loc$ and $\maxInBin\supk \in \loc$ as the lowest and highest indices of trips in the $k\th$ partition $\loc\supk$:
\begin{align}
	\minInBin\supk & \triangleq \min \{ i \in \loc | i \in \loc\supk \},  \label{equ:minInBin} \\
	\maxInBin\supk & \triangleq \max \{i \in \loc | i \in \loc \supk \}. \label{equ:maxInBin}
\end{align}
We know that a trip to location $\minInBin\supk$ (or $\maxInBin\supk$) is the highest (or lowest) paying trip in $\loc\supk$. 
Let $\equQL$ denote the equilibrium queue length under direct FIFO, i.e. $\equQL = \numSkip_\maxJ$ when $\lambda \leq \sum_{i\in \loc} \mu_j$, and $\equQL = \maxQL$ when $\lambda > \sum_{i \in \loc} \mu_j$. 
Let $\sigmast = (\alphast, \betast, \gammast)$ be a strategy given by:
\begin{enumerate}[$\bullet$]
	\item accepting all trips in the top $k$ partitions while in the $k\th$ bin in the queue, and randomize only on trips to location $\maxJ$ while in the last bin:  %
	\begin{align}
		\alphast(\queue, \equQL, i) & = \pwfun{ \one{ i \in \cup_{k' = 1}^k \loc\supkprime}, &\txtif \queue \in [\binLB\supk, \binUB\supk] \text{ for some } k \leq \numBins, \txtand i \neq \maxJ, \\
	\min\{(\lambda - \sum_{i < \maxJ} \mu_i), ~\mu_\maxJ \} / \mu_\maxJ, &\txtif \queue \in [\binLB\supK, \binUB\supK] \txtand  i = \maxJ,} \label{equ:rand_fifo_case_2_alphast}
	\end{align} 
	\item never move to the tail of the queue:
	\begin{align}
		\betast(\queue, \equQL) &= 	0, ~\forall \queue \in [0, \equQL],
	\end{align}
	\item never leave the queue without a trip after joining the queue, and join the queue with probability $\min\{ \sum_{i \in \loc}\mu_i / \lambda,~1\}$, i.e.
	\begin{align}
		\gammast(\queue, \equQL) &= \pwfun{
			0, & \txtif  \queue < \equQL, \\ 
			1 - \min\{ \sum_{i \in \loc}\mu_i / \lambda,~1\}, & \txtif \queue = \equQL. 
			 }
	\end{align}		
\end{enumerate} 
We now prove by induction that $\sigmast$ forms a Nash equilibrium among the drivers in steady state with queue length $\equQL$.

\bigskip

\noindent{}\textit{Step 1: the base case with $k=1$.} We first prove that when %
the length of the queue is $\equQL$ and when every other driver adopts strategy $\sigmast$, it is a best response for any driver in the first bin $[\binLB\1, \binUB\1] $ to also adopt strategy $\sigmast$.

The first bin consists of the top $\binUB\1$ drivers at the head of the queue. When $\loc \1 = \{1\}$, $\binUB\1 = \binLB\1 = 0$, thus all trips are first dispatched to the driver at the head of the queue. In this case, it is clear that accepting only trips to location $1$ is the best response for a driver at $\queue = 0$, and this is aligned with $\sigmast$. Now consider the case where $|\loc\1| > 1$, such that $\binUB\1 > 0$. 
With $\numBins > 1$, the equilibrium length of the queue $\equQL$  is above $\binUB\1$, thus the first bin is ``full''. 
For a driver at any position $\queue \in [0, \binUB\1]$, the rate at which she receives dispatches to each location $i \in \loc$ is 
\begin{align*}
	\eta_i\1 = \mu_i / \binUB\1.
\end{align*}
Note that the rates $\{\eta_i\1 \}_{i \in \loc}$ %
are independent to both the strategies adopted by the rest of the drivers in the first bin, and also the strategies employed by all drivers later in the queue.

We first prove that if a driver does not leave the queue or move to the tail of the queue, then there is no useful deviation from $\alphast(\queue, \equQL, i) = \one{i \in \loc\1}$, i.e. accepting all trips in $\loc\1$. This is a stationary setting that we have analyzed in Lemma~\ref{lem:single_driver_best_response_stationary}. 
Given \eqref{equ:defn_of_binUB}, we know that $\binUB\1$ is of the form: $\binUB\1 =  \frac{1}{\cost} \left(\sum_{i \in \loc\1} (\gb_i - \min_{i' \in \loc^{(1)}}\{ \gb_{i'} \}) \mu_i \right)$.
The utility for a driver in the first bin from accepting only the top $\maxInBin\1$ trips (as defined in \eqref{equ:util_top_j}) can therefore be written as:
\begin{align*}
	\rho\1_{\maxInBin\1} = & \left(\sum_{i \in \loc\1} \gb_i \eta_i\1 - \cost \right) \bigg/ \sum_{i \in \loc\1} \eta_i\1 
	= \left(\sum_{i \in \loc\1	} \gb_i \mu_i - \cost \binUB\1 \right) \bigg/ \sum_{i \in \loc\1} \mu_i \1 \\ 
	=& \left(\sum_{i \in \loc\1} \gb_i  - \left(\sum_{i \in  \loc\1} \left( \gb_i- \min_{i' \in \loc\1}\{ \gb_{i'} \} \right)  \mu_i  \right) \right) \bigg/ \sum_{i \in \loc\1} \mu_i \1 \\
	= & \min_{i \in \loc\1}\{ \gb_i \}.
\end{align*}

This implies $\rho\1_{\maxInBin\1} \leq \gb_i$ for all $i \in \loc\1$, and $\rho\1_{\maxInBin\1} > \gb_i$ for all $i \notin \loc\1$ (recall that the partitions are ordered). 
Lemma~\ref{lem:single_driver_best_response_stationary} then implies %
that %
an optimal acceptance strategy is to accept all trips to locations $1$ through $\maxInBin\1$, and this is aligned with $\sigmast$.
Lemma~\ref{lem:single_driver_best_response_stationary} also implies that the continuation payoff of any driver in the first bin given strategy $\sigmast$ is:
\begin{align}
	\pi(\queue, \equQL, \sigmast, \sigmast) = \min_{i \in \loc\1}\{ \gb_i \}, ~\forall \queue \in [0, \binUB\1].  \label{equ:case_2_pist_expression_1}
\end{align}

Since $\min_{i \in \loc\1}\{ \gb_i \} \geq 0$, deviating from $\gammast(\queue, \equQL) = 0$ and leaving the %
is not a useful deviation. Moreover, by moving to the tail of the queue, a driver will not receive any trip with net earnings higher than $\min_{i \in \loc\1}\{ \gb_i \}$, if the driver does not move all the way back to the first bin again. Once a driver is back to the first bin (after incurring some non-negative waiting costs), the driver is in the exact same situation as before moving to the tail, receiving trips again at rates $\{\eta_i\1 \}_{i \in \loc}$. %
Deviating from $\betast(\queue, \equQL) = 0$ is therefore not a useful strategy either. %
This implies that $\sigmast$ is a best response for drivers in the first bin, and completes the proof of the base case with $k = 1$.

\bigskip

\noindent{}\textit{Step 2: induction step for $1 < k < \numBins$.} Assume that when the length of the queue is $\equQL$, and when every other driver adopts strategy $\sigmast$, it is a best response for a driver at any position $\queue \in [0, \binUB\supkmo]$ to adopt strategy $\sigmast$. We now prove in this induction step, that it is also a best response for any driver at positions $\queue \in (\binUB\supkmo, \binUB\supk]$ to adopt strategy $\sigmast$. 

We take the following steps in proving this result:
\begin{enumerate}[{Step 2.}1]
	\item %
	Under strategy $\sigmast$, the  equilibrium continuation payoff $\pi(\queue, \equQL, \sigmast, \sigmast)$ is linearly decreasing in $\queue$ when $\queue \in [\binUB\supkmo, \binLB\supk]$, %
	and constant for $\queue \in [\binLB\supk, \binUB\supk]$: %
	\begin{align}
		\pi(\queue, \equQL, \sigmast, \sigmast) = \pwfun{ \min\limits_{i \in \loc \supkmo} \{\gb_i\} - \cost \left(\queue -  \binUB\supkmo \right) \bigg/ \sum\limits_{i \in \cup_{k' = 1}^{k-1} \loc \supkprime } \mu_i, &\txtif \queue \in [\binUB\supkmo, \binLB\supk], \\ 
			\min\limits_{i \in \loc \supk} \{\gb_i\}, & \txtif \queue \in [\binLB\supk, \binUB\supk].} \label{equ:proof_rand_fifo_step21}
	\end{align}
	\item %
	Under any feasible strategy $\sigma = (\alpha, \beta, \gamma)$ such that the driver does not leave the queue or move to the tail of the queue (i.e. if $\beta(\queue, \equQL) = \gamma(\queue, \equQL) = 0$ for all $\queue \in [\binUB\supkmo, \binUB\supk]$), the continuation payoff cannot exceed that under $\sigmast$:
	\begin{align*}
		\pi(\queue, \equQL, \sigma, \sigmast) \leq \pi(\queue, \equQL, \sigmast, \sigmast), ~\forall \queue \in [\binUB\supkmo, \binUB\supk]. 
	\end{align*}
	\item $\sigmast$ is a best response for drivers at positions $\queue \in [\binUB\supkmo, \binUB\supk]$ in the queue.
\end{enumerate}

Step 2.1 implies that for any driver at some $\queue \in [\binLB\supk, \binUB\supk]$, it cannot be a useful deviation from $\sigmast$ to accept any trip in later bins $\cup_{k' > k} \loc\supkprime$ since the driver gets $\pi(\queue, \equQL, \sigmast, \sigmast) \geq \min_{i \in \loc\supk} \gb_i > \max_{i \in \cup_{k' > k} \loc\supkprime } \gb_i$ from following strategy $\sigmast$. Moreover, all best-response strategies must have $\gamma(\queue, \equQL) = 0$ for all $\queue \in [\binUB\supkmo, \binUB\supk]$, because $\min_{i \in \loc \supk} \{ \gb_i\} > 0$ thus leaving the queue and getting zero cannot be a useful deviation.

With Step 2.2, we know that %
across all feasible strategies where the driver does not move to the tail of the queue, %
$\sigmast$ is a best strategy for drivers at any $\queue \in [\binUB\supkmo, \binUB\supk]$.
With 2.1 and 2.2, we prove the final step, that even when we consider all feasible strategies where people may move to the tail of the queue, there is still no strategy that results in a higher payoff than $\sigmast$.

We start from Step 2.1. 
 
\medskip

\noindent{}\textit{Step 2.1.} We prove \eqref{equ:proof_rand_fifo_step21} in this step. First, we show that under $\sigmast$, the continuation payoff of drivers satisfy $\pi(\binUB\supkmo, \equQL, \sigmast, \sigmast) = \min_{i \in \loc\supkmo} \gb_i$, and $\pi(\binLB\supk, \equQL, \sigmast, \sigmast) = \pi(\binUB\supk, \equQL, \sigmast, \sigmast) = \min_{i \in \loc\supk} \gb_i$.
For simplicity of notation, denote the continuation payoff under strategy $\sigmast$ given the equilibrium queue length $\equQL$ as:
\begin{align}
	\pist(\queue ) \triangleq \pi(\queue, \equQL, \sigmast, \sigmast). \label{equ:eq_cont_payoff_simplification}
\end{align}
Moreover, denote the total trip volume and average net earnings for a given subset of partitions as:
\begin{align}
	\tpRand_{k_1:k_2}  & \triangleq  \sum_{i \in \cup_{k' = k_1}^{k_2}  \loc\supkprime} \mu_i,  \label{equ:total_TP_top_k} \\
	\meangb _{k_1:k_2} & \triangleq \sum_{i \in \cup_{k' = k_1}^{k_2} \loc\supkprime}  \gb_i \mu_i \bigg/ \sum_{i \in \cup_{k' = k_1}^{k_2}  \loc\supkprime} \mu_i =  \sum_{i \in \cup_{k' = k_1}^{k_2} \loc\supkprime}  \gb_i \mu_i \bigg/  \tpRand _{k_1:k_2}.  \label{equ:mean_gb_top_k}
\end{align}

Consider now a driver who had just reached the position position $\binUB \supkmo$ in the queue. Under $\sigmast$, a driver at  $[0, \binUB\supkmo]$ in the queue only accept trips in the top $k-1$ partitions $\cup_{k' \leq k-1} \loc \supkprime$. When all drivers adopt the same strategy $\sigmast$, \emph{a priori} there is no difference in the waiting times or earnings from trips of any driver who are at $\binUB\supkmo$ in the queue. 
The average net earnings a driver at $\queue =  \binUB\supkmo$ will get from the trip she will accept in the future is therefore $\meangb \subokmo$. By Little's Law, the average amount of time the driver spends waiting in the queue is $\binUB\supkmo / \tpRand \subokmo$.  
Thus the average continuation payoff for the driver at $\queue = \binLB\supk$ is: 
\begin{align*}
	\pist(\binUB\supkmo) =  \meangb \subokmo - \cost \binUB\supkmo / \tpRand \subokmo. 
\end{align*}
Given $\binUB\supk$ %
as defined in \eqref{equ:defn_of_binUB},  we know:
\begin{align}
	& \pist(\binUB\supkmo) 
	= \meangb \subokmo - \left( \sum_{i \in \cup_{k' < k} \loc\supkprime }  \mu_i  \left( \gb_i-  \min_{i' \in \loc \supkmo} \{\gb_{i'}\} \right) \right) \bigg/ \tpRand \subokmo 
	= \min_{i \in \loc\supkmo} \{\gb_i\}. \label{equ:cont_earning_at_UB}
\end{align}

Similarly, by reasoning about the net earnings an average driver gets from an average trip, and the average waiting cost a driver incurs, we can show that $\pi(\binLB\supk, \equQL, \sigmast, \sigmast) = \pi(\binUB\supk, \equQL, \sigmast, \sigmast) = \min_{i \in \loc\supk} \gb_i$. 
Under $\sigmast$, drivers at some position $\queue \in (\binUB\supkmo, \binLB\supk)$ will wait for $(\queue - \binLB\supkmo)/ \tpRand  \subokmo$ units of time before reaching $\binUB\supkmo$ in the queue, therefore her continuation payoff is of the form:
\begin{align}
	\pist(\queue) = \min_{i \in \loc \supkmo} \{\gb_i\} - \cost \left(\queue -  \binUB\supkmo \right) \bigg/\tpRand \subokmo , &\txtif \queue \in [\binLB\supkmo, \binLB\supk). \label{equ:pist_induction_step_1}
\end{align}
It is straightforward to verify that $\pist$ is left continuous at $\binLB\supk$:
$$
	\lim_{\queue \rightarrow \binLB\supk-}\pist(\queue) = \min_{i \in \loc \supkmo} \{\gb_i\} - \cost \left(\binLB\supk-  \binUB\supkmo \right) \bigg/\tpRand \subokmo = \min_{i \in \loc\supk} \{\gb_i\}.
$$	

What is left to prove for Step 2.1 is that $\pist(\queue)$ remains constant where $\queue \in [\binLB\supk, \binUB\supk]$. This is trivial when $|\loc\supk| = 1$, in which case $\binUB\supk = \binLB\supk$. Therefore, we now consider the case where $|\loc\supk| > 1$ such that $\binUB\supk > \binLB\supk$. 
When all drivers adopt strategy $\sigmast$, all trips in the first $k-1$ partitions $\cup_{k' = 1}^{k-1} \loc\supkprime$ are accepted before reaching the $k\th$ bin. For a driver in the $k\th$ bin, the rate at which she receives trip dispatches to each location $i \in \loc$ is therefore:
\begin{align}
	\eta \supk _i = \pwfun{0, & \txtif i \in \cup_{k' = 1}^{k-1} \loc\supkprime \\
		\mu_i / (\binUB\supk - \binLB\supk), & \txtif i \in \cup_{k' \geq k} \loc\supkprime
		} \label{equ:bin_k_offer_rates}
\end{align}
Note that the rates $\{\eta\supk_i\}_{i \in \loc}$ are independent to the strategies taken by drivers in later bins of the queue.
With a slight abuse of notation, let
\begin{align*}
	\eta \supk \triangleq \sum_{i \in \loc \supk} \eta\supk_i
\end{align*}
be the total rate at which drivers in the $k\th$ bin receives trips in $\loc\supk$.

Fix an arbitrary point in time and call it time $t = 0$, and consider a driver who is at position $\binUB\supk$ at time $t = 0$. Let $\goft(t)$ be the position of the driver in the queue, if the driver has not yet accepted a trip and leave the queue. We know $\goft(0) = \binUB\supk$. %
Before the driver reaches position $\binLB\supk$ in the queue, we know that in the next $\dd t$ units of time, when every other driver adopts strategy $\sigmast$, there are $\tpRand \subokmo \dd t$ drivers who are dispatched from queue positions earlier than $\binLB\supk$, and there are $\dd t \tpRand \subkk (\goft(t) - \binLB\supk) / (\binUB\supk - \binLB\supk) $ drivers who are dispatched from the $k\th$ bin, ahead of this driver. As a result, the time derivative of the driver's queue position satisfies
\begin{align}
	\frac{\dd \goft(t)}{\dd t} = - \tpRand \subokmo -  \tpRand \subkk \frac{\goft(t) - \binLB\supk}{\binUB\supk - \binLB\supk}, %
	\label{equ:d_position_dt}
\end{align}
i.e. the driver moves forward in the queue at a rate of $ \tpRand \subokmo + \tpRand \subkk (\goft(t) - \binLB\supk)/(\binUB\supk - \binLB\supk)$ positions per unit of time. Since $\tpRand \subokmo > 0$, we know that the driver will reach $\binLB\supk$ within finite time.

$\pist(\goft(t))$ denotes continuation payoff of this driver as a function of time.  
For a driver at some position $\goft(t) \in (\binLB\supk, ~\binUB\supk]$ at time $t$, the probability for the driver to be dispatched a trip she will accept under $\sigmast$ in the next $\dd t$ units of time is $\eta\supk \dd t$. If the driver is not dispatched, she moves forward in the queue to position $\goft(t + \dd t)$ after incurring a cost of $\cost \dd t$. If the driver is dispatched, she takes a trip with an average net earnings of $\meangb \subkk$ after incurring a waiting cost in the order of $\cost O(\dd t)$. %
Therefore, the driver's continuation payoff as a function of time $t$ can be written as:
\begin{align}
	\pist(\goft(t)) = (1 - \eta \supk \dd t) (\pist(\goft(t + \dd t)) - \cost \dd t)  + \eta \supk \dd t  ( \meangb \subkk - \cost O(\dd t)) \label{equ:pi_kth_bin_recursion}
\end{align}
Reorganizing \eqref{equ:pi_kth_bin_recursion}, and taking the limit as $\dd t \rightarrow 0$, we have
\begin{align}
	\frac{\dd \pist(\goft(t))}{\dd t} = \cost +  \eta\supk \left( \pist(\goft(t))  - \meangb \subkk \right) =  \eta\supk  \left( \pist(\goft(t))  - \left( \meangb \subkk  - \cost / \eta \supk \right) \right),  \label{equ:dpi_dt}
\end{align}
and this implies
\begin{align}
	\pist(\goft(t)) = \meangb \subkk  - \cost / \eta \supk + C e^{ \eta\supk t}, \label{equ:pi_goft}
\end{align}
where $C$ is some constant. Given that $\goft(0) = \binUB\supk$ and $\pi(\binUB\supk) = \min_{i \in \loc\supk} \{\gb_i\}$, we have:
\begin{align*}
	 \pist(\goft(0)) = \min_{i \in \loc\supk} \{\gb_i\} = \meangb \subkk  - \cost / \eta \supk + C. 
\end{align*}

Given \eqref{equ:defn_of_binLB} and \eqref{equ:defn_of_binUB}, the size of the $k\th$ bin is:
\begin{align}
	\binUB\supk - \binLB\supk = \frac{1}{\cost} \sum_{i \in \loc \supk} (\gb_i - \min\limits_{i' \in \loc^{(k)}}\{ \gb_{i'} \})  \mu_i.  \label{equ:bin_size_lb}
\end{align}
$\meangb \subkk  - \cost / \eta \supk$ therefore satisfies
\begin{align*}
	\meangb \subkk  - \cost / \eta \supk = \left(\sum_{i  \in \loc \supk} \gb_i \mu_i - \cost (\binUB\supk - \binLB\supk) \right) \bigg/ \tpRand \subkk = \min_{i \in \loc \supk} \{\gb_{i}\}.
\end{align*}
As a result, $C = 0$ must hold, meaning that $\pist(\queue)$ remains constant with respect to $t$ for all $t$ such that $\goft(t) \leq \binUB\supk$ and $\goft(t) \geq  \binLB\supk$. This completes the proof that $\pist(\queue) = \meangb \subkk  - \cost / \eta \supk  =  \min_{i \in \loc \supk} \{\gb_{i}\}$ for all $\queue \in (\binLB\supk, \binUB\supk]$. %

\smallskip

This completes the proof of Step 2.1. What we know from this step and Lemma~\ref{lem:best_response_necessary_conditions} is that for any driver at some $\queue \in [\binLB\supk, \binUB\supk]$, it cannot be a useful deviation from $\sigmast$ to accept any trip in later bins $\cup_{k' > k} \loc\supkprime$ since the driver gets $\pi(\queue, \equQL, \sigmast, \sigmast) \geq \min_{i \in \loc\supk} \gb_i > \max_{i \in \cup_{k' > k} \loc\supkprime } \gb_i$ from following strategy $\sigmast$. 
Moreover, all best-response strategies must have $\gamma(\queue, \equQL) = 0$ for all $\queue \in [\binUB\supkmo, \binUB\supk]$, because $\min_{i \in \loc \supk } \{ \gb_i\} \geq 0$ thus leaving the queue and getting zero cannot be a useful deviation.

\medskip

\noindent{}\textit{Step 2.2.} 
We now prove that under any feasible strategy $\sigma = (\alpha, \beta, \gamma)$ such that the driver does not leave the queue or move to the tail of the queue (i.e. if $\beta(\queue, \equQL) = \gamma(\queue, \equQL) = 0$ for all $\queue \in [\binUB\supkmo, \binUB\supk]$), for any position $\queue \in  [\binUB\supkmo, \binUB\supk]$,
\begin{align}
	\pi(\queue, \equQL, \sigma, \sigmast) \leq \pi(\queue, \equQL, \sigmast, \sigmast). \label{equ:proof_rand_fifo_step22}
\end{align}

First, \eqref{equ:proof_rand_fifo_step22} is straightforward to establish for $\queue \in [\binUB\supkmo, \binLB\supk)$, since for a driver who %
does not leave or move to the tail of the queue, the driver will %
wait in line until she reaches position $\binUB\supkmo$ in the queue, and this is aligned with $\sigmast$. 
Once a driver is at $\binUB\supkmo$, the best strategy moving forward is $\sigmast$ (by induction assumption). As a result, it is impossible to achieve a better continuation payoff than that under $\sigmast$.

Now consider $\queue \in [\binLB\supk, \binUB\supk]$, and there are two cases depending on whether $|\loc \supk| = 1$. When $|\loc\supk| = 1$, Lemma~\ref{lem:bin_properties} implies that $\binUB\supk = \binLB\supk$, thus all trips in $\cup_{k' \geq k} \loc \supkprime$ are dispatched to the driver at position $\binLB\supk$ in the queue. Under $\sigmast$, a driver at $\queue = \binLB\supk$ gets a continuation payoff of $\gb_i$ where $i \in \loc\supk$ is the only trip in the $k\th$ partition, regardless of whether the driver accepts a trip or moved forward in the queue. An argument very similar to the proof of the induction step of Lemma~\ref{lem:strict_fifo_spe_formal} shows no alternative strategy may achieve a higher continuation payoff.

What is left to study in the case where $|\loc\supk| > 1$ and $\binUB\supk - \binLB\supk > 0$. 
Assume towards a contradiction that there exists some $\queue \in (\binLB\supk, \binUB\supk]$ such that $\pi(\queue, \equQL, \sigma, \sigmast) > \pi(\queue, \equQL, \sigmast, \sigmast)  = \min_{i \in \loc\supk}\gb_i$. 
We introduce the following notation:
\begin{enumerate}[$\bullet$]
	\item If the driver did not accept any trip dispatches under $\sigma$ before reaching position $\binLB\supk$ in the queue, denote the time it takes for the driver to move from $\queue$ to $\binLB\supk$ as $\kappa(\queue)$. 
	\item Denote the probability for the driver to be dispatched a trip she is willing to accept under $\sigma$ before the driver reaches $\binLB\supk$ (i.e. within the next $\kappa(\queue)$ units of time, given strategy $\sigma$) as $\xi(\kappa(\queue))$. 
	\item %
	Conditioning on a driver's receiving a trip within the $\kappa(\queue)$ units of time while following strategy $\sigma$, let $\omega(\kappa(\queue))$ be the driver's expected payoff, which includes both the net earnings from the trip the driver accepts and the waiting cost the driver incurs.
\end{enumerate}
The driver's continuation payoff at position $\queue$ under strategy $\sigma$ can therefore be written as:
\begin{align}
	\pi(\queue, \equQL, \sigma, \sigmast) = & \xi(\kappa(\queue))  \omega(\kappa(\queue))  + \left(1- \xi(\kappa(\queue)) \right) \left( \pi(\binLB\supk, \equQL,\sigma, \sigmast) %
	- \cost \kappa(\queue) \right). \label{equ:cont_earning_from_q_rand}
\end{align}
$\pi(\binLB\supk, \equQL,\sigma, \sigmast)$ shows up in the second term because if a driver did not accept a dispatch before time $\kappa(\queue)$ had passed (starting from the time of her being at position $\queue$), the driver has now reached $\binLB\supk$. 
We have just argued that this continuation payoff is bounded by $\pi(\binLB\supk, \equQL,\sigma, \sigmast) \leq \pist(\binLB\supk) = \min_{i \in \loc\supk}\gb_i$. When $\xi(\kappa(\queue)) = 0$, $\pi(\queue, \equQL, \sigma, \sigmast)  \leq \min_{i \in \loc\supk}\gb_i$ trivially holds. When $\xi(\kappa(\queue))  > 0$, 
combining \eqref{equ:cont_earning_from_q_rand} and the assumption that $\pi(\queue, \equQL, \sigma, \sigmast) > \min_{i \in \loc\supk}\gb_i$, we get
\begin{align}
	 \omega(\kappa(\queue)) >
	 \min_{i \in \loc\supk}\{\gb_i\} + 
	 (1 - \xi(\kappa(\queue)) ) \cdot \cost \kappa(\queue)/\xi(\kappa(\queue)).  \label{equ:expected_utility_if_dispatched}
\end{align}

Observe that in the first $\kappa(\queue)$ units of time, the driver receives trip dispatches at rates $\{ \eta_i \supk \}_{i \in \loc}$. Now consider a stationary setting that we analyzed in Lemma~\ref{lem:single_driver_best_response_stationary}, where a driver always receives trip dispatches at rates $\{ \eta_i \supk \}_{i \in \loc}$.\footnote{In other words, we allow the driver to remain in the $k\th$ bin forever, instead of forcing her to move past $\binLB\supk$.}
If the driver employs strategy $\sigma$ (restricted to the first $\kappa(\queue)$ units of time starting from position $\queue$ in the queue) in this stationary setting, the driver's expected utility, which we denote as $\pihat(\sigma)$, can be written as:
\begin{align*}
	\pihat(\sigma) = & \xi(\kappa(\queue))  \omega(\kappa(\queue))  + (1 - \xi(\kappa(\queue))) (\pihat(\sigma) - \cost \kappa(\queue)).
\end{align*}
Intuitively, if the driver gets dispatched in the first $\kappa(\queue)$ units of time given strategy $\sigma$, she gets an expected payoff of $\omega(\kappa(\queue))$, and this happens with probability $\xi(\kappa(\queue))$. 
If the driver is not dispatched in the first $\kappa(\queue)$ units of time, the driver's continuation payoff starting from that point of time onward is again $\pihat(\sigma)$. Reorganizing this expression, and applying \eqref{equ:expected_utility_if_dispatched}, we get:
\begin{align*}
 \xi(\kappa(\queue)) \pihat(\sigma) = & \xi(\kappa(\queue))  \omega(\kappa(\queue))  - (1 - \xi(\kappa(\queue)))  \cost \kappa(\queue) \\
	> & \xi(\kappa(\queue)) \min_{i \in \loc\supk}\{\gb_i\} +  (1 - \xi(\kappa(\queue)) ) \cdot \cost \kappa(\queue) - (1 - \xi(\kappa(\queue)))  \cost \kappa(\queue) \\
	=& \xi(\kappa(\queue)) \min_{i \in \loc\supk}\{\gb_i\}. 
\end{align*}

This implies $\pihat(\sigma) > \min_{i \in \loc\supk} \{\gb_i\}$, meaning that there exists a strategy for a driver to get a continuation payoff strictly above $ \min_{i \in \loc\supk} \{\gb_i\}$ in the stationary setting where the driver always receives trip dispatches at rate $\{ \eta_i \supk \}_{i \in \loc}$ given by \eqref{equ:bin_k_offer_rates}. We now prove that this is not possible, and as a result we have a contradiction.

Given \eqref{equ:bin_k_offer_rates} and \eqref{equ:bin_size_lb}, for a driver who receives trip dispatches at rates $\{ \eta_i \supk \}_{i \in \loc}$, the expected utility from accepting top $j$ trips (as defined in \eqref{equ:util_top_j}) for some $j \geq \minInBin\supk$ can be rewritten as: %
\begin{align}
	\rho_j\supk =  & \left(\sum_{i \leq j} \gb_i \eta_i \supk - \cost \right) \bigg/ \sum_{i \leq j} \eta_i \supk 
	= \left(\sum_{i = \minInBin\supk}^j \gb_i \mu_i - \cost (	\binUB\supk - \binLB\supk) \right) \bigg/ \sum_{i = \minInBin\supk}^j \mu_i  \label{equ:util_top_j_kth_bin} \\
	= & \left(\sum_{i = \minInBin\supk}^j \gb_i \mu_i - \left( \sum_{i \in \loc \supk} (\gb_i - \min\limits_{i' \in \loc\supk}\{ \gb_{i'} \})  \mu_i  \right) \right) \bigg/ \sum_{i = \minInBin\supk}^j \mu_i. \label{equ:util_top_j_kth_bin_ub}
\end{align}
Recall that $\minInBin\supk$ and $\maxInBin\supk$ are the lower-index and highest-index trips in $\loc\supk$, respectively. When $j \leq \maxInBin\supk$, 
\begin{align*}
	\rho_j\supk = \min\limits_{i' \in \loc\supk}\{ \gb_{i'} \} - \left(\sum_{i = j+1}^{\maxInBin\supk} ( \gb_i - \min\limits_{i' \in \loc\supk}\{ \gb_{i'} \})  \mu_i  \right) \bigg/ \sum_{i = \minInBin\supk}^j \mu_i \leq \min\limits_{i' \in \loc\supk}\{ \gb_{i'} \}.
\end{align*}
When $j \geq \maxInBin\supk$, we also have:
\begin{align*}
	\rho_j\supk =  \left(\sum_{i = \minInBin\supkpo}^j \gb_i \mu_i +  \min\limits_{i' \in \loc\supk}\{ \gb_{i'} \}   \sum_{i \in \loc \supk}  \mu_i  \right) \bigg/ \sum_{i = \minInBin\supk}^j \mu_i \leq \min\limits_{i' \in \loc\supk}\{ \gb_{i'} \}.
\end{align*}
As a result, $\rho_j\supk$ is optimized at $j = \maxInBin\supk$, taking value $\min_{i \in \loc \supk }\{\gb_{i}\}$. Lemma~\ref{lem:single_driver_best_response_stationary} then implies that a driver in such a stationary setting cannot achieve a utility strictly higher than $\min_{i \in \loc \supk }\{\gb_{i}\}$.
This completes the proof of Step 2.2. 

\medskip

\newcommand{\qhat}{\hat{\queue}}

\noindent{}\emph{Step 2.3.} 
We now complete the induction step by proving that $\sigmast$ is a best response for drivers at any position $\queue \in [\binUB\supkmo, \binUB\supk]$ in the queue. %
We prove this by contradiction. Assume %
that there exists a strategy $\sigma$ such that $\pi(\queue, \equQL, \sigma, \sigmast) > \pist(\queue)$ for some $\qhat \in [\binUB\supkmo, \binUB\supk]$. 
We show a contradiction with the following steps.
\begin{enumerate}[(i)]
    \item We first argue that it is %
    without loss of generality to restrict our analysis to strategies %
    such that the driver does not leave the queue, i.e. $\sigma = (\alpha, \beta, \gamma)$ for which  $\gamma(\queue, \equQL) = 0$, $\forall \queue \in [\binUB\supkmo, \binUB\supk]$. 
    This is because if we have a strategy $\sigma$ where 
    the driver leaves with a non-zero probability at some $\queue \in [\binUB\supkmo, \binUB\supk]$, we may construct an alternative strategy where instead of leaving the queue at $\queue$, the driver stays in the queue and follows $\sigmast$ from then on. Step 2.1 implies that this is an improvement, since the driver gets a continuation payoff of $\pist(\queue) \geq \min_{i \in \loc\supk} \{\gb_i\} > 0$ instead of $0$. Thus we get an alternative strategy that improves over $\sigmast$, and also satisfies $\gamma(\queue, \equQL) = 0$ for all $\queue \in [\binUB\supkmo, \binUB\supk]$. This contradicts Step 2.2. 
    \item We now prove that the continuation payoff at the tail of the queue under $\sigma$ must satisfy $\pi(\equQL, \equQL, \sigma, \sigmast) > \min_{i \in \loc\supk}\{\gb_i\}$. First, there must exist some $\qtilde \in [\binUB\supkmo, \binUB\supk]$ such that $\beta(\qtilde, \equQL) > 0$. Otherwise, given (i), $\sigma$ is a strategy such that $\gamma(\queue, \equQL) = \beta(\queue,\equQL) = 0$ for all $\queue \in [\binUB\supkmo, \binUB\supk]$, and we have proved in Step 2.2 that among all such strategies, $\sigmast$ is a best response.
    Assuming now towards a contradiction, that $\pi(\equQL, \equQL, \sigma, \sigmast) \leq \min_{i \in \loc\supk}\{\gb_i\}$. $\pist(\queue) \geq \min_{i \in \loc\supk}\{\gb_i\}$ (from Step 2.1) implies that it is a (weak) improvement if instead of moving to the tail of the queue, the driver remains in the queue and adopts $\sigmast$ from then on. This, again, is a strategy with $\gamma(\queue, \equQL) = \beta(\queue,\equQL) = 0$ for all $\queue \in [\binUB\supkmo, \binUB\supk]$, thereby contradicting Step 2.2.
        \item With $\pi(\equQL, \equQL, \sigma, \sigmast) > \min_{i \in \loc\supk}\{\gb_i\}$, we claim that
    \begin{align}
    	\pi(\binUB\supk, \equQL, \sigma, \sigmast) > \pi(\equQL, \equQL, \sigma, \sigmast) > \min_{i \in \loc\supk}\{\gb_i\}. \label{equ:proof_rand_fifo_binub_better_than_tail}
    \end{align}
    First, observe that a driver at the tail of the queue $\queue = \equQL$ will not receive any trip with net earnings weakly above $\min_{i \in \loc \supk} \{\gb_i\}$ until the driver reaches position $\binUB\supk$ in the queue. In the scenarios where the driver is dispatched under $\sigma$ before reaching $\binUB\supk$, the driver's payoff is strictly below $\min_{i \in \loc \supk} \{\gb_i\}$.
    $\pi(\equQL, \equQL, \sigma, \sigmast)$ is a weighted average of 
        (I) the payoff the driver gets from being dispatched before reaching $\binUB\supk$, and 
        (II) the continuation payoff after reaching $\binUB\supk$ $\pi(\binUB\supk, \equQL, \sigma, \sigmast)$, minus the waiting cost a driver incurs before reaching $\binUB\supk$.
    Therefore, we must have $\pi(\binUB\supk, \equQL, \sigma, \sigmast) > \pi(\equQL, \equQL, \sigma, \sigmast)$ in order for $\pi(\equQL, \equQL, \sigma, \sigmast) > \min_{i \in \loc\supk}\{\gb_i\}$ to hold. 
    \item It is without loss of generality to assume that there exists $\qtilde \in [\binUB\supkmo, \binUB\supk]$ such that $\beta(\qtilde, \equQL) = 1$, i.e. the driver always moves back to the tail of the queue at $\qtilde$. First, observe that $\pi(\qtilde, \equQL, \sigma, \sigmast) \leq \pi(\equQL, \equQL, \sigma, \sigmast)$ must hold for some $\qtilde \in [\binUB\supkmo, \binUB\supk]$--- otherwise, reducing $\beta(\queue, \equQL)$ to zero for all $\queue \in [\binUB\supkmo, \binUB\supk]$ is a weak improvement, again contradicting Step 2.2. Now, increasing $\beta(\qtilde, \equQL)$ to $1$ for one such $\qtilde \in [\binUB\supkmo, \binUB\supk]$ will be a weak improvement over $\sigma$, thus in this way, we've constructed a strategy that achieves a better continuation payoff than $\sigmast$ at some point, with $\qtilde \in [\binUB\supkmo, \binUB\supk]$ for some $\beta(\qtilde, \equQL) = 1$.
    \item Now consider an alternative setting, where the driver follow strategy $\sigma$, except that the driver always moves back to $\binUB\supk$ instead of the tail of the queue, whenever the driver is prescribed to move to the tail of the queue under $\sigma$.\footnote{This is not allowed under randomized FIFO--- we construct this scenario for the purpose of this proof only.} 
    Denote this new strategy as $\sigma'$.
    \eqref{equ:proof_rand_fifo_binub_better_than_tail} implies that this will be an improvement, such that the continuation payoff under this new setting, which we denote as $\pihat$, must also satisfy $\pihat(\binUB\supk, \equQL, \sigma', \sigmast) > \min_{i \in \loc\supk}\{\gb_i\}$.

    This is, however, not possible. Observe that in this alternative setting, under $\sigma'$, the driver at $\binUB\supk$ will either accept a trip and leave the queue before reaching $\qtilde \in [\binUB\supkmo, \binLB\supk]$, or move back to $\binUB\supk$ before reaching $\qtilde$ or at $\qtilde$. As a result, the driver will either be receiving trip dispatches at rates $\{ \eta_i \supk \}_{i \in \loc}$ defined in \eqref{equ:bin_k_offer_rates} (when the driver is at some position in $[\binLB\supk, \binUB\supk]$), or not receive any trip dispatches at all (when the driver is in $[\binUB\supkmo, \binLB\supk)$). The driver's payoff is therefore upper bounded by the scenario where she is in a stationary setting, always receiving trip dispatches at rates $\{ \eta_i \supk \}_{i \in \loc}$, but we have proved in Step 2.2 that the highest achievable expected payoff in this setting is $\min_{i \in \loc\supk}\{\gb_i\}$. 
    
    This is a contradiction, and concludes the proof of the induction step, that $\sigmast$ is a best response for a driver at any position $[\binUB\supkmo, \binUB\supk]$ in the queue. 
\end{enumerate}

\bigskip

\noindent{}\emph{Step 3: The last bin $k = \numBins$ and beyond.} 
Given Steps 1 and 2, we know that $\sigmast$ is a best response for a driver at any position $\queue \leq \binUB\supKmo$ in the queue. What is left to prove that $\sigmast$ is also a best response for any driver at $\queue \in (\binUB\supKmo, \equQL]$ in the queue. 

First, with the same arguments as in Step 2.1, we can show that %
\begin{align}
    \pi(\queue, \equQL, \sigmast, \sigmast) = 
	      \min_{i \in \loc \supKmo} \{\gb_i\}  - \cost (\queue - \binUB\supKmo) / \tpRand_{1:\numBins - 1}, &\txtif \queue \in [\binUB\supKmo, \binLB\supK], 
\end{align}
and that
\begin{align}
    \pi(\binLB\supK, \equQL, \sigmast, \sigmast)  = \pi(\binUB\supK, \equQL, \sigmast, \sigmast) = \min_{i \in \loc\supK} \{\gb_i\} = \gb_\maxJ.
\end{align}

In the case where $|\loc \supK| = 1$, i.e. when $\loc \supK = \{ \maxJ \}$, Lemma~\ref{lem:bin_properties} implies that $\binUB\supK = \binLB\supK = \numSkip_\maxJ$. 
Under $\sigmast$, all trips to locations $j < \maxJ$ are accepted by drivers in the top $\numBins - 1$ bins, and all trips to locations $ j \geq \maxJ$ are dispatched (for the last time) to drivers at position $\numSkip_\maxJ$ in the queue, where drivers accept only trips to location $\maxJ$. The equilibrium queue length is $\equQL = \numSkip_\maxJ$ when $\lambda \leq \sum_{i \in \loc} \mu_i$. When $\lambda > \sum_{i \in \loc} \mu_i$, the equilibrium queue length is $\equQL = \maxQL$, and $\pist(\queue)$ decreases linearly in $\queue$ when $\queue \geq \numSkip_\maxJ = \numSkip_\ell$, with $\pist(\maxQL) = 0$ at the tail of the queue. Using the same arguments as those in the proof of Lemma~\ref{lem:strict_fifo_spe_formal}, we can show that when the length of the queue is $\equQL$ and when the rest of the drivers adopt $\sigmast$, it is also a best response for a driver at any $\queue \in [\binUB\supKmo, \equQL]$ to adopt $\sigmast$. We do not repeat the same reasoning here.

\medskip

What is left to analyze is the case where  $|\loc \supK| > 1$, in which case $\binUB\supK - \binLB\supK  > 0$. 
When all drivers adopt strategy $\sigmast$, for any driver in the last bin $[\binLB\supK, \binUB\supK]$, the rates at which the driver receives trip dispatches to each location are given by
\begin{align}
	\eta \supK_i = \pwfun{
		0, & \txtif i \in \cup_{k' = 1}^{\numBins - 1} \loc\supkprime, \\
		\mu_i / (\binUB\supK - \binLB\supK), & \txtif i \in \loc\supK.
		} \label{equ:bin_K_offer_rates}
\end{align}
Applying Lemma~\ref{lem:single_driver_best_response_stationary} in the same way as we did in Step 2.2 above, we can show that for a driver in a stationary setting, where the driver always receives trips to all locations at rates $\{\eta_i \supK\}_{i \in \loc}$, the highest expected payoff a driver may get is $\min_{i \in \loc\supK} \{ \gb_i \} = \gb_\maxJ$.

We prove that under $\sigmast$, $\pist(\queue) = \gb_\maxJ$ holds for all $\queue \in [\binLB\supK, \binUB\supK]$. 
Denote $\mutilde_\maxJ \triangleq \min\{\mu_\maxJ, ~\lambda - \sum_{j < \maxJ} \mu_j\}$.
Under $\sigmast$, a driver in the last bin accepts trip dispatches to location $\maxJ$ with probability $\mutilde_\maxJ / \mu_\maxJ$ (see \eqref{equ:rand_fifo_case_2_alphast}). 
As a result, for a driver at position $\queue \in [\binLB\supK, \binUB\supK]$ in the queue, under $\sigmast$, the total rate at which the driver receives and accepts trip dispatches is 
\begin{align*}
    \etatilde \supK \triangleq \sum_{i \in \loc \supK, ~ i < \maxJ} \eta \supK_i + \mutilde_\maxJ /(\binUB\supK - \binLB\supK).
\end{align*}

Consider now a driver who is at position $\binUB\supK$ at time $t = 0$, and denote the driver's position as a function of time $t$ as $\goft(t)$. 
The same argument as in the proof of Step 2.1 implies that for all $t$ such that $\goft(t) \in [\binLB\supK, \binUB\supK]$
the derivative of the continuation payoff $\pist(\goft(t))$ %
with respect to $t$ is of the form:
\begin{align*}
    \frac{\dd \pist(\goft(t))}{ \dd t} 
    = & \etatilde \supK \left( \pist(\goft(t))  - \left( \left(\sum_{i \in \loc \supK, i< \maxJ} \gb_i \mu_i + \gb_\maxJ \mutilde_\maxJ \right) \bigg/ \left(\sum_{i \in \loc \supK, i< \maxJ} \mu_i + \mutilde_\maxJ \right)   - \cost / \etatilde \supk \right) \right).
\end{align*}
It is straightforward to verify that 
\begin{align}
     \left( \left(\sum_{i \in \loc \supK, i< \maxJ} \gb_i \mu_i + \gb_\maxJ \mutilde_\maxJ \right) \bigg/ \left(\sum_{i \in \loc \supK, i< \maxJ} \mu_i + \mutilde_\maxJ \right)   - \cost / \etatilde \supK \right) = \gb_\maxJ,
\end{align}
thus solving $\frac{\dd \pist(\goft(t))}{ \dd \queue} = \etatilde\supK \left( \pist(\goft(t)) - \gb_\maxJ \right)$ with boundary condition $ \pist(\goft(t)) = \gb_\maxJ$, we know that $\pist(\goft(t)) = \gb_\maxJ$ holds for all $t$ such that $\goft(t) \in [\binLB\supK, \binUB\supK]$.  %

\smallskip

Regarding the tail of the queue: in the under-supplied scenario where $\lambda \leq \sum_{i \in \loc} \mu_i$, the equilibrium queue length is $\equQL = \binUB\supK = \numSkip_\maxJ$, thus there is no more driver in the queue beyond $\binUB\supK$.
In the over-supplied scenario with $\lambda > \sum_{i \in \loc} \mu_i$, $\maxJ = \ell$ and $\binUB\supK = \numSkip_\ell$, and we have:
\begin{align}
    \pist(\queue) = \gb_\ell - \cost (\queue - \numSkip_\ell) \bigg/ \sum_{i\in \loc} \mu_i, ~\forall \queue \in [\numSkip_\ell, \maxQL].  
\end{align}

\medskip

Assume that the queue length is $\equQL$ and the rest of the drivers adopt strategy $\sigmast$.
To prove that $\sigmast$ is a best-response for drivers in $[\binUB\supKmo, \binUB\supK]$ i.e. $\pist(\queue) \geq \pi(\queue,\equQL,\sigma, \sigmast)$ for all $\queue \in [\binUB\supKmo, \binUB\supK]$ and any feasible strategies $\sigma$, we use arguments very similar to those in Steps 2.2 and 2.3, and therefore do not repeat the details here. Intuitively, if $\sigmast$ is not a best response, we 
are able to construct a strategy under which a driver gets a payoff strictly higher than  $\min_{i \in \loc\supK} \{ \gb_i \} = \gb_\maxJ$ in the stationary setting where the river always receives trips to all locations at rates $\{\eta_i \supK\}_{i \in \loc}$. This is not possible, as we have discussed above.

This establishes that the highest continuation payoff a driver at $\queue = \binUB\supK$ may get under any strategy is $\pist(\binUB\supK)$. To show that $\sigmast$ is also a best response for any driver at $\queue \in [\binUB\supK, \maxQL]$ in the setting with $\lambda > \sum_{i \in \loc} \mu_i$, the same arguments %
used in the proof of Lemma~\ref{lem:strict_fifo_spe_formal} applies, thus we again refer the readers to Appendix~\ref{appx:proof_strict_FIFO}.

\medskip

This completes the proof for Case 2, $\numBins > 1$, and establishes that when the length of the queue is $\equQL$, strategy $\sigmast$ forms a Nash equilibrium in steady state among the drivers.
As we have discussed earlier, this equilibrium outcome is optimal for trip throughput and the platform's net revenue, since it has the same queue length and completes the same set of trips as the equilibrium outcome under direct FIFO (which is optimal - see Theorem~\ref{thm:fifo_skip_second_best}).  
Combining the three steps of Case 2, we also see that the continuation payoff $\pist(\queue)$ %
is non-negative and monotonically non-increasing in $\queue$. %
As a result, the equilibrium outcome of Case 2 is also individually rational and envy-free.

\medskip

This completes the proof of this theorem. 
\end{proof}

\section{Equilibrium Outcome Under Various Mechanisms} \label{appx:var_of_earnings_derivation}

In this section, we derive the steady state equilibrium outcome under various benchmarks and mechanisms that we discussed in this paper. For each mechanism, (the first best, strict FIFO, direct FIFO, random dispatching, and randomized FIFO), 
we compute the equilibrium trip throughput, net revenue, average driver payoff, length of the queue, the minimum and maximum waiting times in the queue, and the variance in drivers' total payoffs.

\smallskip

Recall that $\maxJ$ as defined in \eqref{equ:maxJ} is the lowest earning trip that is (partially) completed under the first best outcome, and that $\mutilde_\maxJ \triangleq \min\{\mu_\maxJ, ~\lambda - \sum_{i = 1}^{\maxJ - 1} \mu_j\}$ denotes the amount of type $\maxJ$ jobs fulfilled per unit of time in steady state. 
Moreover, for any $i \in \loc$, $\waitGap_{i, i+1}$ denotes the amount of time a driver is willing to wait for a trip to location~$i$, assuming that the driver has the option to immediately accept a trip to location $i+1$: $\waitGap_{i, i+1} = (\gb_i - \gb_{i+1})/\cost$. 

\subsection{The First Best}

We first derive the steady state first best outcome as discussed in Section~\ref{sec:preliminaries},
which dispatches available drivers (upon arrival) to destinations in decreasing order of $\gb_i$, until either all drivers are dispatched, or all riders are picked up. 

\begin{enumerate}[$\bullet$]
	\item Trip throughput: $\tp\fb = \min\{ \lambda, ~ \sum_{i \in \loc} \mu_i\}$. 
	\item Net revenue: $\rev\fb = \sum_{i =1}^{ \maxJ - 1} \gb_i \mu_i + \gb_\maxJ \mutilde_\maxJ $.  %
	\item The average payoff of each driver who arrived at the queue: $\util \fb = \rev\fb / \lambda$, since no driver incurs any waiting cost, and $\rev\fb$ is equal to the total payoff achieved by all drivers per unit of time. 
	\item The first best outcome maintains an empty driver queue, therefore the queue length, the minimum and maximum waiting time of any driver, and the average driver waiting time are all zero. 
	\item Since drivers are either dispatched a random trip or asked to leave the queue without waiting any time in the queue, the variance in drivers' payoffs can be computed as follows:
	\begin{enumerate}[-]
		\item When the platform is under-supplied, i.e. $\lambda \leq \sum_{i \in \loc}\mu_i$, every driver gets dispatched a trip to some location $i \leq \maxJ$. The variance in drivers' payoffs is therefore the variance of the net earnings of the completed trips: 
		\begin{align*}
			\var{\Util\fb} = \frac{1}{\lambda} \left( \sum_{i = 1}^{\maxJ - 1} \mu_i (\gb_i - \util  \fb)^2 + \mutilde_\maxJ (\gb_\maxJ - \util  \fb)^2 \right).
		\end{align*}
		\item When the platform is over-supplied, there are $\mu_j$ drivers per unit of time each getting a payoff of $\gb_j$, and the rest of the drivers all get zero since they leave the airport a rider trip. Therefore, 
		\begin{align*}
			\var{\Util\fb} = \frac{1}{\lambda} \left(\sum_{i \in \loc} \mu_i (\gb_i - \util  \fb)^2  + (\lambda - \sum_{i \in \loc} \mu_i) (0 - \util\fb)^2 \right).
		\end{align*}
	\end{enumerate}
\end{enumerate}

\subsection{Equilibrium Outcome under Direct FIFO}

As we have proved in the paper, the direct FIFO mechanism has the same trip throughput as the first best throughput, and zero variance in drivers' earnings. Moreover:
\begin{enumerate}[$\bullet$]
	\item As we have shown in Theorem~\ref{thm:fifo_skip_second_best},  the steady state equilibrium queue length is $\equQL\direct = \maxQL$ %
	if $\lambda > \sum_{i \in \loc} \mu_i$, and $\equQL\direct = \numSkip_\maxJ$, otherwise. %
	\item Net revenue: 
	\begin{enumerate}[-]
		\item When $\lambda > \sum_{i \in \loc} \mu_i$, all trips are completed, thus $\rev\direct = \sum_{i \in \loc}\mu_i \gb_i - \equQL\direct \mechCost$.
		\item When $\lambda \leq \sum_{i \in \loc} \mu_i$, we have $\rev\direct = \sum_{i =1}^{\maxJ - 1} \mu_i \gb_i + \mutilde_\maxJ \gb_\maxJ - \equQL\direct \mechCost$.    
	\end{enumerate}		
	\item The average driver payoff in equilibrium is $\equtil \direct = 0$ if $\lambda > \sum_{i \in \loc} \mu_i$, and $\equtil \direct  = \gb_\maxJ$ otherwise.
	\item Regarding the maximum, minimum, and average waiting times in queue:
	\begin{enumerate}[-]
		\item When $\lambda \leq \sum_{i \in \loc} \mu_i$, the minimum waiting time for a trip (in this case, a trip to location $\maxJ$) would be zero. The maximum waiting time (which would be for a trip to location $1$) is $\sum_{i =1}^{\maxJ - 1} \waitGap_{i, i+1} = (\gb_1 - \gb-\maxJ) / \cost$. The average waiting time for a driver who arrived at the queue is $\equQL\direct  / \lambda$, and the average waiting time for a driver who joined the queue is the same. 
		\item When the queue is $\lambda > \sum_{i \in \loc} \mu_i$, %
		the minimum amount of time the driver needs to wait in the queue for a trip (which would be for a trip to location $\numLoc$) is $\gb_\ell / \cost$.  The maximum waiting time is $\gb_1 / \cost$. %
		The average waiting time for a driver who arrived at the airport is $\equQL\direct  / \lambda$, and the average waiting time for a driver who joined the queue $\equQL\direct / \sum_{i \in \loc} \mu_i$. 
	\end{enumerate}
\end{enumerate}

\subsection{Equilibrium Outcome under Strict FIFO Dispatching}

Under strict FIFO, when $\patience \geq \numSkip_\maxJ$, all trips that are completed under direct FIFO will be able to reach a driver who is willing to accept them under strict FIFO. As a result, the equilibrium outcome will be identical to that under direct FIFO. 

When $\patience < \numSkip_\maxJ$, some trips that are completed under direct FIFO will not be completed under strict FIFO, thus there exists excess drivers who need to leave the queue without a rider trip. Let $\maxJ(\patience)$ be the lowest earning trip with $\numSkip_i \leq \patience$, the equilibrium outcome is as follows:
\begin{enumerate}[$\bullet$]
	\item Trip throughput: $\tp\strict = \sum_{i = 1}^{\maxJ(\patience)} \mu_i$. 
	\item The average payoff of drivers is thereby also $\equtil\strict = 0$, since drivers will join the queue until when they are indifferent between joining the queue and leaving without a rider.
	\item Drivers will be willing to wait $\gb_{\maxJ(\patience)}/\cost$ units of time for a trip to location $\maxJ(\patience)$, thus the total length of the queue would be $\equQL\strict = \numSkip_{\maxJ(\patience)} + \tp\strict \gb_{\maxJ(\patience)}/\cost $, which is equal to $\sum_{i = 1}^{\maxJ(\patience)} \mu_i \gb_i / \cost$.
	\item Net revenue: $\rev\strict=  \sum_{i = 1}^{\maxJ(\patience)} \mu_i \gb_i  - \equQL\strict \mechCost$. %
	\item $\gb_{\maxJ(\patience)}/\cost$ is the minimum amount of time a driver has to wait for any trip, and  the maximum waiting time (which would be for a trip to location~1) is $\gb_1 / \cost$. %
	\item On average, the total amount of time spent by all drivers on waiting is $\equQL$ units of time, per unit of time. Therefore, the average waiting time for a driver who joined the virtual queue is $\equQL/\tp\strict$, and the average waiting time for a driver who arrived at the origin is $\equQL / \lambda$. 
	\item Every driver gets a zero net payoff, thus the variance in drivers' earnings is also zero. 
\end{enumerate}

\subsection{Equilibrium Outcome under Random Dispatching} \label{appx:eq_under_pure_random}

As we have proved in Proposition~\ref{prop:pure_random}, random dispatching achieves the same equilibrium trip throughput, net revenue, and queue length as those under the direct FIFO mechanism. As a result, drivers also have the same average payoff and average waiting time. Given the fact that dispatching is random, theoretically drivers might not have to wait any time for a trip dispatch, and there is also no upper bound on a driver's waiting time in the queue.

What is left to compute is the variance in drivers' total payoff. We discuss the over-supplied and the under-supplied settings separately.

\paragraph{Over-supplied.} With $\lambda > \sum_{i \in \loc} \mu_i$, %
the average net earnings from a completed trip is  $\meangb \triangleq \sum_{i \in \loc} \gb_i \mu_i  \bigg/ \sum_{i \in \loc} \mu_i.$ 
The average waiting time of a driver who joined the queue is $\meangb /\cost$, since in equilibrium drivers are indifferent towards whether to join the queue.

Note that (i) a driver's waiting time in the queue is independent to the driver's net earnings from the trip she accepts, and (ii) whether a driver gets dispatched in memoryless, thus a driver's waiting time is exponentially distributed, with mean $\meangb /\cost$. The variance in drivers' waiting times is therefore $(\meangb /\cost)^2$, thus the variance in drivers waiting costs is $\meangb^2$. 

In steady state, $ \sum_{i \in \loc} \mu_i $ drivers join the queue per unit of time. The rest of the drivers do not join the queue thus get $0$, which is equal to the average payoff of all drivers. 
The total variance in the payoff of a driver who arrived at the airport is therefore:
\begin{align*}
	\left( \meangb^2 + \sum_{i \in \loc} (\gb_i - \meangb)^2 \mu_i  \bigg/ \sum_{i \in \loc} \mu_i\right)   \sum_{i \in \loc} \mu_i / \lambda. 
\end{align*}

\paragraph{Under-supplied.} 
Consider now the case when the queue is not over-supplied, and the lowest-earning trips that's completed is $\maxJ$. %
In equilibrium, $\mutilde_\maxJ$ units of trips to location $\maxJ$ are completed in each unit of time. 
The average net earnings of the trip completed by each driver
\begin{align*}
	\meangb = \left(\sum_{i = 1}^\maxJ \gb_i \mu_i  +  \mutilde_\maxJ \gb_\maxJ \right) \bigg/  \lambda,
\end{align*}
and the variance in drivers' net earnings from trips is
\begin{align*}
	 \left( \sum_{i = 1}^{\maxJ-1} (\gb_i - \meangb)^2 \mu_i + (\gb_\maxJ - \gb_\maxJ)^2 \mutilde_\maxJ  \right) \bigg/ \lambda. 
\end{align*}

The average waiting time for a driver is $\equQL / \tp = (\meangb - \gb_\maxJ)/\cost$, and the distribution of waiting times is exponential. Therefore, the variance in drivers waiting costs is $(\cost \equQL / \tp)^2 = (\meangb - \gb_\maxJ)^2$,  and the total variance in drivers' payoff is:
\begin{align*}
	(\meangb - \gb_\maxJ)^2 +  \sum_{i = 1}^{\maxJ-1} (\gb_i - \meangb)^2 \mu_i  \bigg/ \lambda. 
\end{align*}

\subsection{Equilibrium Outcome under Randomized FIFO} \label{appx:eq_under_randomized_FIFO}

As we have proved in Theorem~\ref{thm:randFIFO_second_best}, the randomized FIFO mechanisms achieve the same equilibrium trip throughput, net revenue, and queue length as those under the direct FIFO mechanism. As a result, drivers also have the same average payoff and average waiting time in the queue. With randomization in dispatching, it is generally possible for drivers to wait zero or infinite units of time for a dispatch, although there exist special cases where the minimum and maximum waiting times in the queue are non-zero or finite.

We now derive the minimum and maximum waiting times, and the variance in drivers' payoffs. We discuss the same set of cases as analyzed in the proof of Theorem~\ref{thm:randFIFO_second_best} in Appendix~\ref{appx:proof_thm_randFIFO_second_best}. 

\paragraph{Case 1.1: $\numBins = \maxJ = 1$.} As we have proved in Appendix~\ref{appx:proof_thm_randFIFO_second_best}, the outcome under randomized FIFO in this case is identical to that under the direct FIFO mechanism. %

\paragraph{Case 1.2: $\maxJ > 1$, $\numBins = \patience = 1$.} As we've discussed in Appendix~\ref{appx:proof_thm_randFIFO_second_best}, the outcome under randomized FIFO in this case will be identical to that under random dispatching, when the queue is not over supplied, i.e. when $\lambda \leq \sum_{i \in \loc}\mu_i$. %

What is left to discuss in the over-supplied case with $\lambda > \sum_{i \in \loc}\mu_i$. In this case, every trip is randomly dispatched to drivers at positions $[0, \numSkip_\ell]$ in the queue, and no driver declines any dispatches in equilibrium. With this randomization, the maximum waiting time for a driver in the queue can be infinite, but the minimum time a driver has to wait for a trip would be $(\maxQL - \numSkip_\ell)/\sum_{i \in \loc}\mu_i = \gb_\ell / \cost $, since a driver does not receive any dispatch until she has moved from the tail of the queue $\equQL = \maxQL$ to position $\numSkip_\ell$ in the queue. 

For a driver who joined the queue, the variance in her earnings from the trip she completes is $\sum_{i \in \loc} (\gb_i - \meangb)^2 \mu_i  \bigg/ \sum_{i \in \loc} \mu_i$, where $\meangb = \sum_{i \in \loc} \gb_i \mu_i  \bigg/ \sum_{i \in \loc} \mu_i$ is the average net earnings from a rider trip. Once a driver had reached $\numSkip_\ell$ in the queue, the additional time she has to wait for a trip is exponentially distributed with mean $(\meangb - \gb_\ell) / \cost$. As a result, for a driver who joined the queue, the variance in her waiting cost is $(\meangb - \gb_\ell)^2$, and the overall variance of all drivers who arrived at the queue is:
\begin{align*}
	\left( (\meangb - \gb_\ell)^2 + \sum_{i \in \loc} (\gb_i - \meangb)^2 \mu_i  \bigg/ \sum_{i \in \loc} \mu_i\right)   \sum_{i \in \loc} \mu_i / \lambda. 
\end{align*}

\newcommand{\0}{^{(0)}}

\paragraph{Case 2: $\numBins > 1$.} We now consider the case where there are multiple bins under randomized FIFO. 

No driver receives any dispatch until the driver reaches $\binUB\supK$. 
Since $\binUB\supK = \numSkip_\maxJ$, the minimum amount of time any driver has to wait for a trip is equal to $(\equQL\rand - \numSkip_\maxJ)/\tp\rand$, which is equal to $0$ if the queue is not over-supplied, and is equal to $(\maxQL - \numSkip_\ell)/\sum_{i \in \loc} \mu_i$ otherwise. What is left to compute is drivers' maximum waiting time in queue, and the variance in drivers' total payoffs. 

\medskip

Recall that $\eqUtil = \Util(\equQL, \equQL, \sigmast, \sigmast)$ is the random variable representing the payoff of all drivers who arrived at the queue, and that $\equtil \triangleq \E{\eqUtil} = \pist(\equQL)$. From Theorem~\ref{thm:randFIFO_second_best}, the average payoff of all drivers who arrived at the queue is the same as that under direct FIFO, i.e. $\equtil = \gb_\maxJ$ when the queue is not oversupplied, and $\equtil = 0$, otherwise. 

Let $\Util \supk$ represent the (equilibrium, steady state) total payoff of a driver who is dispatched from the $k\th$ bin, and let $\Util\0 = 0$ denote the payoff of drivers who did not join the queue (if any). We know that
\begin{enumerate}[$\bullet$]
	\item $\eqUtil$ takes value $\Util\0$ with probably $\psi\0 \triangleq \max\{\lambda - \sum_{i \in \loc}\mu_i, ~0\} / \lambda$.
	\item $\eqUtil$ takes value $\Util\supk$ with probability $\psi\supk \triangleq \tpRand\subkk / \lambda = \sum_{i \in \loc\supk} \mu_i /\lambda$ for each $k = 1, \dots, \numBins - 1$, and
	\item $\eqUtil$ takes value $\Util\supK$ with probability $\psi\supK \triangleq \min\{\lambda -  \tpRand_{1:m-1},~ \tpRand_{m:m} \} / \lambda$.
\end{enumerate}
The overall variance in drivers' total payoff $\var{\eqUtil} $ can be written as:
\begin{align*}
	 &  %
	\E{(\eqUtil - \equtil)^2} 
	= \sum_{k = 0}^\numBins \psi\supk \E{ \left(\Util \supk - \equtil \right)^2 } 
	= \sum_{k = 0}^\numBins \psi\supk \left( \var{\Util \supk} +  \left(\E{\Util \supk} - \equtil \right)^2\right).
\end{align*}
$\E{\Util\0} = \var{\Util\0} = 0$. We now compute  $\E{\Util\supk}$ and $\var{\Util\supk}$ for each $k\geq 1$.

\medskip

\noindent{}\emph{The last bin: $k = \numBins$.} 
When $|\loc\supK| = 1$, $\loc\supK = \{ \maxJ \}$ and $\binLB\supK = \binUB\supK = \numSkip_\maxJ$.  In this case, there is no variance in the payoffs of all drivers who are dispatched from the last bin: $\var{\Util\supK} = 0$. Moreover, drivers dispatched from the last bin gets the average payoff of all drivers: $\E{\Util\supK} = \equtil$.

Now consider the setting where $|\loc\supK| > 1$. Recall that $\mutilde_\maxJ \triangleq \min\{\mu_\maxJ, ~\lambda - \sum_{j < \maxJ} \mu_j\}$ is the unit of location $\maxJ$ trip completed per unit of time in steady state. The rate at which drivers are dispatched from the last bin is:
\begin{align*}
	\tilde{\tpRand}_{m:m} \triangleq \sum_{i \in \loc\supK, ~i < \maxJ}  \mu_i + \mutilde_\maxJ
\end{align*}
and the average net earnings from  a trip accepted by a driver dispatched from the last bin is:
\begin{align*}
	\meangb\supK \triangleq \left( \sum_{i \in \loc\supK, ~i < \maxJ} \gb_i \mu_i + \gb_\maxJ \mutilde_\maxJ \right) \bigg/ \tilde{\tpRand}_{m:m}. 
\end{align*}

$\Util\supK$ is equal to the earning from the trip the driver (who is dispatched from the last bin) completes, minus the total waiting costs the driver incurred. The expected net earnings from trip is $\meangb\supK$.
For the expected waiting cost, %
note that once a driver reached $\binUB\supK$, the driver's additional waiting time for a dispatch should be exponentially distributed, truncated at the time the driver reaches $\binLB\supK$. The parameter of this exponential distribution is $\etatilde \supK \triangleq \tilde{\tpRand}_{m:m} / (\binUB\supK - \binLB\supK)$, and by the time the driver reaches $\binLB\supK$, $\tilde{\tpRand}_{m:m}$ out of the $\tp\rand = \min\{\lambda, \sum_{i \in \loc} \mu_i \}$ drivers who reached position $\binUB\supK$  in the queue are dispatched. Denote
\begin{align*}
	\zeta\supK \triangleq \tilde{\tpRand}_{m:m} / \tp\rand
\end{align*}
as the fraction of drivers who are dispatched in the last bin (out of all of the drivers who joined the queue), and 
denote $\Delta\supK$ as the time it takes for a driver to move from $\binUB\supK$ to $\binLB\supK$ in the queue if the driver is not dispatched before reaching $\binLB\supK$, we know:
\begin{align}
	1 - e^{- \Delta\supK \etatilde \supK  } = \zeta\supK  \Rightarrow \Delta\supK = -  \log(1 -  \zeta\supK ) / \etatilde \supK . \label{equ:time_to_move_from_ub_to_lb}
\end{align}
Denote $\nu \supK \triangleq \Delta\supK + (\equQL - \binUB\supK) / \tp\rand$, we know that $\nu \supK$ is the time it takes for a driver to move from the tail of the queue to the lower bound of the last bin $\binLB\supK$, if the driver is not dispatched before reaching $\binLB\supK$.

For a trip for a driver who is dispatched from the last bin, the average waiting time the driver spends in the last bin is therefore:
\begin{align}
	\int_{0}^{\Delta\supK}  t \cdot \etatilde \supK  e^{- \etatilde \supK t } \dd t = \frac{1}{\etatilde \supK   } ( 1 + (1/ \zeta\supK -  1)  \log( (1 -  \zeta\supK )). \label{equ:expected_waiting_time_in_bin}
\end{align}
Let $\kappa\supK$ be the random variable representing the total waiting time of a driver who is dispatched from the last bin, we know that
\begin{align*}
	\E{\kappa\supK} = (\equQL - \binUB\supK) / \tp\rand + \frac{1}{\etatilde \supK   } ( 1 + (1/ \zeta\supK -  1)  \log( (1 -  \zeta\supK )),
\end{align*}
and with this we can compute $\E{\Util\supK} = \meangb\supK - \cost \E{\kappa\supK}$.  

\smallskip

What is left to compute is $\var{\Util\supK}$.
For a driver who is dispatched from the last bin, the earning from the trip the driver receives is independent to the time at which the driver receives the trip.  As a result, the  variance $\var{\Util\supK}$ should be the sum of the variance in trip earnings and the variance in the waiting costs. The former is equal to:
\begin{align*}
	\left(\sum_{i \in \loc\supK, ~i < \maxJ} (\gb_i - \meangb\supK)^2 \mu_i + (\gb_\maxJ - \meangb\supK) \mutilde_\maxJ \right) \bigg/ \tilde{\tpRand}_{m:m},
\end{align*} 
and the latter is of the form:
\begin{align*}
	& \cost^2 \int_{0}^{\Delta\supK}  \left(t - \frac{1}{\etatilde \supK   } ( 1 + (1/ \zeta\supK -  1)  \log( (1 -  \zeta\supK )) \right)^2  \etatilde \supK  e^{- \etatilde \supK t } \dd t \\
	 = & \left( \frac{\cost} {\etatilde \supK \zeta\supK}\right)^2 
	 \left(( \zeta\supK)^2 + (-1 +  \zeta\supK) \left(\log(1 -  \zeta\supK)\right)^2 \right).
\end{align*}

\medskip

\noindent{}\emph{A middle bin.} Now consider each $k = \numBins -1, \numBins - 2, \dots, 2$.  %
When $|\loc\supk| = 1$, there is no variance in the payoffs of drivers who are dispatched from the $k\th$ bin: $\var{\Util\supk} = 0$. 
It takes the driver a total of $\nu\supkpo + (\binLB\supkpo - \binUB\supk)/\tpRand\subok$ units of time to move from the tail of the queue to $\binUB\supk$, and once the driver gets to $\binUB\supk$, the driver receives $\meangb\subkk$, which is equal to the net earnings from the only trip in $\loc\supk$. Therefore, $\E{\Util\supk} = \meangb\subkk - \cost \left(\nu \supkpo + (\binLB\supkpo - \binUB\supk)/\tpRand\subok \right)$.

For the case where $|\loc\supk| > 1$, recall that $\eta\supk \triangleq \tpRand\subkk / (\binUB\supk - \binLB\supk)$, and denote:
\begin{align*}
	\zeta \supk \triangleq \tpRand \subkk / \tpRand \subok.
\end{align*} 
With the same argument as those for the last bin,
the time it takes a driver to move from $\binUB\supk$ to $\binLB\supk$ (if the driver is not dispatched before reaching $\binLB\supk$) is
\begin{align*}
	\Delta \supk \triangleq  -  \log(1 -  \zeta\supk ) / \eta \supk.
\end{align*}
This implies that the total waiting time for a driver to reach $\binLB\supk$ (if the driver is not dispatched before then) is $\nu\supk \triangleq \nu \supkpo + (\binLB\supkpo - \binUB\supk)/\tpRand\subok + \Delta \supk$. 

The expected time a driver waits in the $k\th$ bin, if the driver is dispatched from the $k\th$ bin, is of the form:
\begin{align*}
	\int_{0}^{\Delta\supk}  t \cdot \eta \supk  e^{- \eta \supk t } \dd t = \frac{1}{\eta \supk } ( 1 + (1/ \zeta\supk -  1)  \log( (1 -  \zeta\supk)).	
\end{align*}
thus the expected payoff of a driver who is dispatched in the $k\th$ bin is:
\begin{align*}
	\E{\Util\supk} = \meangb\subkk - \cost \left(  \nu \supkpo + (\binLB\supkpo - \binUB\supk)/\tpRand\subok +  \frac{1}{\eta \supk } ( 1 + (1/ \zeta\supk -  1)  \log( (1 -  \zeta\supk)) \right).
\end{align*}

The variance $\var{\Util\supk}$ is similarly consisted of two parts. The variance from the net earnings from a trip a driver accepts in the $k\th$ bin is
\begin{align*}
	\sum_{i \in \loc\supk} (\gb_i - \meangb\supk)^2 \mu_i / \tpRand\subkk, 
\end{align*}
and the variance of drivers' waiting costs is:
\begin{align*}
	& \cost^2 \int_{0}^{\Delta\supk}  \left(t - \frac{1}{\eta \supk   } ( 1 + (1/ \zeta\supk -  1)  \log( (1 -  \zeta\supk )) \right)^2  \eta \supk  e^{- \eta \supk t } \dd t \\
	 = & \left( \frac{\cost} {\eta \supk \zeta\supk}\right)^2 
	 \left(( \zeta\supk)^2 + (-1 +  \zeta\supk) \left(\log(1 -  \zeta\supk)\right)^2 \right).
\end{align*}

\medskip

\noindent{}\emph{The first bin.}  When $|\loc\1| = 1$,  $\var{\Util\1} = 0$. The total waiting time for a driver to reach $\binUB\1 = 0$ is $\nu\2 + \binLB\2 / \mu_1$, thus the expected payoff of a driver dispatched from the first bin is $\E{\Util\1} = \gb_1 - \cost \left( \nu\2 + \binLB\2 / \mu_1 \right)$. Moreover, in this case, the maximum waiting time for any driver who have joined the queue is $\nu\2 + \binLB\2 / \mu_1$.

When $|\loc\1| > 1$, the waiting time a driver may spend waiting in the queue is unbounded. Once a driver reached $\binUB\1$, the driver's waiting time for a trip is exponentially distributed with parameter $\eta\1 \triangleq \tpRand_{1:1} / (\binUB\1 - \binLB\1)$. The expected waiting time in the first bin is $1/\eta\1$, thus the expected payoff of a driver who is dispatched from the first bin is
\begin{align*}
	\E{\Util\1} = \meangb_{1:1} - \cost\left(1 / \eta\1  + \nu\2 + (\binLB\2 - \binUB\1) / \tpRand_{1:1} \right),
\end{align*}
and $\var{\Util\1}$,  the variance of the earnings of a driver who is dispatched from the first bin is:
\begin{align*}
	\var{\Util\1} = \sum_{i \in \loc\1} (\gb_i - \meangb_{1:1})^2 \mu_i / \tpRand_{1:1} + (\cost / \eta\1)^2. 
\end{align*}

\section{Additional Discussion and Examples} \label{appx:additional_discussion}

\subsection{Net Earnings from Prices and Distances} \label{appx:net_earnings}

For each location $i \in \loc$, the \emph{net earnings} $\gb_i$ from a trip to location $i$ represents the total payoff of a driver who completed a trip to location $i$, minus the total payoff to a driver who left the queue without a rider (in this way, the net earnings of a driver who left the queue without a rider is normalized to be zero). The net earnings incorporate payments from the immediate trip, as well as
drivers' continuation earnings after arriving at different destinations (which are affected by market conditions at the destinations). %

\smallskip

Assuming that drivers get the same continuation earnings from every destination onward, we now illustrate %
how net earnings of trips can be derived from the prices
and distances of trips to different destinations.
For each destination $i \in \loc$, let $\dist_i > 0$ denote the amount of time (e.g. minutes) it takes for a driver to complete a trip to destination $i$, which includes time it takes for a driver to pick up the rider. $\rate_i > 0$ denotes the effective earnings rate from a trip to location $i \in \loc$, meaning that the total payment to a driver for a trip to location $i$ is $ \rate_i\dist_i$. The earnings rates are induced by the time and distance rates the platform pays the drivers, and may vary across destinations due to differences in trip lengths and traffic conditions, etc.
For drivers who decides to relocate without a rider and drive elsewhere, $\minDist > 0$ %
is the minimum relocation distance, i.e. the amount of time a driver needs to spend driving from the airport in order to start making an average earnings rate of $\cost$.

A driver who accepts a trip from the virtual queue to some location $i \in \loc$ will make $\rate_i$ per unit of time for $\dist_i$ periods, followed by making $\cost$ per period after arriving at location $i$.  
A driver who relocates back to the city without a rider makes $0$ for the first $\minDist$ periods, and then starts to earn $\cost$ per period. See Figure~\ref{fig:trip_deadhead_timeline}. The additional earnings from a trip to location $i$, relative to that from relocating without a rider (and then driving in the city), is the \emph{net earnings} from this trip.  %
For each location $i \in \loc$, the %
net earnings $\gb_i$ is of the form:
\begin{align}
	\gb_i = \dist_i \rate_i - \left( \minDist \cdot 0 +  (\dist_i - \minDist) \cost \right) %
	= \dist_i (\rate_i - \cost) + \minDist \cost. \label{equ:net_gb}
\end{align}

\renewcommand{\lamdmarkHeight}{0.4}
\renewcommand{\patienceLocation}{7}
\newcommand{\verticalGap}{3}

\begin{figure}[t!]
\centering
\begin{tikzpicture}[scale = 0.82][font =\small]

\draw[->, line width=0.25mm] (-0,0) -- (12,0) node[anchor=north] {time};

\draw(-2.8, -0.3) node[anchor=south] {A trip to location $i$};

\draw[-, line width=0.5mm] 	(0,0.1) -- (0,-0.1);
\draw[-, line width=0.5mm] 	(6,0.1) -- (6,-0.1);

\draw [decorate, decoration={brace, mirror, amplitude=8 pt}, xshift=0.4pt,yshift=-0.4pt] (0,0) -- (6,0); %
\node[text width=2.3cm] at (3.4, -0.8) {$\dist_i$ periods};

\draw(3, \lamdmarkHeight) node[anchor=south] {$\rate_i$ per period};
\draw(8.8, \lamdmarkHeight) node[anchor=south] {$\cost$ per period};

\draw[->, line width=0.25mm] (-0,-\verticalGap) -- (12,-\verticalGap) node[anchor=north] {time};

\draw(-2.8, -\verticalGap-0.3) node[anchor=south] {Relocation w/o a rider};

\draw[-, line width=0.5mm] 	(0,-\verticalGap+0.1) -- (0,-\verticalGap-0.1);
\draw[-, line width=0.5mm] 	(4,-\verticalGap+0.1) -- (4,-\verticalGap-0.1);

\draw [decorate, decoration={brace, mirror, amplitude=8 pt}, xshift=0.4pt,yshift=-0.4pt] (0,-\verticalGap) -- (4,-\verticalGap); %
\node[text width=2.3cm] at (2.4, -\verticalGap-0.8) {$\minDist$ periods};

\draw(2, \lamdmarkHeight-\verticalGap) node[anchor=south] {$0$ per period};
\draw(7.85, \lamdmarkHeight-\verticalGap) node[anchor=south] {$\cost$ per period};

\end{tikzpicture}
\caption{The timeline of a trip to location $i$ (above) and relocation without a rider (below). 
\label{fig:trip_deadhead_timeline}} 
\end{figure}
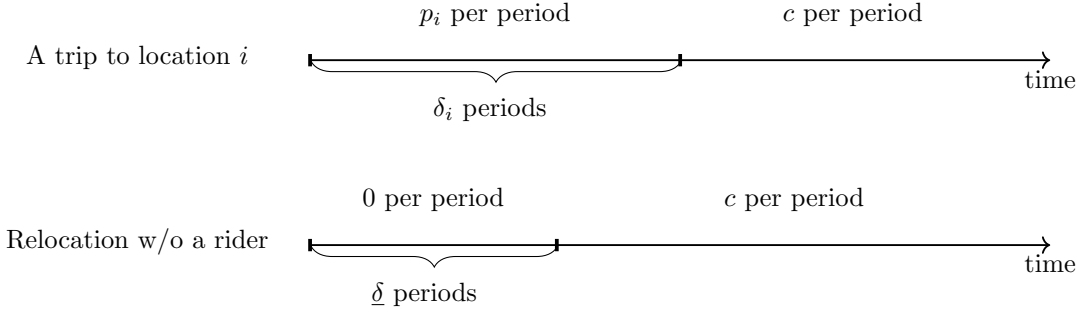

\subsection{Additional Examples}

Example~\ref{exmp:last_bin_with_trip_worse_than_maxJ} shows that a randomized FIFO mechanism may not achieve the second best outcome, if some trips with lower earnings than $\gb_\maxJ$ are included in the ordered partition of destinations.

\begin{example}[Last bin w/ trips to location $i \geq \maxJ$] \label{exmp:last_bin_with_trip_worse_than_maxJ}

Consider an economy with three destinations, where $\mu_1 = 1$, $\gb_1 = 100$,  $\mu_2 = 2$, $\gb_2 = 40$, and $\mu_3 = 5$, $\gb_3 = 10$. The arrival rate of drivers is $\lambda = 2$, and that the opportunity cost for drivers' time is $\cost = 1$. Under the first best outcome, one unit of trips to location 1 and one unit of trips to location 2 are completed per unit of time. With $\maxJ = 2$, the average net payoff of drivers under the second best outcome would be $\gb_\maxJ = \gb_2 = 40$, and the equilibrium, steady state queue length is $\equQL = \numSkip_\maxJ = \mu_1 (\gb_1 - \gb_2)/\cost  = 60$.

Assume that riders have a patience level of $\patience = 2$. The appropriate construction of randomized FIFO corresponds to the ordered partition $\loc\1 = \{1\}$ and $\loc\2 = \{2\}$. Now consider a randomized FIFO mechanism associated with the ordered partition $\loc\1 = \{1\}$ and $\loc\2 = \{2, ~3\}$. Constructing the bins according to \eqref{equ:defn_of_binLB} and \eqref{equ:defn_of_binUB}, we have  $\binLB\1 = \binUB\1 = 0$ and
\begin{align*}
	\binLB\2 = \frac{1}{\cost } \mu_1 (\gb_1 - \gb_3) = 90. 
\end{align*}
Note that $\binLB\2$ is higher than the equilibrium queue length $\equQL$ under the second best outcome. We now show that the randomized FIFO mechanism constructed in this way will not achieve the second best as long as $\mechCost > 0$. For a driver in the first bin, i.e. at the head of the queue, the driver is only willing to accept a trip to location~$1$.  When the queue is shorter than $\binLB\2$, trips to location $2$ or $3$ will not be dispatched again by the randomized FIFO mechanism after being rejected by the driver at the head of the queue, thus all but location~$1$ trips become unfulfilled. 
But when the queue is longer than $\binLB\2$, we will not be achieving our second best outcome either, since the total waiting costs incurred by the drivers will be higher than that under the second best outcome, even when we are completing the same set of trips.  
\qed
\end{example}

\section{Additional Simulations} \label{appx:additional_simulations} 

We include in this section descriptions of the dataset made public by the City of Chicago, additional simulation results for O'Hare that are omitted from Section~\ref{sec:simulations} of the paper, as well as simulation results for the Chicago Midway International airport.

\subsection{Chicago O'Hare International Airport} \label{appx:additional_simulations_ohare}

\subsubsection{Trip Volume by Day and Hour-of-Week}

We first provide the volume of trips originating from Chicago O'Hare, and the average duration and earning rates by destination. Figure~\ref{fig:ohare_daily_trips_from_airport} shows the number of trips that originate from the O'Hare airport on each day, from November 1, 2018 to mid March, 2020. We can see strong weekly patterns, seasonality patterns (e.g. low trip volume during Christmas through New Year), and also the sharp decline in trip volume after the onset of the COVID-19 pandemic.

\newcommand{\timeSeriesHeight}{2}

\begin{figure}[t!]
	\centering
    \includegraphics[height= \timeSeriesHeight in] {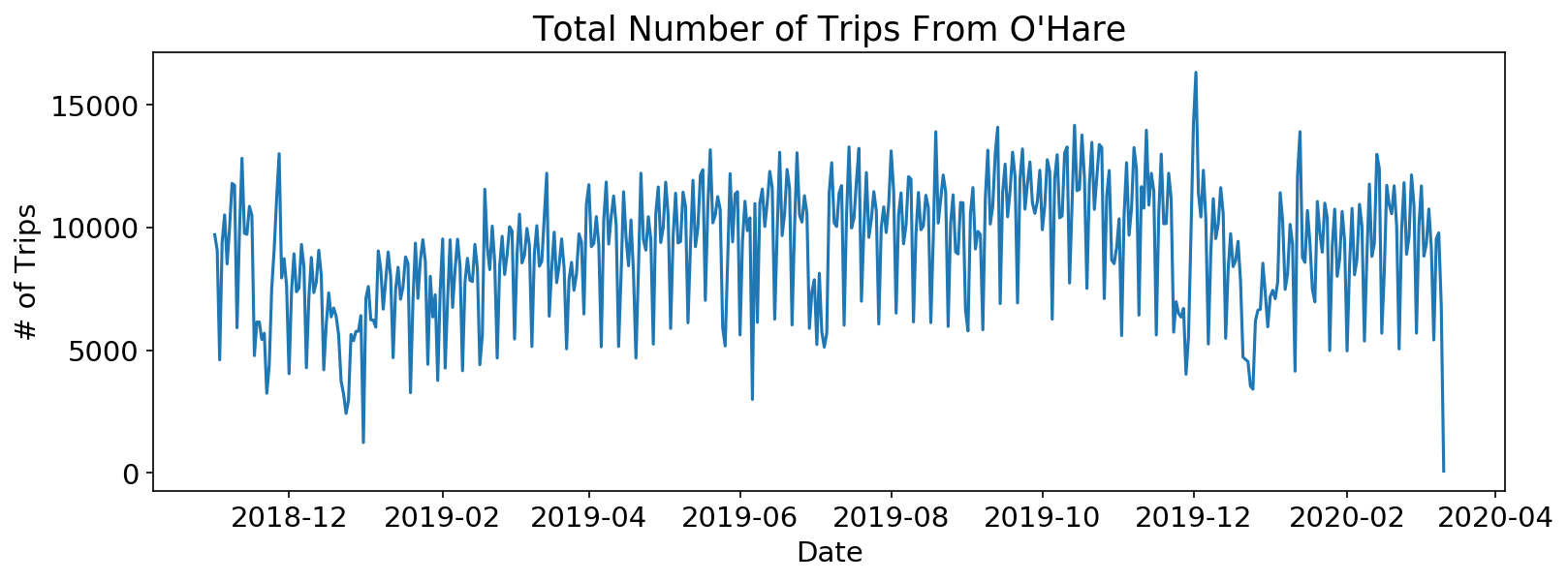}
	\caption{Total number of trips per day from the Chicago O'Hare International Airport.
	\label{fig:ohare_daily_trips_from_airport}}
\end{figure}

The average number of trips originating from O'Hare during each \emph{hour-of-week} is as shown in Figure~\ref{fig:ohare_how_trips_from_airport}. Here, the $0\th$ hour-of-week corresponds to midnight - 1am on Mondays, and the $1^{\mathrm{st}}$ hour-of-week corresponds to 1am - 2am on Mondays, and so on.  We can see that the number of trips originating from the airport peaks during early evenings, averaging around $12$ trips per minute during the weekdays, and reaches a maximum of over $15$ trips per minute on Thursday. Note that these are completed trips, thus the rider request rates are strictly higher.

\begin{figure}[t!]
	\centering
    \includegraphics[height= \timeSeriesHeight in] {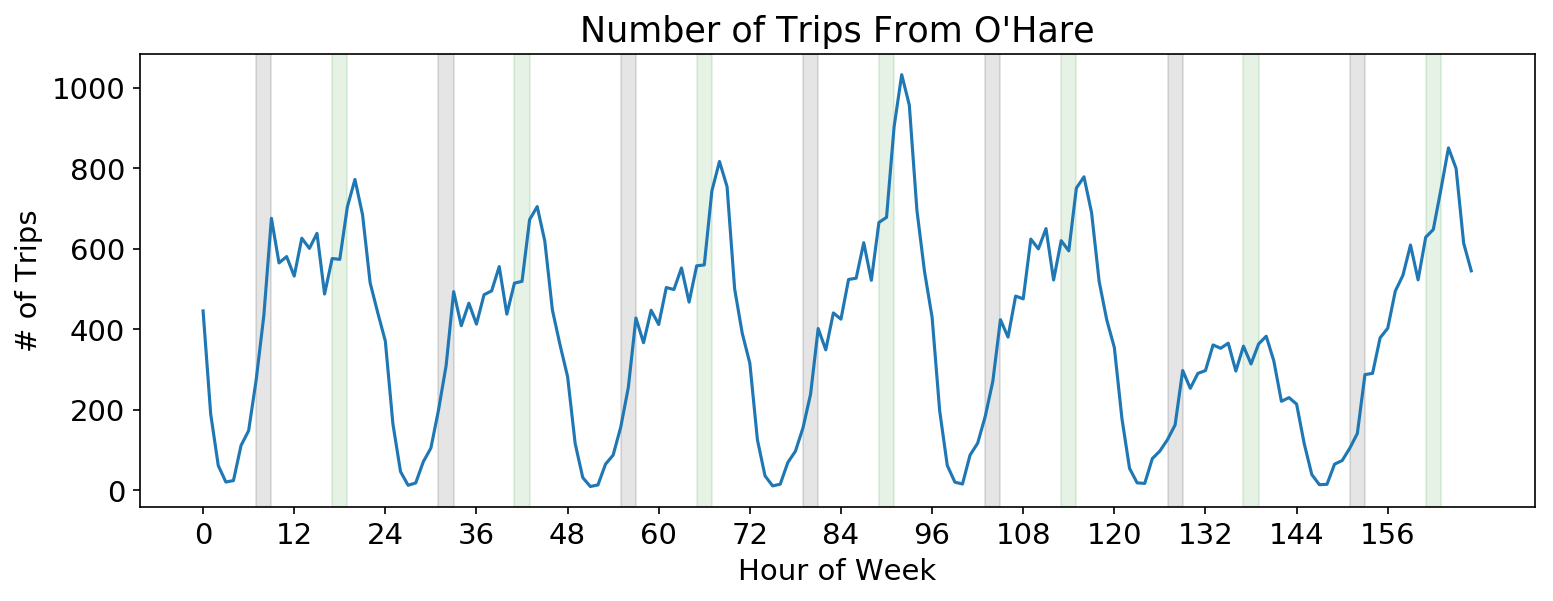}
	\caption{Average number of trips from the O'Hare International Airport, by hour-of-week. The gray stripes indicate the morning rush hours (7am - 9am) and the green stripes indicate the evening rush hours (5pm - 7pm).
	\label{fig:ohare_how_trips_from_airport}}
\end{figure}

\medskip

Figure~\ref{fig:heatmaps_from_ohare_2} illustrates the average duration and the average earning rates (trip fare divided by trip duration) for trips ending in each census tract. We can see that longer trips take more time on average, and trips ending closer to major highways have better earnings rates.

\begin{figure}[t!]
\centering
\begin{subfigure}[t]{\subfigWidthTwo \textwidth}
	\centering
    \includegraphics[height=\heatMapHeight in]{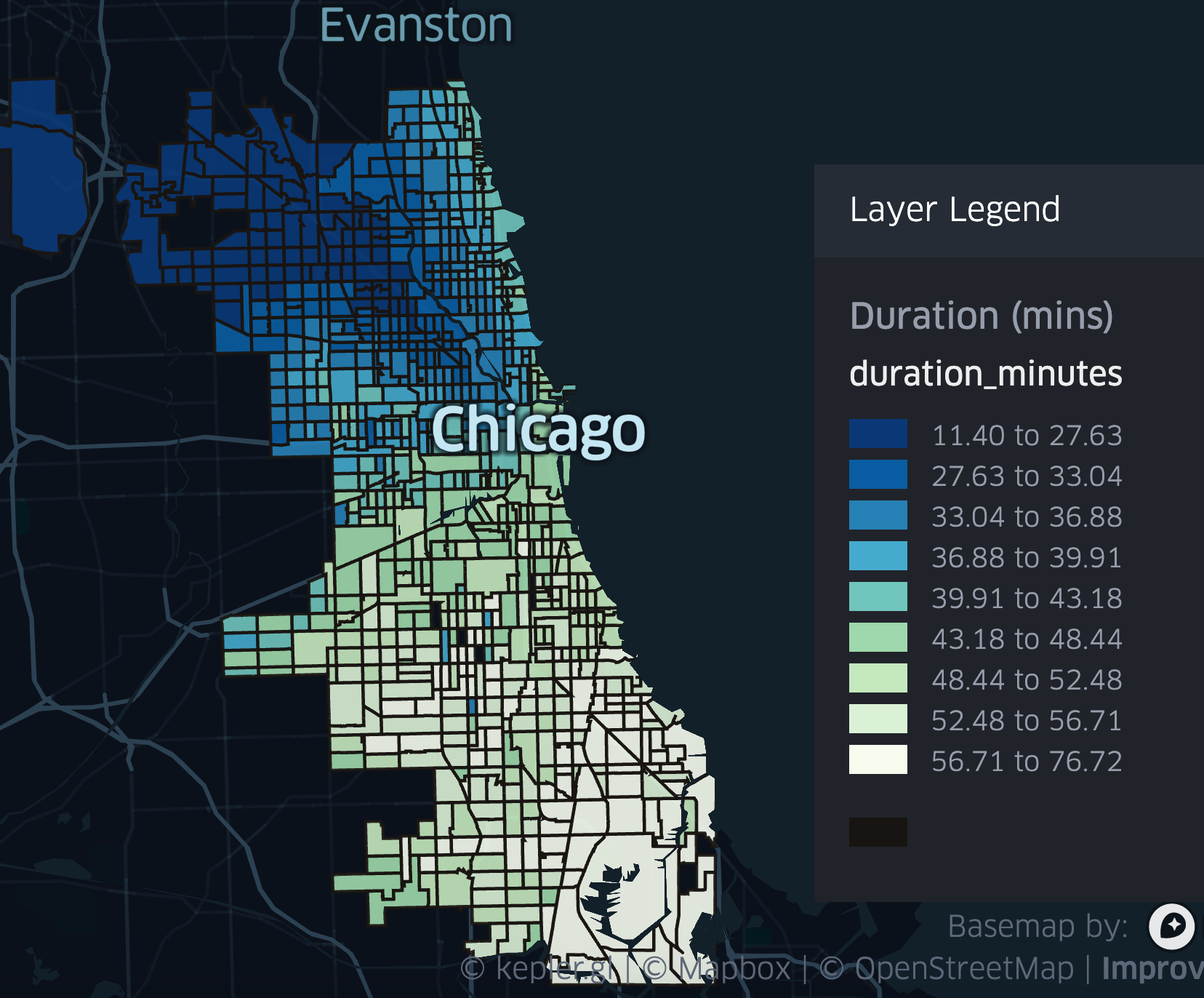}
    \caption{Average trip duration. \label{fig:heatmap_ohare_duration}}
\end{subfigure}%
\begin{subfigure}[t]{\subfigWidthTwo \textwidth}
	\centering
    \includegraphics[height=\heatMapHeight in]{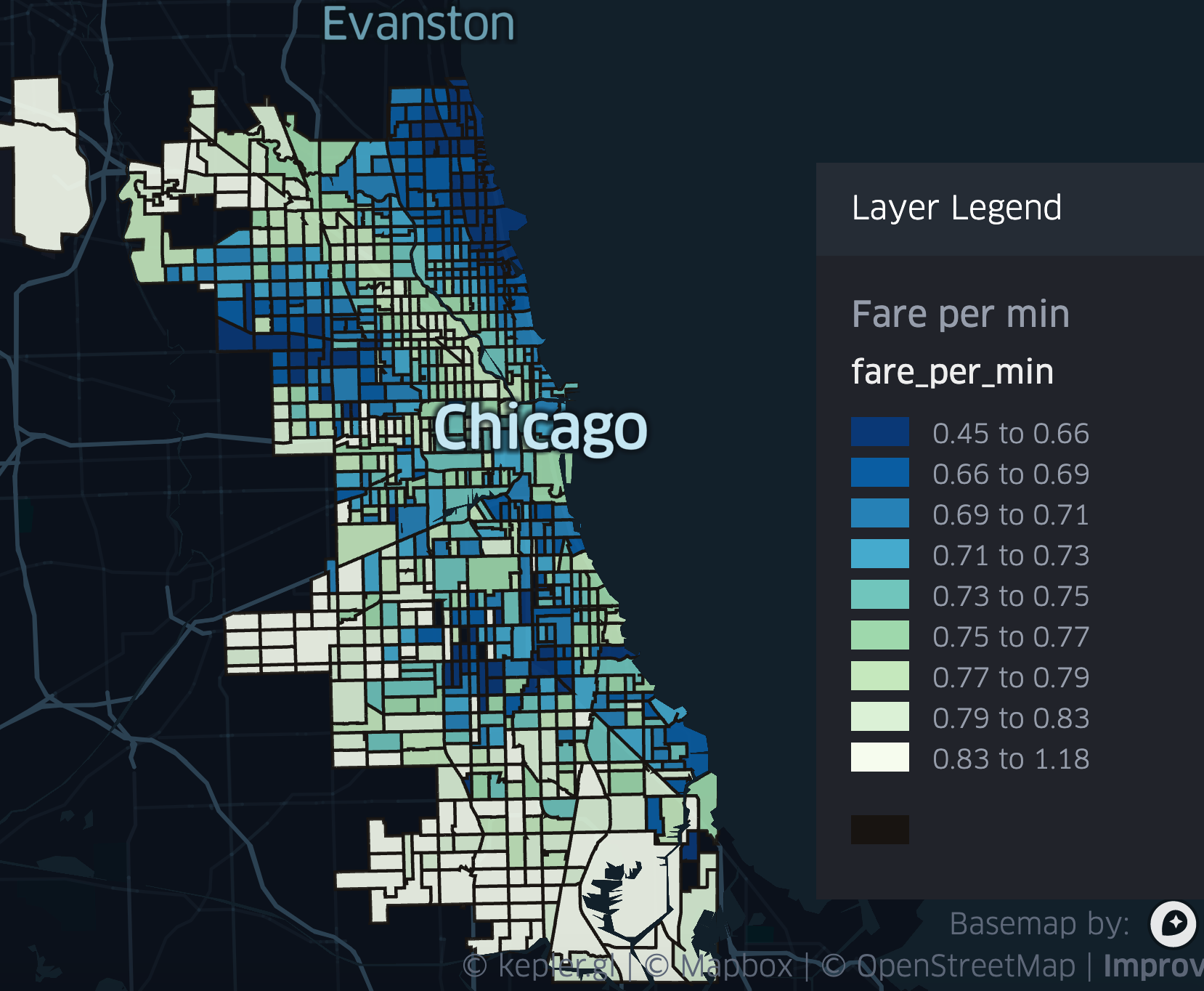}
    \caption{Average fare per minute. \label{fig:heatmap_ohare_fare_per_min}}
\end{subfigure}%
\caption{Average trip duration (in minutes) and the average fare per minute by destination Census Tract in Chicago, for trips originating from the Chicago O'Hare International Airport.
\label{fig:heatmaps_from_ohare_2}}   
\end{figure}

\subsubsection{Counterfactual Simulations}

We now provide additional results for O'Hare that are omitted from %
the body of the paper.

\paragraph{Varying Driver Supply}

As the arrival rate of driver varies, Figure~\ref{fig:varying_lambda_ohare_3} presents the minimum and maximum waiting times for drivers who joined the queue in equilibrium in steady state. For strict FIFO, and direct FIFO, the minimum waiting time is the time a driver needs to wait in the queue for the lowest earning trip that is completed in equilibrium. Under randomized FIFO, the minimum waiting time is the time it takes for a driver to move from the tail of the queue to the %
last bin (i.e. position $\binLB\supK$ in the queue). When the queue is not over-supplied, the minimum waiting times under direct FIFO and randomized FIFO are both zero. 
Under strict FIFO and direct FIFO, the maximum waiting time is the time a driver needs to wait in the queue for a trip to location 1, the highest earning trip.  
Under random dispatching or randomized FIFO with $|\loc\1|>1$, %
there is no upper bound on how long a driver may need to wait in the queue.

\newcommand{\subfigWidth}{0.43}

\begin{figure}[t!]
\centering
\begin{subfigure}[t]{\subfigWidth \textwidth}
	\centering
    \includegraphics[height=\figHeight in]{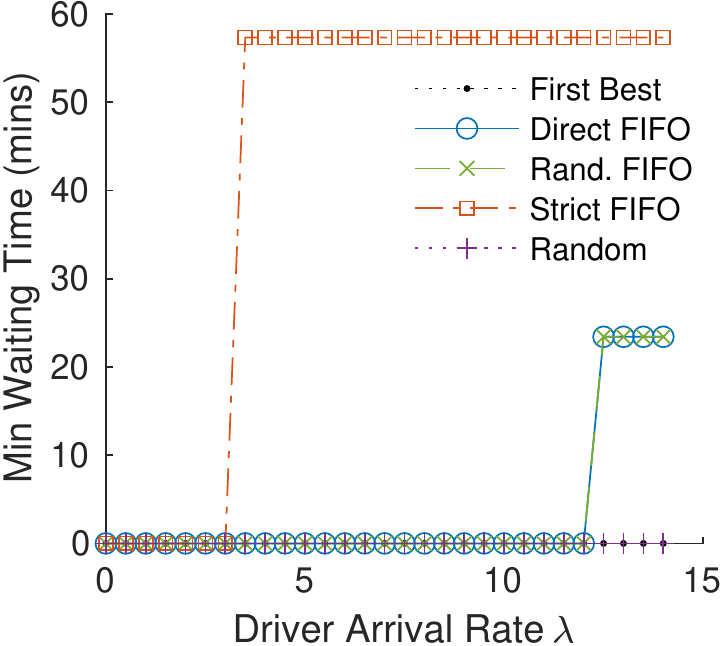}
    \caption{Minimum waiting time. \label{fig:varying_m_min_wait_ohare}}
\end{subfigure}%
\begin{subfigure}[t]{\subfigWidth \textwidth}
	\centering
    \includegraphics[height=\figHeight in]{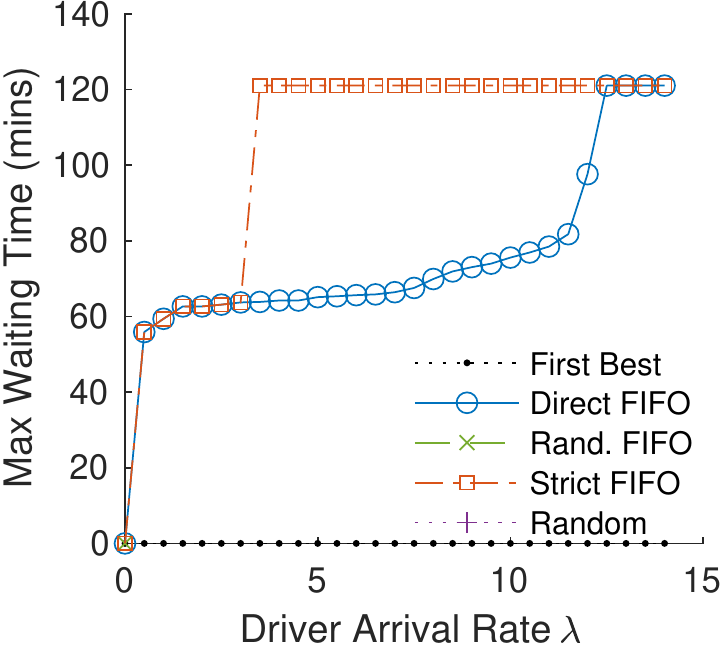}
    \caption{Maximum waiting time. \label{fig:varying_m_max_wait_ohare}}
\end{subfigure}%
\caption{The minimum and maximum waiting time for drivers who join the queue, in equilibrium in steady state, as the arrival rate of drivers varies. Chicago O'Hare. \label{fig:varying_lambda_ohare_3}}   
\end{figure}

\paragraph{Varying Rider Patience}
Figure~\ref{fig:varying_patience_ohare_3} presents the minimum and maximum waiting times for drivers who joined the queue as we vary the patience level of riders. %
Under strict FIFO, the minimum waiting time %
decreases very slowly as riders' patience level increases, despite the fact that %
the minimum waiting time for a trip under every other mechanism is zero.

\begin{figure}[t!]
\centering
\begin{subfigure}[t]{\subfigWidth \textwidth}
	\centering
    \includegraphics[height=\figHeight in]{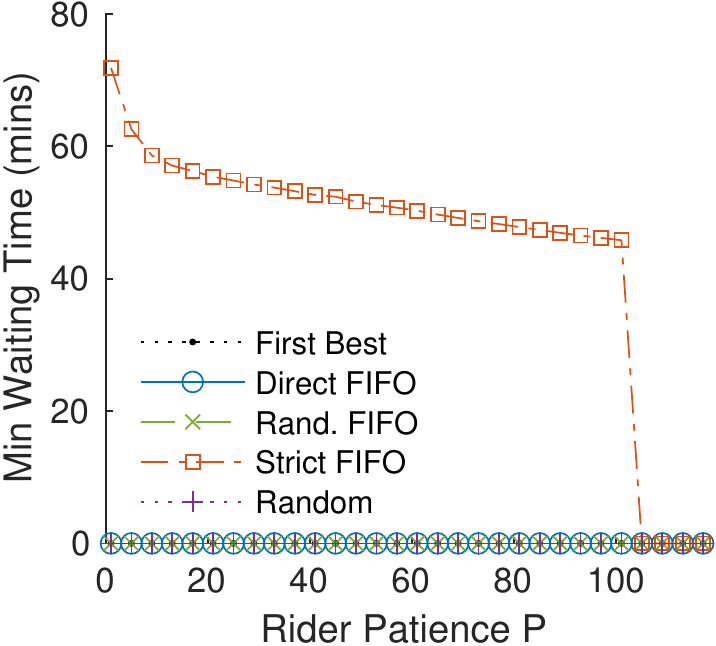}
    \caption{Minimum waiting time. \label{fig:varying_P_min_wait_ohare}}
\end{subfigure}%
\begin{subfigure}[t]{\subfigWidth \textwidth}
	\centering
    \includegraphics[height=\figHeight in]{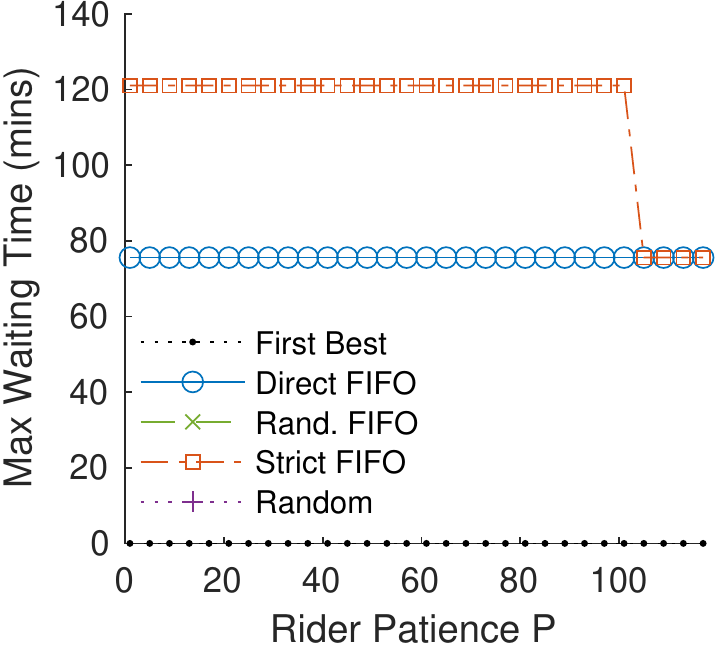}
    \caption{Maximum waiting time. \label{fig:varying_P_max_wait_ohare}}
\end{subfigure}%
\caption{The minimum and maximum waiting time for drivers who join the queue, in equilibrium in steady state, as the patience level of the riders varies. Chicago O'Hare. \label{fig:varying_patience_ohare_3}}   
\end{figure}

\subsection{Chicago Midway International Airport}

In this section, we present simulation results for %
the Chicago Midway International airport. %
Figures~\ref{fig:midway_daily_trips_from_airport} and~\ref{fig:midway_how_trips_from_airport} plot the daily number of trips originating from Midway and the average number of of trips by hour-of-week. The weekly and seasonality patterns are similar to what we observed for O'Hare.
Figure~\ref{fig:heatmaps_from_midway} shows the total trip count by destination census tract, and the estimated net earnings by destination assuming that drivers' opportunity cost is $\cost = 1/3$ per minute.

\begin{figure}[t!]
	\centering
    \includegraphics[height= \timeSeriesHeight in] {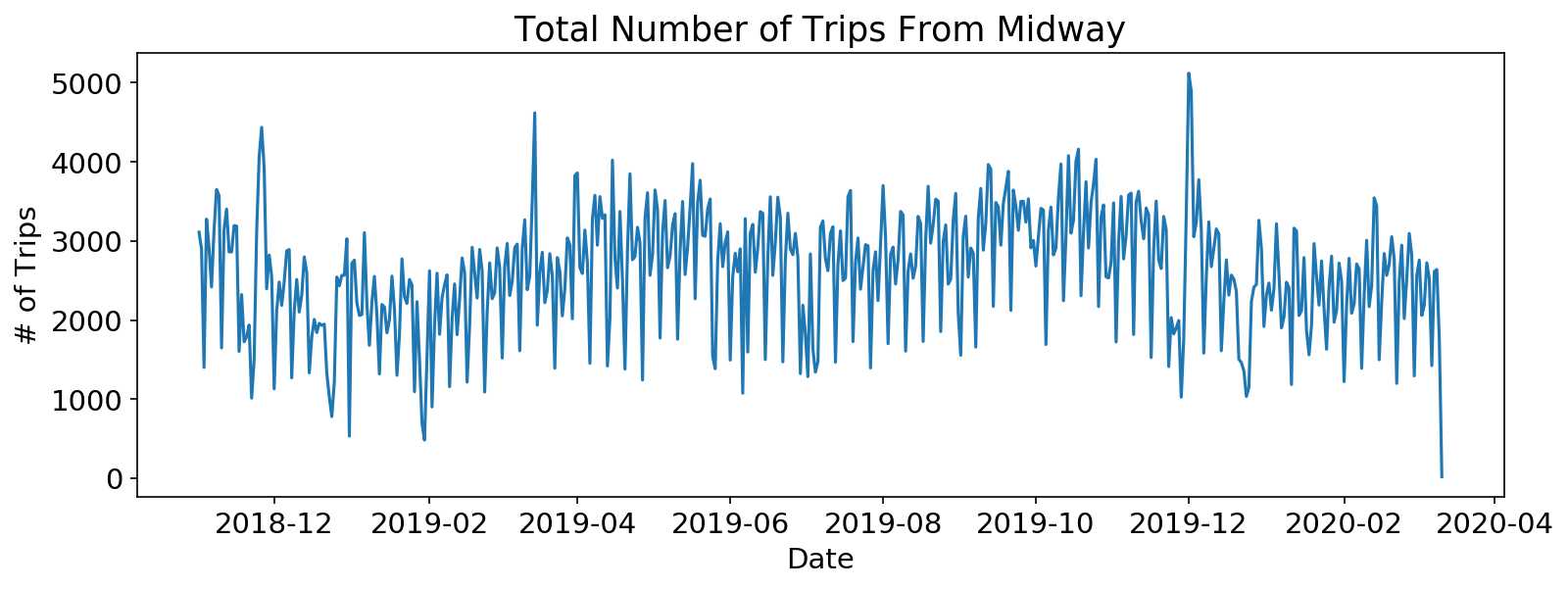}
	\caption{Total number of trips per day from the Chicago Midway International Airport.
	\label{fig:midway_daily_trips_from_airport}}
\end{figure}

\begin{figure}[t!]
	\centering
    \includegraphics[height= \timeSeriesHeight in] {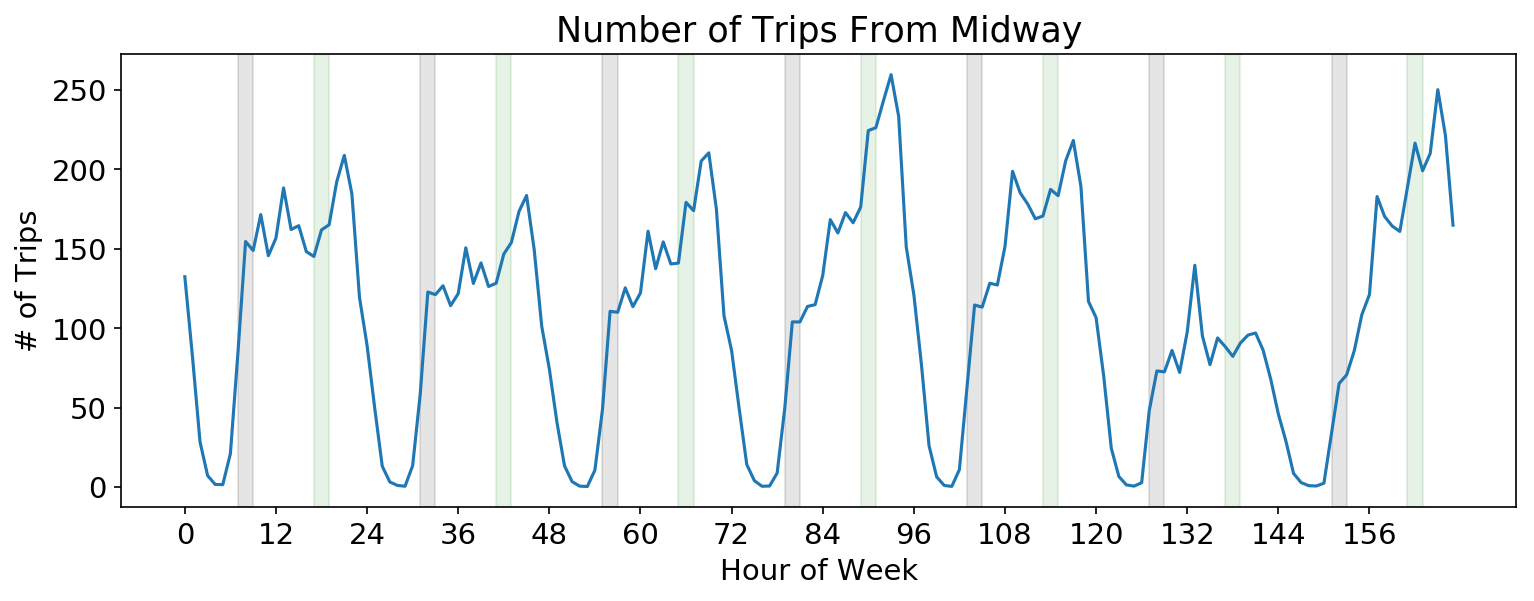}
	\caption{Average number of trips from the Midway International Airport, by hour-of-week.
	\label{fig:midway_how_trips_from_airport}}
\end{figure}

\renewcommand{\subfigWidth}{0.49}

\begin{figure}[t!]
\centering
\begin{subfigure}[t]{\subfigWidth \textwidth}
	\centering
    \includegraphics[height=\heatMapHeight in]{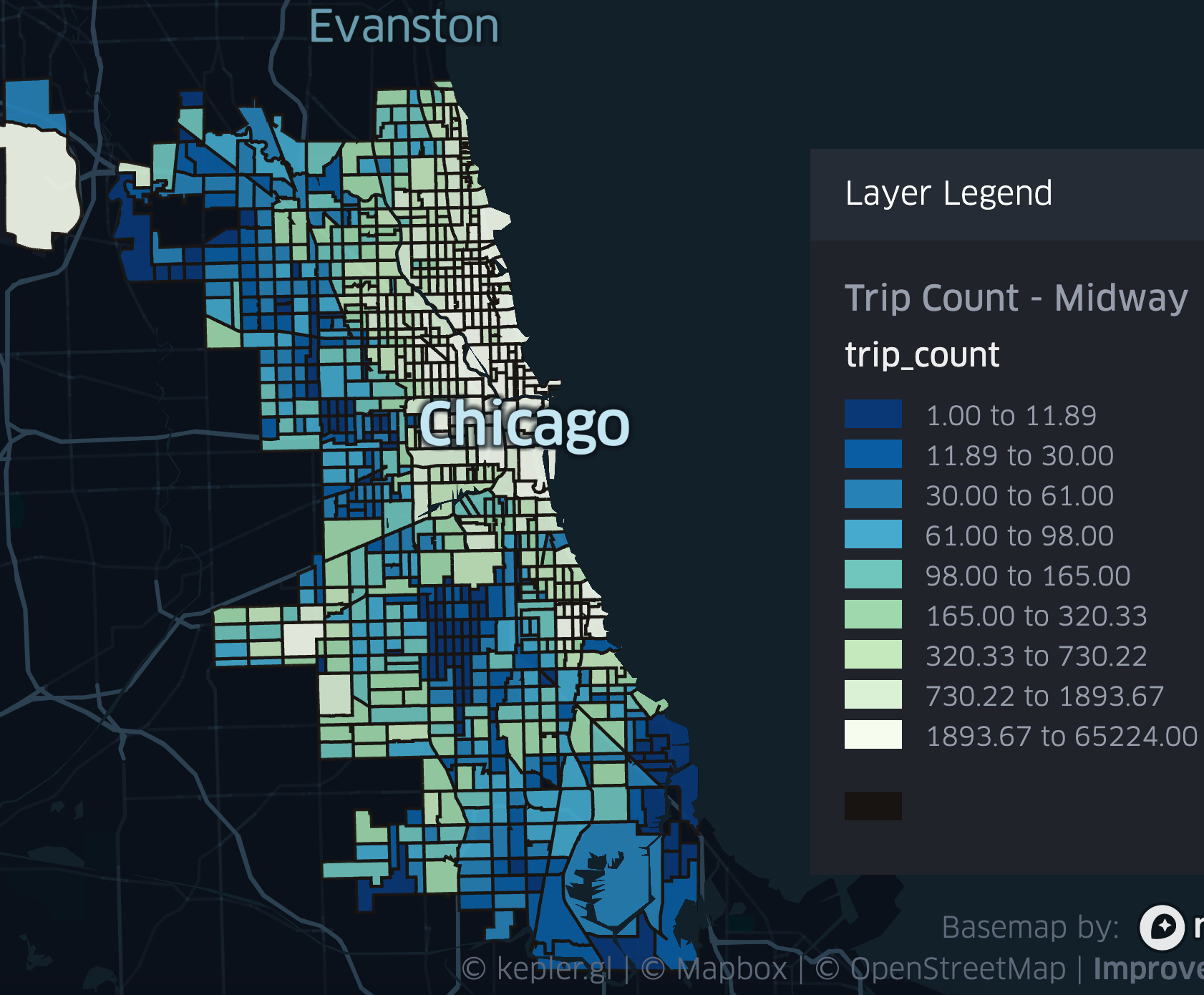}
    \caption{Trip count by destination. \label{fig:heatmap_midway_trip_count}}
\end{subfigure}%
\begin{subfigure}[t]{\subfigWidth \textwidth}
	\centering
    \includegraphics[height=\heatMapHeight in]{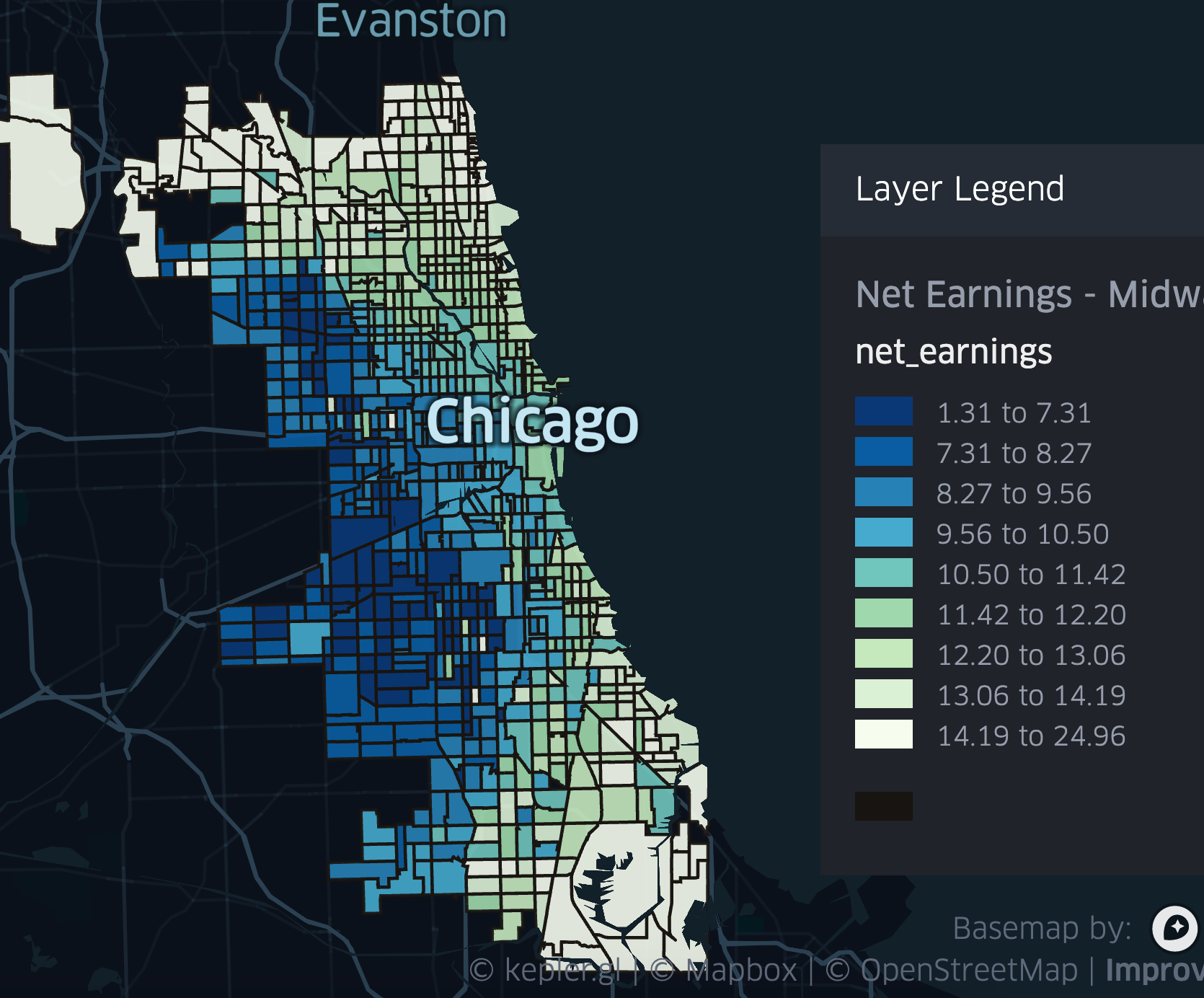}
    \caption{Net earnings by destination. \label{fig:heatmap_midway_net_earnings}}
\end{subfigure}%
\caption{Trip volume and net earnings (assuming $\cost = 1/3$) by destination Census Tract in Chicago, for trips originating from the Chicago Midway International Airport.
\label{fig:heatmaps_from_midway}}   
\end{figure}

\clearpage

\paragraph{Varying Driver Supply}

We now compare the equilibrium, steady state outcome under various mechanisms and benchmarks, as we vary the arrival rate of drivers from $0$ to $6$ drivers per minute. We fix the total arrival rate of riders at $\sum_{i \in \loc}\mu_i = 5$, and the rider patience level at $\patience = 12$. See Figures~\ref{fig:varying_lambda_midway_1}, \ref{fig:varying_lambda_midway_2} and \ref{fig:varying_lambda_midway_3}.  
The observations here are aligned with those for the O'Hare airport presented in Section~\ref{sec:simulations}.

\begin{figure}[t!]
\centering
\begin{subfigure}[t]{\subfigWidthThree \textwidth}
	\centering
    \includegraphics[height= \figHeight in]{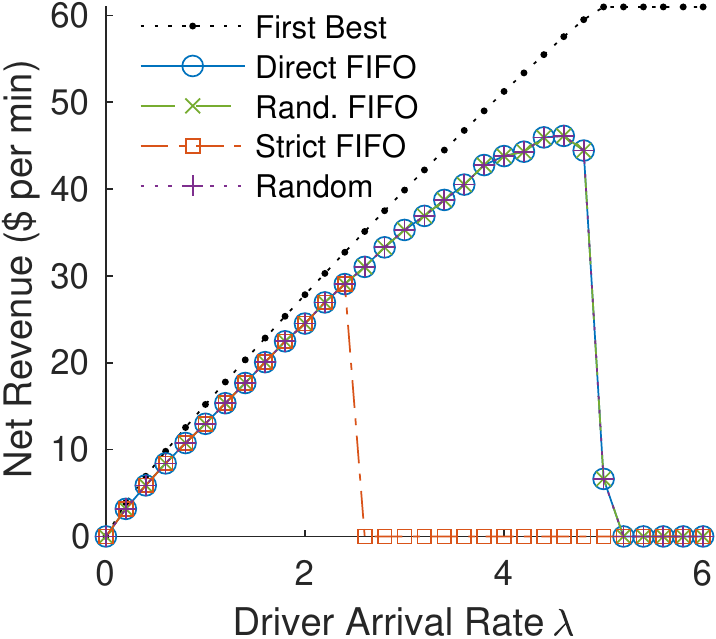}
    \caption{Net revenue.\label{fig:varying_m_gb_midway}}
\end{subfigure}%
\begin{subfigure}[t]{\subfigWidthThree \textwidth}
	\centering
    \includegraphics[height=\figHeight in]{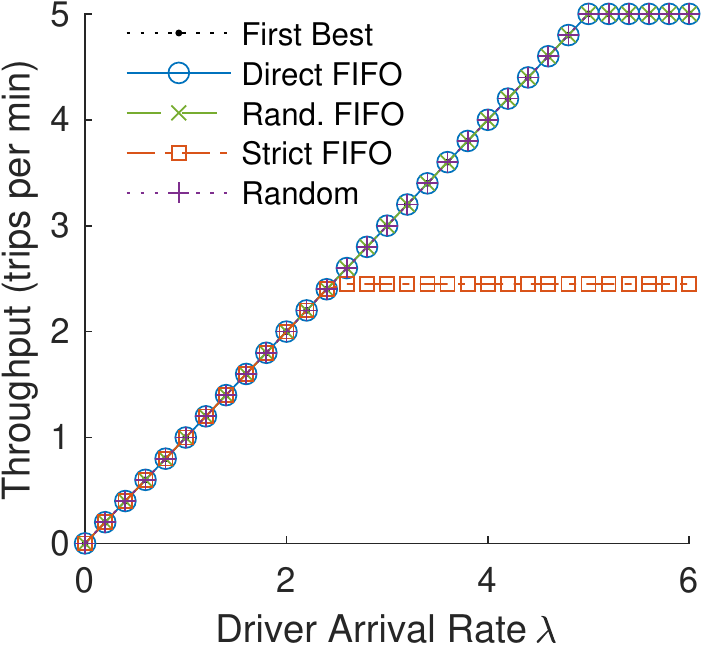}
    \caption{Trip throughput. \label{fig:varying_m_throughput_midway}}
\end{subfigure}%
\begin{subfigure}[t]{\subfigWidthThree \textwidth}
	\centering
    \includegraphics[height=\figHeight in]{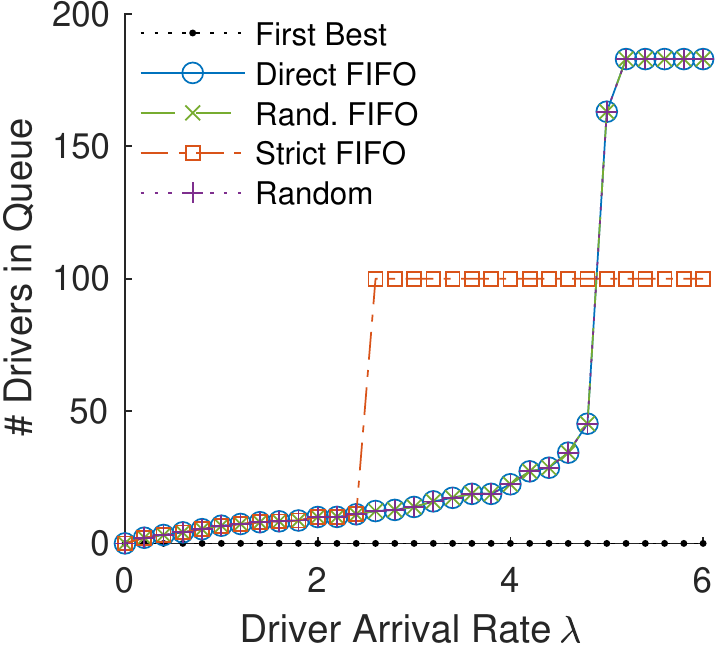}
    \caption{Equilibrium queue length. \label{fig:varying_m_max_q_length_midway}}
\end{subfigure}%
\caption{Equilibrium net revenue, trip throughput, and length of the queue in steady state, as the arrival rate of drivers varies. Chicago Midway. \label{fig:varying_lambda_midway_1}} 
\end{figure}

\begin{figure}[t!]
\centering
\begin{subfigure}[t]{\subfigWidthThree \textwidth}
	\centering
    \includegraphics[height=\figHeight in]{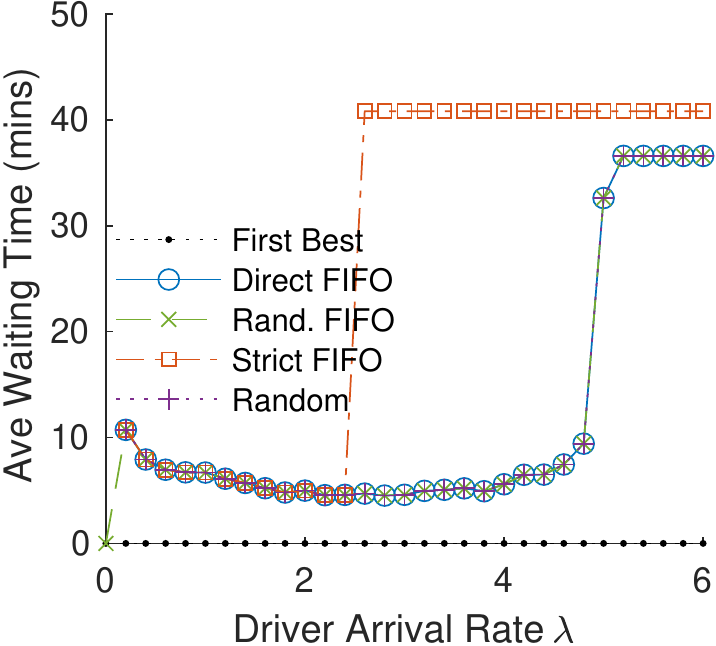}
    \caption{Average waiting time. \label{fig:varying_m_ave_wait_midway}}
\end{subfigure}%
\begin{subfigure}[t]{\subfigWidthThree \textwidth}
	\centering
    \includegraphics[height=\figHeight in]{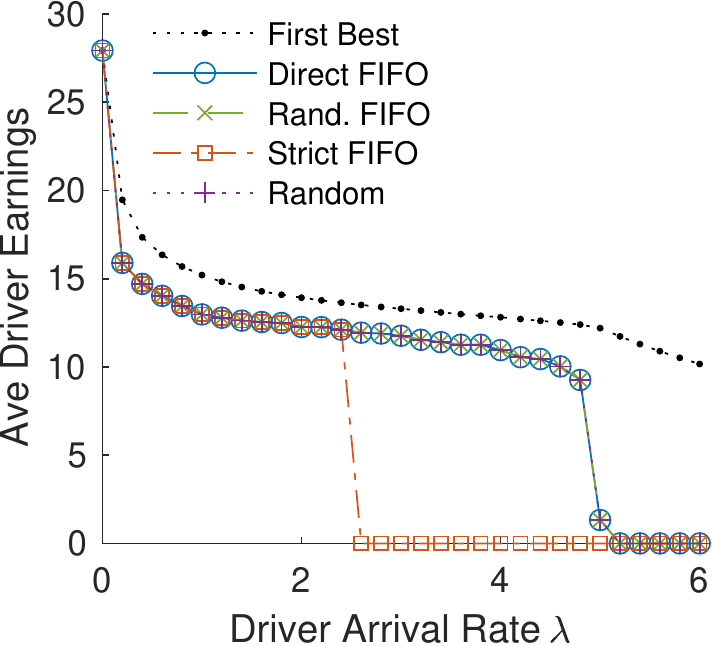}
    \caption{Average driver payoff. \label{fig:varying_m_ave_driver_earning_midway}}
\end{subfigure}%
\begin{subfigure}[t]{\subfigWidthThree \textwidth}
	\centering
    \includegraphics[height=\figHeight in]{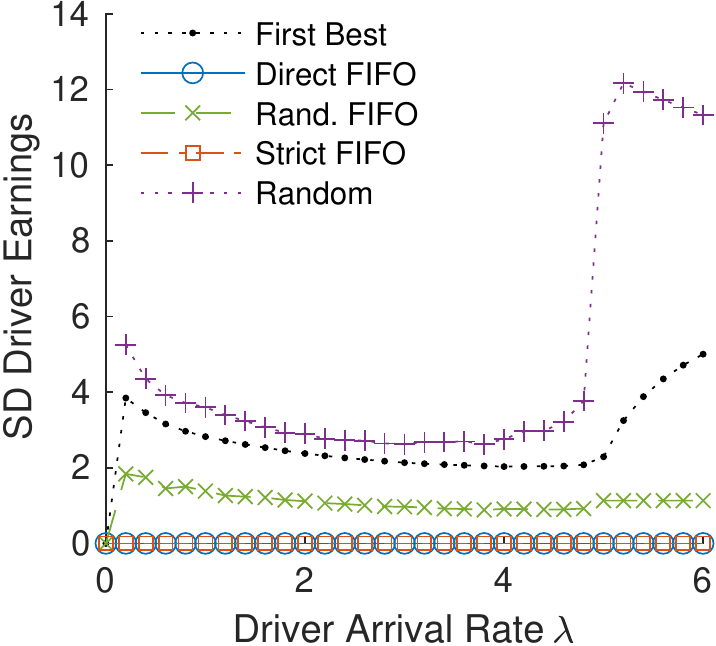}
    \caption{SD of driver payoffs. \label{fig:varying_m_sd_driver_earning_midway}}
\end{subfigure}%
\caption{Drivers' average waiting times, total payoff, and the standard deviation (SD) in drivers' total payoff in equilibrium in steady state, as the arrival rate of drivers varies. Chicago Midway. \label{fig:varying_lambda_midway_2}} 
\end{figure}

\begin{figure}[t!]
\centering
\begin{subfigure}[t]{\subfigWidth \textwidth}
	\centering
    \includegraphics[height=\figHeight in]{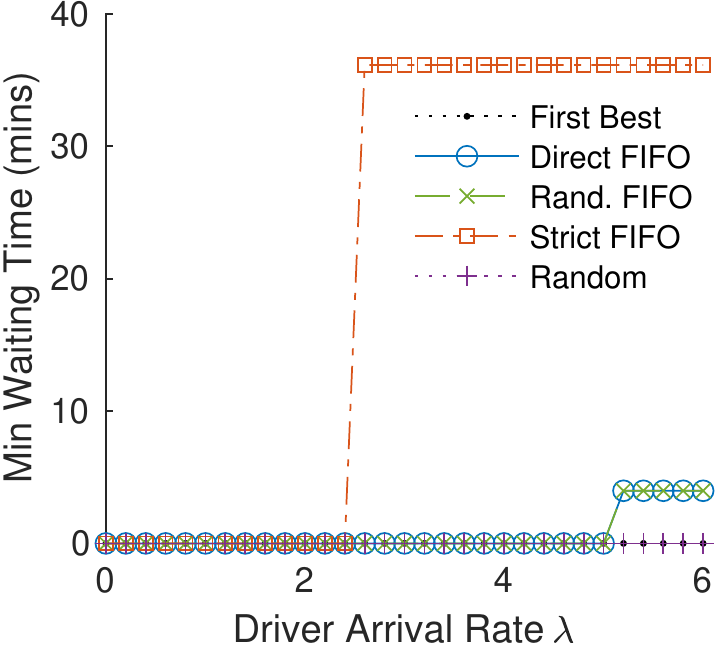}
    \caption{Minimum waiting time. \label{fig:varying_m_min_wait_midway}}
\end{subfigure}%
\begin{subfigure}[t]{\subfigWidth \textwidth}
	\centering
    \includegraphics[height=\figHeight in]{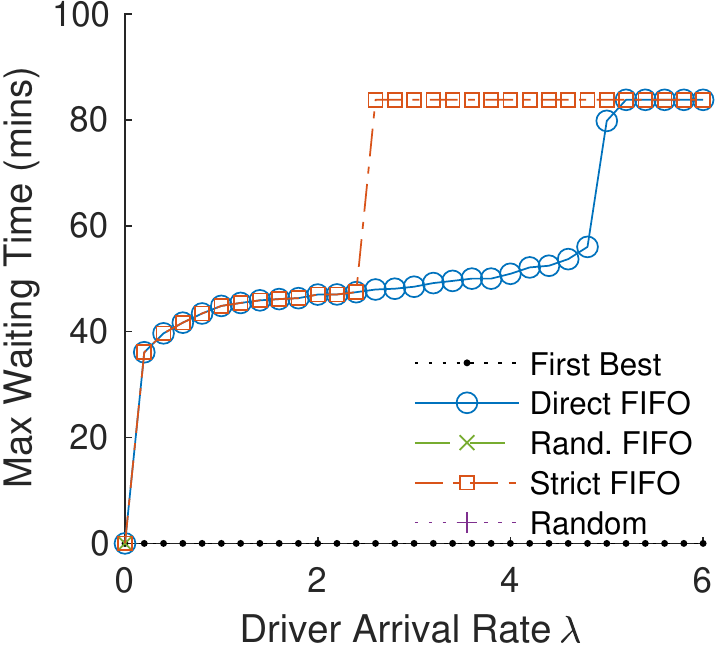}
    \caption{Maximum waiting time. \label{fig:varying_m_max_wait_midway}}
\end{subfigure}%
\caption{The minimum and maximum waiting time for drivers who join the queue, in equilibrium in steady state, as the arrival rate of drivers varies. Chicago Midway. \label{fig:varying_lambda_midway_3}}   
\end{figure}

\paragraph{Varying Rider Patience}

Fixing the arrival rate of drivers at $\lambda = 4$,  %
Figures~\ref{fig:varying_patience_midway_1}, \ref{fig:varying_patience_midway_2} and~\ref{fig:varying_patience_midway_3} compare the equilibrium outcome under various mechanisms and benchmarks as we vary the patience level of the riders.
The observations are, again, fully aligned with those for the O'Hare airport presented in Section~\ref{sec:simulations} of the paper.  

\begin{figure}[t!]
\centering
\begin{subfigure}[t]{\subfigWidthThree \textwidth}
	\centering
    \includegraphics[height= \figHeight in]{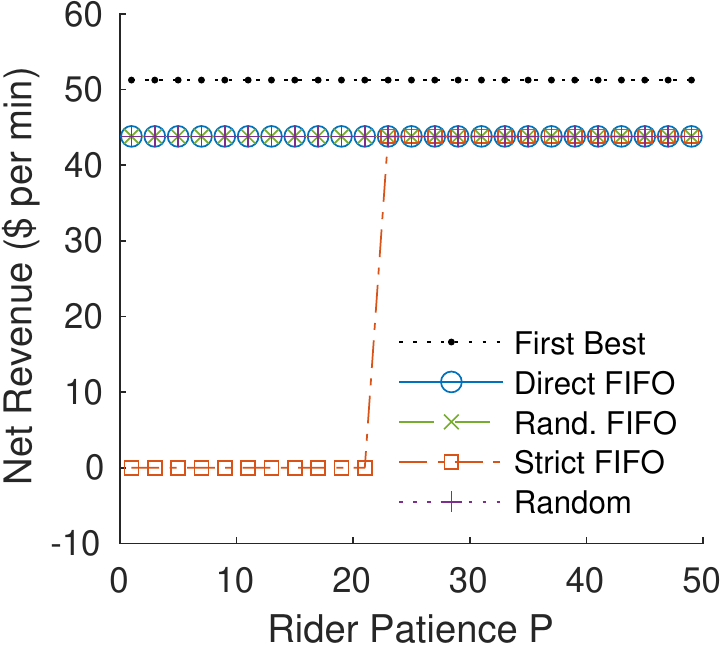}
    \caption{Net revenue.\label{fig:varying_P_gb_midway}}
\end{subfigure}%
\begin{subfigure}[t]{\subfigWidthThree \textwidth}
	\centering
    \includegraphics[height=\figHeight in]{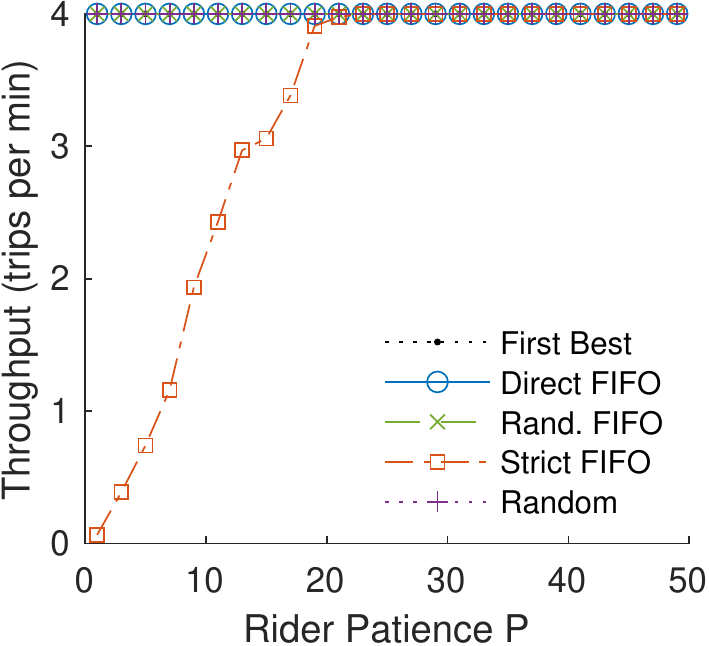}
    \caption{Trip throughput. \label{fig:varying_P_throughput_midway}}
\end{subfigure}%
\begin{subfigure}[t]{\subfigWidthThree \textwidth}
	\centering
    \includegraphics[height=\figHeight in]{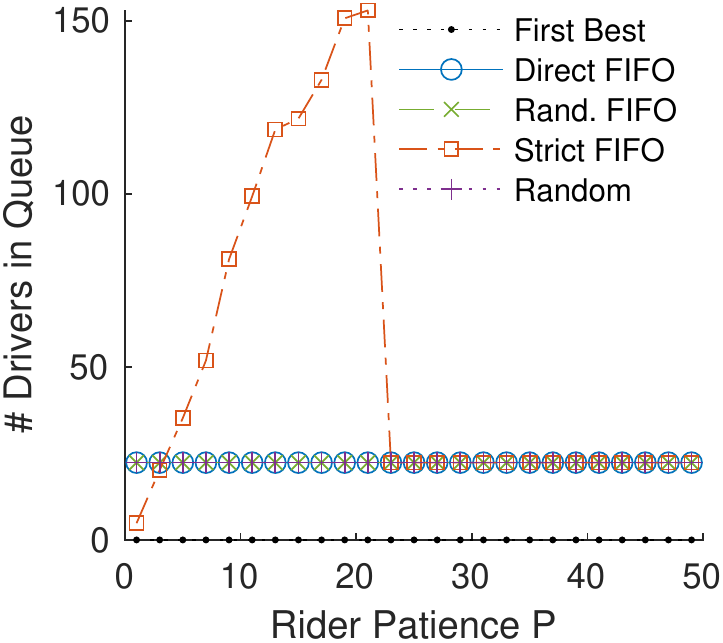}
    \caption{Equilibrium queue length. \label{fig:varying_P_max_q_length_midway}}
\end{subfigure}%
\caption{Equilibrium net revenue, trip throughput, and length of the queue in steady state, as the rider patience level varies. Chicago Midway. \label{fig:varying_patience_midway_1}} 
\end{figure}

\begin{figure}[t!]
\centering
\begin{subfigure}[t]{\subfigWidthThree \textwidth}
	\centering
    \includegraphics[height=\figHeight in]{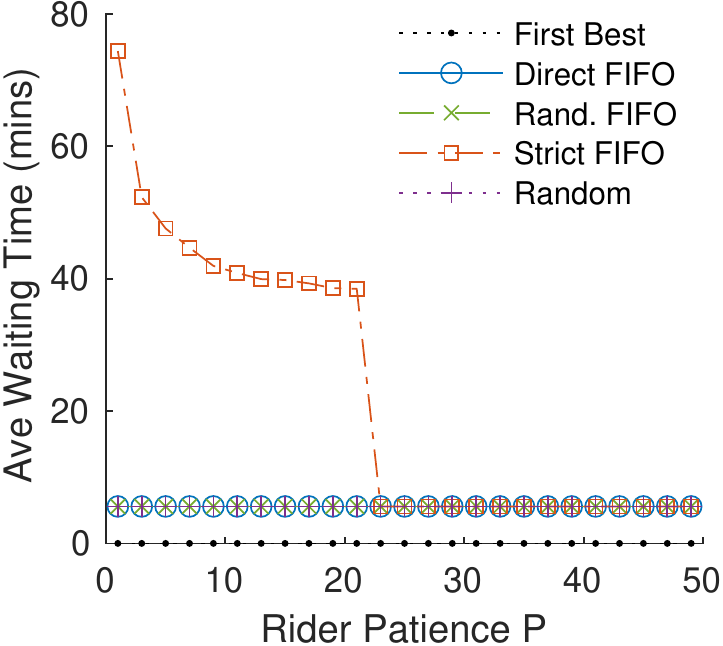}
    \caption{Average waiting time. \label{fig:varying_P_ave_wait_midway}}
\end{subfigure}%
\begin{subfigure}[t]{\subfigWidthThree \textwidth}
	\centering
    \includegraphics[height=\figHeight in]{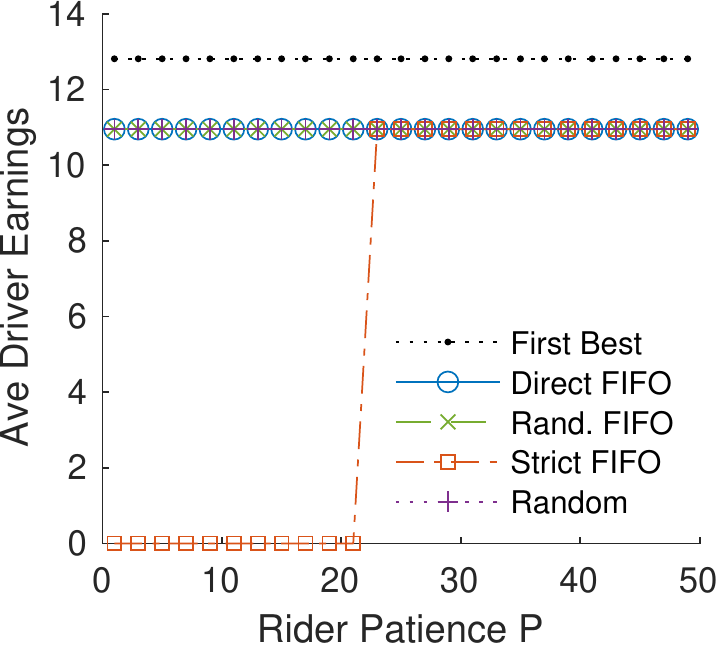}
    \caption{Average driver payoff. \label{fig:varying_P_ave_driver_earning_midway}}
\end{subfigure}%
\begin{subfigure}[t]{\subfigWidthThree \textwidth}
	\centering
    \includegraphics[height=\figHeight in]{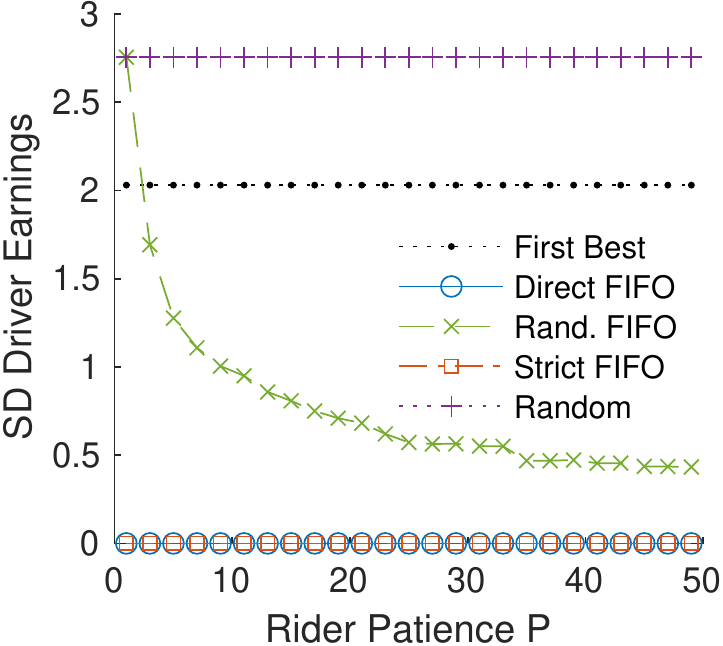}
    \caption{SD of driver payoffs. \label{fig:varying_P_sd_driver_earning_midway}}
\end{subfigure}%
\caption{Drivers' average waiting times, total payoff, and the standard deviation (SD) in drivers' total payoff in equilibrium in steady state, as the rider patience level varies. Chicago Midway. \label{fig:varying_patience_midway_2}}  %
\end{figure}

\begin{figure}[t!]
\centering
\begin{subfigure}[t]{\subfigWidth \textwidth}
	\centering
    \includegraphics[height=\figHeight in]{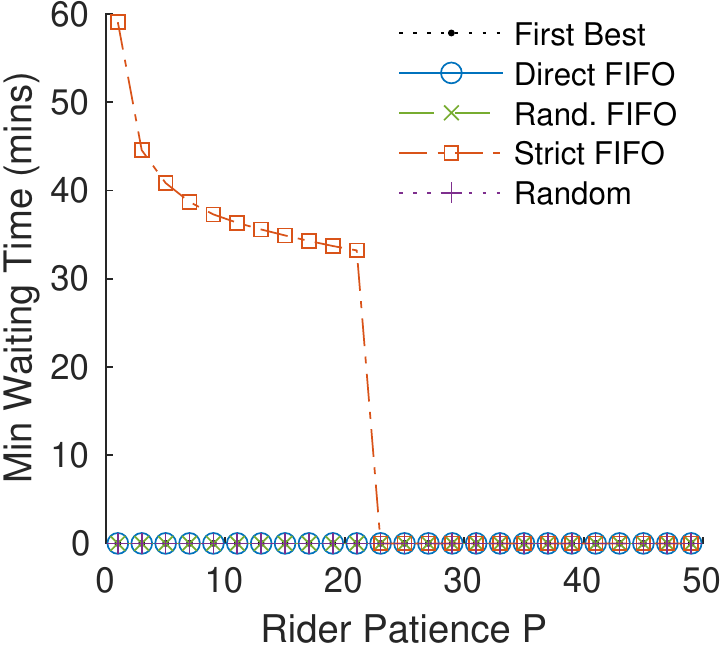}
    \caption{Minimum waiting time. \label{fig:varying_P_min_wait_midway}}
\end{subfigure}%
\begin{subfigure}[t]{\subfigWidth \textwidth}
	\centering
    \includegraphics[height=\figHeight in]{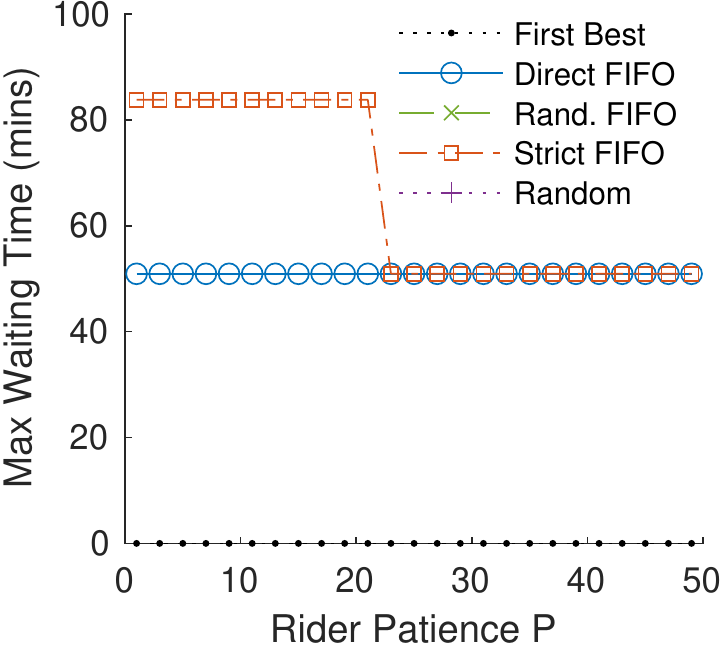}
    \caption{Maximum waiting time. \label{fig:varying_P_max_wait_midway}}
\end{subfigure}%
\caption{The minimum and maximum waiting time for drivers who join the queue, in equilibrium in steady state, as the patience level of the riders varies. Chicago Midway. \label{fig:varying_patience_midway_3}}   
\end{figure}

\end{document}